\documentclass[11pt,twoside,letterpaper]{article} 
\usepackage{fancyhdr}
\usepackage[dvips]{graphicx}
\usepackage{rotating}
\usepackage{graphics}
\usepackage{graphicx}
\usepackage{amsmath}
\usepackage{amsfonts}
\usepackage{amssymb}
\usepackage{comment}
\usepackage{color}
\usepackage{lineno}
\usepackage{epsf}
\usepackage[colorlinks=true]{hyperref}
\usepackage{enumerate}
\usepackage{makeidx}
\usepackage{amsthm}
\usepackage{times}
\usepackage{palatino}
\usepackage{charter}
\usepackage{authblk}
\usepackage[retainorgcmds]{IEEEtrantools}
\usepackage{fancyhdr}
\usepackage{fancybox}
\usepackage{mathrsfs}
\usepackage{multirow}
\usepackage{bm}

\makeatletter
\def\cleardoublepage{\clearpage\if@twoside \ifodd\c@page\else%
    \hbox{}%
    \thispagestyle{empty}%
    \newpage%
    \if@twocolumn\hbox{}\newpage\fi\fi\fi}
\makeatother

\newcommand{\ud}{\mathrm{d}}
\newcommand{\R}{\mathbb{R}}

\def\figurename{Figure}
\makeatletter
\renewcommand{\fnum@figure}[1]{\figurename~\thefigure.}
\makeatother

\def\tablename{Table}
\makeatletter
\renewcommand{\fnum@table}[1]{\bfseries\tablename~\thetable.}
\makeatother

\newtheorem{theorem}{Theorem}
\newtheorem{corollary}{Corollary}
\newtheorem{remark}{Remark}
\newtheorem{definition}{Definition}
\newtheorem{lemma}[theorem]{Lemma}
\newtheorem{proposition}{Proposition}

\newtheorem{assumption}{Assumption}

\begin{document}
\title{
{\begin{flushleft}
\vskip 0.45in
{\normalsize\bfseries\textit{Chapter~1}}
\end{flushleft}
\vskip 0.45in
\bfseries\scshape  Linear and Nonlinear Partial Integro-Differential Equations arising from Finance}}
\author{
{\bfseries\itshape 
Jos\'{e} M. T. S. Cruz\textsuperscript{1}, Maria do R. L. Grossinho\textsuperscript{1}, 
Daniel \v{S}ev\v{c}ovi\v{c}\textsuperscript{2}\thanks{Corresponding author: D. ~\v{S}ev\v{c}ovi\v{c}; email: sevcovic@fmph.uniba.sk \\ 
Acknowledgments: Support of the Slovak Research and Development Agency under the project APVV-20-0311 (C.U.) is kindly acknowledged. The research was also supported by the VEGA 1/0611/21 grant (D.\v{S}.).}, and Cyril Izuchukwu Udeani\textsuperscript{2}
}
\\
${}^1$ {ISEG, University of Lisbon, Portugal,}
\\
${}^2$ {Comenius University in Bratislava, Slovakia}
}
\date{}
\maketitle
\thispagestyle{empty}
\setcounter{page}{1}
\thispagestyle{fancy}
\fancyhead{}
\fancyhead[L]{In: Book Title \\ 
Editor: Editor Name, pp. {\thepage-\pageref{lastpage-01}}} 
\fancyhead[R]{ISBN 0000000000  \\
\copyright~2022 Nova Science Publishers, Inc.}
\fancyfoot{}
\renewcommand{\headrulewidth}{0pt}


\begin{abstract}

\noindent 
The purpose of this review paper is to present our recent results on nonlinear and nonlocal mathematical models arising from modern financial mathematics. It is based on our four papers written jointly by J. Cruz, M. Grossinho, D. \v{S}ev\v{c}ovi\v{c}, and C. Udeani \cite{NBS19}, \cite{CruzSevcovic2020}, \cite{vsevvcovivc2021multidimensional}, \cite{udeani2021application}, as well as parts of PhD thesis by J. Cruz \cite{CruzPhD2021}. We investigated linear and nonlinear partial integro-differential equations (PIDEs) arising from option pricing and portfolio selection problems and studied the systematic relationships between the PIDEs with option pricing theory and Black--Scholes models. First, we relax the liquid and complete market assumptions and extend the models that study the market's illiquidity to the case where the underlying asset price follows a L\'evy stochastic process with jumps. Then, we establish the corresponding PIDE for option pricing under suitable assumptions. The qualitative properties of solutions to nonlocal linear and nonlinear PIDE are presented using the theory of abstract semilinear parabolic equation in the scale of Bessel potential spaces. The existence and uniqueness of solutions to the PIDE for a general class of the so-called admissible L\'evy measures satisfying suitable growth conditions at infinity and origin are also established in the multidimensional space. Additionally, the qualitative properties of solutions to the generalized PIDE are investigated by considering a general shift function arising from nonlinear option pricing models, which takes into account a large trader stock-trading strategy with the underlying asset price following the L\'evy process. For the portfolio management problem, we present the existence and uniqueness results of the fully nonlinear HBJ equation arising from the stochastic dynamic optimization problem in Sobolev spaces using the theory of monotone operator technique, which can also be viewed as PIDE in some sense. Furthermore, a stable, convergent, and consistent numerical scheme that can give an approximate solution to such PIDE is presented, and various numerical experiments are conducted to illustrate the influence of a large trader and the intensity of jumps on the option price.\enlargethispage{-16pt}
\end{abstract}

\smallskip

\noindent{\bf Keywords:} L\'evy measure, Option pricing, Partial integro-differential equation, Hamilton-Jacobi-Bellman equation, Maximal monotone operator, Dynamic stochastic portfolio optimization

\hypersetup{linkcolor=black}
\tableofcontents

\section{Introduction}
\label{sec:introduction}


This review paper contain our recent advances in the research focused on nonlinear and nonlocal mathematical models arising from modern financial mathematics. The major parts of this chapter are based on our four papers jointly written by J. Cruz, M. Grossinho, D. \v{S}ev\v{c}ovi\v{c}, and C. Udeani \cite{NBS19}, \cite{CruzSevcovic2020}, \cite{vsevvcovivc2021multidimensional}, \cite{udeani2021application}, as well as parts of PhD thesis by J. Cruz. 

The classical Black--Scholes model has been widely used in financial industry because of its simplicity and existence of analytical formula for pricing derivative securities. This model relies on the restrictive assumptions, such as completeness, frictionless of the market, and the assumption that the underlying asset price follows a geometric Brownian motion. However, the assumption that an investor can trade large amount of assets without affecting the underlying asset price is usually not satisfied, especially in illiquid markets. It is also known that the fully nonlinear Hamilton--Jacobi--Bellman equation plays an essential role in finance. For instance, it gives the necessary and sufficient condition of a control with respect to the value function. Therefore, this chapter investigates linear and nonlinear partial integro-differential equations (PIDEs) arising from option pricing and portfolio selection problem. We investigate the systematic relationships of the PIDEs with option pricing theory and Black--Scholes models. First, we relax the liquid and complete market assumptions and extend the models that study market's illiquidity to the case where the underlying asset price follows a L\'evy stochastic process with jumps. Then, we establish the corresponding PIDE for option pricing under suitable assumptions. The qualitative properties of solutions to nonlocal linear and nonlinear PIDE are presented using the theory of abstract semilinear parabolic equation in the scale of Bessel potential spaces. The existence and uniqueness of solutions to the PIDE for a general class of the so-called admissible L\'evy measures satisfying suitable growth conditions at infinity and origin are also established, and the solution to the corresponding PIDE are presented in the multidimensional space. Additionally, the qualitative properties of solutions to the generalized PIDE are investigated by considering a general shift function arising from nonlinear option pricing models, which takes into account a large trader stock-trading strategy with the underlying asset price following the L\'evy process. For the portfolio management problem, we present the existence and uniqueness results to the fully nonlinear HBJ equation arising from stochastic dynamic optimization problem in Sobolev spaces using the theory of monotone operator technique, which can also be viewed as PIDE in some sense. Furthermore, a stable, convergent, and consistent numerical scheme that can give an approximate solution of such PIDE is presented, and various numerical experiments are conducted to illustrate the influence of a large trader and intensity of jumps on the option price.

The Black--Scholes model and Hamilton--Jacobi--Bellman (HJB) equation have been widely used in financial markets. However, evidence from the stock market observation shows that the Black-Scholes model is not the most realistic one because it depends on some restrictive assumptions, such as the liquidity, completeness, and frictionless of the market. Additionally, the linear Black--Scholes equation provides a solution that corrsponds to a perfectly replicated portfolio, which is not a desirable property. For this reason, several attempts have been made to generalize and relax some of these assumptions. Some authors relaxed these assumptions by considering the presence of transaction costs (see Kwok \cite{NBS5} and Avellaneda and Paras \cite{NBS7}), feedback and illiquid market effects due to large traders choosing given stock-trading strategies (Sch\"onbucher and Willmott \cite{NBS13}, Frey and Patie \cite{NBS11}, Frey and Stremme \cite{NBS10}), and risk from the unprotected portfolio (Janda\v{c}ka and \v{S}ev\v{c}ovi\v{c} \cite{NBS1}). In these generalizations, the constant volatility was replaced by a nonlinear function depending on the second derivative of the option price. Frey and Stremme derived a nonlinear Black--Scholes model that plays an essential role in the class of the generalized Black--Scholes equation with such a  nonlinear diffusion function \cite{NBS11,NBS1, Frey98}). In this model, the asset dynamics considers the presence of feedback effects due to a large trader choosing his/her stock-trading strategy \cite{NBS13}. Another important direction in generalizing the original Black--Scholes equation arises from the fact that the sample paths of a Brownian motion are continuous; however, the realized stock price of a typical company exhibits random jumps over the intraday scale, making the price trajectories discontinuous. In the classical Black--Scholes model, the logarithm of the price process has normal distribution. However, the empirical distribution of stock returns exhibits fat tails. Meanwhile, when calibrating the theoretical prices to the market prices, the implied volatility is not constant as a function of strike price nor a function of time to maturity, contradicting the prediction of the Black--Scholes model. However, the models with jumps and diffusion can solve the problems inherent to the Black--Scholes model. The jump models also plays an essential role in the option market. In the Black--Scholes model, the market is complete, implying that every payoff can exactly be replicated; meanwhile, there is no perfect hedge in jump models, making the way of options not redundant.

Market Illiquidity has been widely studied in the literature $\cite{Jarrow94, PlaSch98, SirPapa98, PhilWil2000, Frey2000}$. The first major contribution was made by Robert Jarrow, in 1994, who studied the market manipulation strategies that may arise in illiquid markets. The author also studied the option pricing theory in discrete time when there is a large trader. The pricing argument used was a condition to ensure that no market manipulation strategy is used by the large trader and large trader's optimality conditions; thus, replacing the usual free-arbitrage argument. Then, Frey (1998) extended Jarrow's analysis to the continuous time case and established the existence and uniqueness of solution of a nonlinear partial differential equation satisfied by the large trader's hedging strategy. Additionally, Platten and Schweizer (1998) proposed an explanation for the smile and skewness for the implied volatilities and showed that hedging strategies followed by large traders can lead to option price bias. Sircar and Papanicolaou (1998) also proposed a model where the derivative security price is characterized by a nonlinear partial differential equation that becomes the Black--Scholes equation when there is no feedback. When the programme traders are a small fraction of the economy, numerical and analytical methods can be used to analyze the nonlinear partial differential equation through perturbation. This equation is derived using an argument similar to the one used in deriving classical Black--Scholes equation. Consequently, they obtained that this model also predict increased implied volatilities as in Platten and Schweizer. Furthermore, Schonbucher and Willmott (2000) analyzed the feedback effects from the presence of hedging strategies. They also derived a nonlinear partial differential equation for an option replication strategy and studied these effects for a put option. The effects are more pronounced in markets with low liquidity, which can induce discontinuities in the price process. However, none of these studies that investigated jump models \cite{Merton76, DPE98, BARNIE01, PD99, NuaSch01, RamCont2005, Eber98} considered in the market's illiquidity. Meanwhile, investors and risk managers have realized that financial models based on the assumption that an investor can trade large amounts of an asset without affecting its price is no longer true in markets that are not liquid. Therefore, in this chapter, we relax the liquid and complete market hypothesis and extend the models that study market's illiquidity to the case where the underlying asset price follows a L\'evy stochastic process with jumps to obtain a model for pricing European and American call and put options on an underlying asset characterized by a L\'evy measure. In this way, it is assumed that trading strategies affects the stock price, and the possibility to account for sudden jumps that might occur when the market is under stress \enlargethispage{-16pt}.

\pagestyle{fancy} 
\fancyhead{} 
\fancyhead[EC]{J. Cruz, M. Grossinho, D. \v{S}ev\v{c}ovi\v{c}, C. I. Udeani} 
\fancyhead[EL,OR]{\thepage}
\fancyhead[OC]{ Linear and Nonlinear PIDE arising from Option Pricing} 
\fancyfoot{}
\renewcommand\headrulewidth{0.5pt}
Recently, the relationships between more general nonlocal operators and jump processes have been widely investigated. For instance, there is an actual connection between the solution to PIDE and properties of the corresponding Markov jump process (c.f.,  Abels and Kassmann \cite{AbelsKass2009}; Florescu and Mariani \cite{florescu2010solutions}). In the past decades, the role of PIDEs has been investigated in various fields, such as pure mathematical areas, biological sciences, and economics \cite{aboodh2016solution, NBS15, yuzbacsi2016improved}. PIDE problems arising from financial mathematics, especially from option pricing models, have been of great interest to many researchers. In most cases, standard methods for solving these problems lead to the study of parabolic equations. Mikulevi\v{c}ius and Pragaraustas \cite{mikulevivcius1992cauchy} investigated solutions of the Cauchy problem to the parabolic PIDE with variable coefficients in Sobolev spaces. They employed their results to obtain solutions of the corresponding martingale problem. Crandal \emph{et al.} \cite{ishii1996viscosity} employed the notion of a viscosity solution to investigate the qualitative results. Soner \emph{et al.} \cite{burzoni2020viscosity} and Barles \emph{et al.} \cite{barles1997backward} extended and generalized their results for the first and second order operators, respectively. Florescu and Mariani \cite{florescu2010solutions} employed the Schaefer fixed point argument to establish existence of a weak solution of the generalized PIDE. Amster \emph{et al.} \cite{Amster12} used the notion of upper and lower solutions to obtain the solution to such PIDEs. They proved the existence of solutions in a general domain for multiple assets and the  regime switching jump-diffusion model. Cont \emph{et al.} \cite{cont2005integro} investigated the actual connection between option pricing in exponential L\'evy models and the corresponding PIDEs for European options and those with single or double barriers. They discussed and established the conditions for which prices of option are classical solution of the corresponding PIDE. In this chapter, we obtain a certain partial integro-differential equation (PIDE) for option pricing in illiquid market by assuming a certain dynamics for the stock's price. The existence of solution and localization results of the associated PIDE are also established. We investigated and established the qualitative properties of solutions to the nonlocal linear and nonlinear PIDE in the scale of Bessel potential spaces using the theory of abstract semilinear parabolic equation. Furthermore, we present the existence and uniqueness results for nonlinear parabolic equations using monotone operator technique, Fourier transform, and Banach fixed point argument. We considered the fully nonlinear HJB equation arising portfolio selection problem, where the goal of an investor is to optimize the condition expected value of the terminal utility of the portfolio. Such nonlinear parabolic equation is presented in an abstract setting, which can also be viewed as a nonlinear PIDE. Many previous studies have developed numerical methods for PIDEs, such as finite difference and finite element methods. However, the equation corresponding to the case of illiquid markets is more difficult. Therefore, this chapter also presents a stable, convergent, and consistent numerical scheme that can give an approximate solution of such PIDEs. Various numerical experiments are presented to illustrate the influence of a large trader and intensity of jumps on the option price.

\section{Background and Motivation} 

Based on the classical theory developed by  Black,  Scholes, and Merton, the price $V(t,S)$ of an option in a stylized market at time $t\in[0,T]$ and the underlying asset price $S$ can be calculated as a solution to the following linear Black--Scholes parabolic equation:
\begin{equation}
\frac{\partial V}{\partial t}(t,S) + \frac{1}{2}\sigma^2 S^2\frac{\partial^2 V}{\partial S^2}(t,S) + r S\frac{\partial V}{\partial S}(t,S) - r V(t,S) = 0,\mbox{ }t \in [0,T),S > 0.
\label{eqBS}
\end{equation}
Here, $\sigma>0$ is the historical volatility of the underlying asset driven by the geometric Brownian motion, and $r > 0$ is the risk-free interest rate of zero-coupon bond. The solution to above equation is subject to the terminal payoff condition $V(T,S) = \Phi(S)$ at maturity $t = T$. Meanwhile, evidence from stock markets observations indicates that the model is not the most realistic one because it assumes that the market is liquid, complete, frictionless and without transaction costs. It is also known that the linear Black--Scholes equation provides a solution corresponding to a perfectly replicated  portfolio, which need not be a desirable property. To solve this issues, several attempts have been made to generalized the linear Black--Scholes equation (\ref{eqBS}) by replacing the constant volatility $\sigma$ with a nonlinear function $\tilde\sigma(S \partial _S^2 V)$ depending on the second derivative $\partial _S^2 V$ of the option price. In this regard, Frey and Stremme derived a nonlinear Black--Scholes model, which plays an essential role in the class of generalized Black--Scholes equation with such a  nonlinear diffusion function \cite{NBS11, Frey98}). They considered that case in which the asset dynamics takes into account the presence of feedback effects due to a large trader choosing his/her stock-trading strategy (see also \cite{NBS13}). The diffusion coefficient is non-constant, and it is given by
\begin{equation}
\tilde\sigma(S \partial^2_S V)^2 = \sigma^2 \left(1-\varrho S\partial^2_S V\right)^{-2},
\label{doplnky-c-frey}
\end{equation}
where $\sigma$ and $\varrho>0$ are constants. 

Furthermore, several researchers have attempted to generalized the original Black--Scholes equation arises from the fact that the sample paths of a Brownian motion are continuous. However, the realized stock price of a typical company exhibits random jumps over the intraday scale, making the price trajectories discontinuous. The underlying asset price process is usually assumed to follow a geometric Brownian motion in the classical Black--Scholes model. However, the empirical distribution of stock returns exhibits fat tails. The models with jumps and diffusion can solve the problems inherent to the linear Black--Scholes model and play an essential role in options pricing. It is well-known that the market is complete in the Black--Scholes model, illustrating that each payoff can be perfectly replicated; however, there is no perfect hedge in jump--diffusion models, making the options not redundant. It turns out that the option price can be computed from the solution $V(t,S)$ to the following PIDE Black--Scholes equation \cite{NBS19}:
\begin{eqnarray}
&&\frac{\partial V}{\partial t}(t,S)+\frac{1}{2}\sigma^2 S^2 \frac{\partial^2 V}{\partial S^2 }(t,S) +r S\frac{\partial V}{\partial S}(t,S)-r V(t,S)\nonumber
\\
&& +\int_{\mathbb{R}} V(t,S+H(z,S))-V(t,S)-H(z,S)\frac{\partial V}{\partial S}(t,S)\nu(\ud z)=0,\label{nonlinearPIDE_one-intro}
\end{eqnarray}
where $H(z,S) = S (e^z-1)$, and $\nu$ is the so-called L\'evy measure characterizing the underlying asset process with random jumps in time and space. It is worth noting that (\ref{nonlinearPIDE_one-intro}) reduces to the classical linear Black--Scholes equation (\ref{eqBS}) if  $\nu=0$. 

In this chapter, we consider both directions of generalizations of the Black--Scholes equation. First, we relaxed the assumption of liquid market following the Frey--Stremme model by assuming that the underlying asset price follows a L\'evy stochastic process with jumps and established the corresponding PIDEs. Then, we present the existence and uniqueness results to the linear and nonlinear nonlocal PIDE in the framework of Bessel potential spaces for the multidimensional case. A more generalized nonlinear nonlocal PIDE is also presented by considering a shift function $\xi=\xi(\tau, x, z)$ depending on the variables $x,z\in\R^n$. Furthermore, we derive, analyze, and perform numerical computation of the model. We also show that the corresponding PIDE nonlinear equation has the following form: 
\begin{eqnarray}
&&\frac{\partial V}{\partial t}+\frac{1}{2}\frac{\sigma^2}{\left(1-\varrho S\partial_S \phi\right)^{2}} S^2 \frac{\partial^2 V}{\partial S^2 } +r S\frac{\partial V}{\partial S}-rV\nonumber
\\
&&+\int_{\mathbb{R}} V(t,S+H(t,z,S))-V(t,S)-H(t,z,S)\frac{\partial V}{\partial S}\nu(\ud z)=0.\label{nonlinPIDE}
\end{eqnarray}
It is worth noting that the function $H(t,z,S)$ may depend on the large trader strategy function $\phi=\phi(t,S)$ and the delta $\partial_S V$ of the price $V$ if $\varrho>0$. 


We consider a stylized economy with two traded assets: a riskless asset (a bond with a price $B_{t}$ taken as numeraire) and a risky asset (stock with a price $S_{t}$). Here, we assume that the bond market is perfectly elastic since it is more liquid than stocks. Here, we consider two type of traders: reference and program traders. The program traders are also known as portfolio insurers because they use dynamic hedging strategies to hedge portfolio against jumps in stock prices. They are classified as single traders or group of traders acting together. It is assumed that their trades influence the stock price equilibrium. In contrast, the reference traders can be considered as representative traders of many small agents. We assume that they act as price takers. Generally, $\tilde{D}(t,Y_{t},S_{t})$ represents the reference trader demand function that depends on the income process $Y_t$ or some other fundamental state variable influencing the reference trader demand. The aggregate demand of program traders is denoted by $\varphi(t,S_{t})=\xi \phi(t,S_{t})$, where $\xi$ is the number of written identical securities that the program traders are trying to hedge, and $\phi(t,S_{t})$ is the demand per unit of the security being hedged. For simplicity, we assume that $\xi$ is the same for every program trader. For a more general case where different securities are considered, see \cite{SirPapa98}. Suppose the supply of a stock with the price  $\tilde{S_{0}}$ is constant, and let $D(t,Y,S)=\frac{\tilde{D}(t,Y,S)}{\tilde{S_{0}}}$ be the quantity demanded by a reference trader per unit supply. Then, the total demand relative to the supply at time $t$ is given by $G(t,Y,S)=D(t,Y,S)+\rho\phi(t,S)$, where $\rho=\frac{\xi}{\tilde{S_{0}}}$, and $\rho\phi(t,S)$ is the proportion of the total supply of the stock traded by program traders. Therefore, to obtain the market equilibrium, the variables $Y$ and $S$ should satisfy $G(t,Y,S)=1$. Assume that the function $G$ is monotone  with respect to the variables $Y$ and $S$, and it is sufficiently smooth. Then, we can solve the implicit equation $G(t,Y_t,S_t)=1$ to obtain $S_{t}=\psi(t,Y_{t})$, where $\psi$ is a sufficiently smooth function. By employing the approach in \cite{SirPapa98}, we assume that the stochastic process $Y_{t}$ has the following dynamics: 
\[
\ud Y_{t}=\mu(t,Y_{t})\ud t+\eta(t,Y_{t}) \ud W_{t}.
\]
Then, using It\^{o}'s lemma for the process $S_{t}=\psi(t,Y_{t})$, we have
\begin{equation}
\ud S_t = \left( \partial_t \psi + \mu\partial_y\psi +\frac{\eta^2}{2}\partial^2_y\psi \right) \ud t + \eta \partial_y\psi  \ud W_t \equiv b(t,S_t) S_t \ud t + v(t,S_t) S_t \ud W_t.
\label{BS-v}
\end{equation}
It means that $S_t$ follows a geometric Brownian motion with a nonconstant volatility function $v(t,S)=\eta(t,Y)\partial_Y\psi(t,Y)/\psi(t,Y)$, where $Y=\psi^{-1}(t,S)$. 
 Thus, we follow the argument used in the derivation of the original Black--Scholes equation to obtain a generalization of the Black--Scholes partial differential equation with a nonconstant volatility function $\sigma=v(t,S)$.
We employ the Frey and Stremme's approach  (cf. \cite{NBS11, Frey98}) to prescribe a dynamics for the underlying stock price instead of deriving it using the market equilibrium and dynamics for the income process $Y_t$ as it is done in \cite{SirPapa98}. In this way, Frey and Stremme derived the same stock price dynamics as in \cite{SirPapa98} corresponding to a situation where the demand function is of logarithmic type, $D(Y,S)=\ln(\frac{Y^{\gamma}}{S})$, where $\gamma=\frac{\sigma}{\eta_0}$, and the income process $Y_t$ follows a geometric Brownian motion, i.e.,
\begin{eqnarray}
&&\partial_Y D(Y,S)=\gamma\frac{1}{Y},\ \partial_S D (Y,S)=-\frac{1}{S},\ \  \ud Y_{t}=\mu_0 Y_{t}\ud t+\eta_0 Y_{t} \ud W_{t},
\label{generalPDE-1}
\\
&&v(t,S)=\eta(t,Y)\frac{\partial_Y\psi(t,Y)}{\psi(t,Y)} = - \frac{\eta_0 Y}{S}\frac{\gamma\frac{1}{Y}}{-\frac{1}{S}+\rho \frac{\partial \phi}{\partial S}}=\frac{\sigma}{1-\rho S\frac{\partial \phi}{\partial S}}.
\nonumber
\end{eqnarray}
Assuming the delta hedging strategy with $\phi(t,S)=\partial_S V(t,S)$ and substituting the volatility function $v(t,S)$ in (\ref{BS-v}), we obtain the generalized Black--Scholes equation with the nonlinear diffusion function of the form (\ref{doplnky-c-frey}).

In this chapter, we first generalized the Frey--Stremme  model by considering an underlying asset following a L\'{e}vy process with jumps. Then, we establish the corresponding PIDE for option pricing. Furthermore, we investigated the existence and uniqueness of solutions to such PIDE in the multidimensional spaces. 

\section{Preliminaries and Definitions}
This section presents some basic definitions and properties of L\'evy measures and notion of admissible activity L\'evy measures. Here,$|\cdot|$ and $\Vert\cdot\Vert$ represent the Euclidean norm in $\R^n$ and the norm in an infinite dimensional function space (e.g., $L^p(\R^n), X^\gamma$). HIn what follows, $a\cdot b$ stands for the usual Euclidean product in $\R^n$ with the norm $|z| = \sqrt{z\cdot z}$.
\begin{definition}\cite{vsevvcovivc2021multidimensional}
A L\'evy process on $\R^n$ is a stochastic (right continuous) process $X = \{X_t, t\geq 0\}$ having the left limit with independent stationary increments. It is uniquely characterized by its L\'evy exponent $ \phi $:
\[
\mathbb{E}_x(e^{i y\cdot X_t}) = e^{-t\phi(y)}, ~ y\in\R^n.
\]
The subscript $x$ in the expectation operator $\mathbb{E}_x$ indicate that the process $X_t$ starts from a given value $x$ at the origin $t=0$. The L\'evy exponent $\phi$ has the following unique decomposition: 
\[
\phi(y) = ib\cdot y + \sum_{i, j=1}^n a_{ij} y_i y_j + \int_{\R^n}\left(1-e^{iy\cdot z} + iy\cdot  z1_{|z|\leq 1}\right)\nu(\ud z), 
\]
where $b\in\R^n$ is a constant vector; $(a_{ij})$ is a constant matrix, which is positive semidefinite; $\nu(\ud z)$ is a nonnegative measure on $\R^n \setminus \{0\}$ such that $\int_{\R^n}\min(1, |z|^2)\nu(\ud z) <\infty$ (c.f.,  \cite{palatucci2017recent}).
\end{definition}
\subsection{Exponential L\'{e}vy models}\label{sectiontwo}

Let $X_t, t\ge0,$ be a stochastic process. The Poisson random measure $\nu(A)$ of a Borel set $A \in \mathcal{B}(\mathbb{R})$ is defined  by $\nu(A)=\mathbb{E}\left[ J_{X}([0,1]\times A)\right]$, where $J_{X}([0,t]\times A)=\# \left\{s \in [0,t]:\Delta X_{s} \in A \right\}$. This measure gives the mean number per unit time of jumps whose amplitude belongs to the set $A$. It is worth noting that the L\'{e}vy--It\^{o} decomposition provides a representation of $X_t$ interpreted as a combination of a Brownian motion with a drift $\omega$ and an infinite sum of independent compensated Poisson processes with variable jump sizes (see \cite{CruzSevcovic2020}), i.e., 
\[
\ud X_{t}=\omega \ud t+\sigma \ud W_{t}+ \int_{\left|x\right|\geq 1}x J_{X}\left(\ud t,\ud x\right)+ \int_{\left|x\right|<1}x \widetilde{J}_{X}\left(\ud t,\ud x\right), 
\]
where $\widetilde{J}_{X}\left(\left[0,t\right]\times A\right) =J_{X}\left(\left[0,t\right]\times A\right)-t \nu\left(A\right)$ is the compensation of $J_{X}$.

\begin{remark}\cite[Remark 1]{CruzSevcovic2020}
Note that any L\'{e}vy process is a strong Markov process, and the associated semigroup is a convolution semigroup. Its infinitesimal generator $L:u\mapsto L[u]$ is a nonlocal  partial integro-differential operator given by
\begin{eqnarray}
L[u](x)&=&\lim_{h \to 0^+} \frac{\mathbb{E}\left[u\left(x+X_{h}\right)\right]-u\left(x\right)}{h}
\nonumber \\
&=&\frac{\sigma^{2}}{2}\frac{\partial^{2}u}{\partial x^{2}}+\gamma\frac{\partial u}{\partial x}+\int_{\mathbb{R}}\left[u\left(x+y\right)-u\left(x\right)-y 1_{\left|y\right| \leq 1}\frac{\partial u}{\partial x}\left(x\right)\right] \nu(\ud y), 
\label{eq:infgen no condition}
\end{eqnarray}
which is well-defined for any compactly supported function $u\in C^{2}_0\left(\mathbb{R}\right)$.
\end{remark}
Let $S_{t}, t\ge 0,$ be a stochastic process representing an underlying asset process under a filtered probability space $\left(\Omega,\mathcal{F},\left\{\mathcal{F}_{t}\right\},\mathbb{P}\right)$. The filtration $\left\{\mathcal{F}_{t}\right\}$ represents the price history up to the time $t$. In an arbitrage-free market, there is an equivalent  measure $\mathbb{Q}$ under which discounted prices of all traded assets are $\mathbb{Q}-$ martingales, which is called the fundamental theorem of asset pricing (see,  \cite{ConTan03}). The measure $\mathbb{Q}$ is also known as the risk-neutral measure. We consider the exponential L\'{e}vy model in which the risk-neutral price process $S_{t}$ under $\mathbb{Q}$ is given by $S_{t}=e^{rt+X_{t}}$, where $X_{t}$ is a L\'{e}vy process under $\mathbb{Q}$ with the characteristic triplet $\left(\sigma,\gamma,\nu\right)$. Then, the arbitrage-free market hypothesis states that $\widehat{S}_{t}=S_{t}e^{-rt}=e^{X_{t}}$ is a martingale, which is equivalent to the following conditions imposed on the triplet $\left(\sigma,\gamma,\nu\right)$:

\begin{equation}
\int_{\left|y\right|\ge 1} e^{y} \nu (\ud y)<\infty,\ \ \gamma\in\mathbb{R},\ \  \gamma=-\frac{\sigma^{2}}{2}-\int_{-\infty}^{+\infty} \left(e^{y}-1-y1_{\left|y\right|\leq 1}\right) \nu(\ud y). \label{expcond}
\end{equation}
The risk-neutral dynamics of $S_{t}$ under $\mathbb{Q}$ is given by
\begin{equation}
\ud S_{t}=r S_t \ud t+ \sigma S_t \ud W_{t}+ \int_{\mathbb{R}}\left(e^{y}-1\right)S_{t}\widetilde{J}_{X}\left(dt,dy\right).
\end{equation}
The exponential price process, $e^{X_{t}}, t\ge 0 $, is also a Markov process with the state space $\left(0,\infty\right)$. It has the following infinitesimal generator:
\begin{eqnarray}
L^{S}[V](S)&=&\lim_{h \rightarrow 0} \frac{\mathbb{E}[V(S e^{X_{h}})]-V(S)}{h} = r S\frac{\partial V}{\partial S}+\frac{\sigma^{2}}{2}S^2\frac{\partial^{2}V}{\partial S^{2}}\\
&& +\int_{\mathbb{R}}\left[V(S e^{y})-V(S)-S (e^{y}-1)\frac{\partial V}{\partial S}\right]\nu(\ud y)\label{eq: infgenS }
\end{eqnarray}
(see \cite{ConTan03}). A L\'{e}vy process with the  following representation:
\[
\ud X_{t}=\omega \ud t+\sigma \ud W_{t}+\int_{|x|\ge 1} K(t,x) J_{X}(\ud t, \ud x)
+\int_{|x|<1} H(t,x) \tilde{J}_{X}(\ud t, \ud x).
\]
is called the L\'{e}vy type stochastic integral. The following variant of It\^{o}'s lemma is an essential results, which will be needed later in this study.

\begin{theorem}\cite[Therorem 2.1]{NBS19}
Let $f\in C^{1,2}([0,T]\times\mathbb{R})$ and  $H, K\in C([0,T]\times\mathbb{R})$. Suppose $X_{t}, t\ge 0,$ is a L\'{e}vy stochastic process. Then, 
\begin{eqnarray}
\ud f(t,X_{t})&=&\frac{\partial f}{\partial t}\ud t+\frac{\partial f}{\partial x}\ud X_{t}+\frac{1}{2}\frac{\partial^2 f}{\partial x^2}\ud [X_{t},X_{t}]\nonumber
\\
&&+\int_{|x|\ge 1} f(t,X_{t}+K(t,x))-f(t,X_{t})J_{X}(\ud t, \ud x)\label{itolemma}
\\
&&+\int_{|x|<1} f(t,X_{t}+H(t,x))-f(t,X_{t})\tilde{J}_{X}(\ud t, \ud x)\nonumber
\\
&&+\int_{|x|<1} f(t,X_{t}+H(t,x))-f(t,X_{t})-H(t,x)\frac{\partial f}{\partial x}(t,X_t)\nu(\ud x)\ud t .
\nonumber
\end{eqnarray}
\end{theorem}
\subsection{Examples of L\'{e}vy processes in finance}\label{sectionthird}

There are two types of exponential L\'{e}vy models considered in the literature. The first types are jump-diffusion models, where the log-price is represented as a L\'{e}vy process with a nonzero diffusion part $(\sigma >0)$ and a jump process with finite activity $($i.e., $\nu(\mathbb{R})<\infty)$. The second types of models are infinite activity pure jump models, where there is no diffusion part and only a jump process with infinite activity $($i.e., $\nu(\mathbb{R})=\infty)$. This section presents different types of exponential L\'{e}vy models that differ in the choice of the L\'{e}vy measure.
\subsubsection{Jump-Diffusion models}

A L\'{e}vy process with jump-diffusion has the following general form:
\begin{equation*}
X_{t}=\gamma t +\sigma W_{t}+\sum_{i=1}^{N_{t}} Y_{i},
\end{equation*} 
where $\sigma>0$, and $ N_{t}$ is a Poisson process with intensity $\lambda$ that counts the jumps of $X_{t}$, and $Y_{i}, i=1,2,3...$ are independent and identically distributed random variables with distribution $\mu $. The L\'{e}vy measure $\nu$ is $\lambda \mu$, and the drift $\gamma$ is given by \begin{equation*}-\frac{\sigma^{2}}{2}-\int_{\mathbb{R}} \left(e^{y}-1-y1_{\left|y\right|\leq 1} \right)\nu \left(\ud y \right).
\end{equation*}

\subsubsection*{Merton's model}
This is the first jump-diffusion model proposed by Merton \cite{Merton76} in the financial application. The random variables $Y_{i}, i=1,2,3...$, are normally distributed with mean $m$ and variance $\delta^2$.
It has the following L\'{e}vy density:
\begin{equation}
\nu(\ud z)=\lambda \frac{1}{(2\pi\delta^2)^{n/2}}e^{-\frac{|z-m|^2}{2\delta^2}}\ud z\,,
\label{merton-density}
\end{equation}
where the parameters $m\in\R^n, \lambda, \delta>0,$ are given. Merton's measure is a finite activity L\'evy measure, i.e., $\nu(\mathbb{R}^n)=\int_{\mathbb{R}^n}\nu(\ud z)<\infty$, with finite variation $\int_{|z|\leq 1}|z|\nu(\ud z)<\infty$. Therefore, the probability density of $X_{t}$ can be obtain as a series that converges rapidly (see \cite{ConTan03}):
\begin{IEEEeqnarray}{rCl}
p_{t}(x)&=&\sum_{j=0}^{\infty}e^{-\lambda t}(\lambda t)^j\frac{e^{-\frac{|z-\gamma t-jm|^2}{2(\sigma^2 t+j\delta^2)}}}{j!\sqrt{2 \pi(\sigma^2 t+j\delta^2)}}.
\end{IEEEeqnarray}
Thus, the price of an European call option can be expressed as a weighted sum of Black--Scholes prices:
\begin{equation}
C_{Merton}(S_{0},K,T,\sigma,r)=e^{-rT}\sum_{j=0}^{\infty}e^{-\lambda t}\frac{(\lambda t)^j}{j!}e^{r_{j}T}C_{BS}(S_{0}e^{\frac{j\delta^2}{T}},K,T,\sigma_{j},r_{j}),
\end{equation}   
where $r_{j}=r-\lambda(e^{m+\frac{\delta^2}{2}}-1)+\frac{jm}{T}$, $\sigma_{j}=\sqrt{\sigma^2+\frac{j\delta^2}{T}}$, and $C_{BS}(S,K,T,\sigma,r)$ is the well-known Black--Scholes formula. 

\subsubsection{Infinite activity pure jump models}

The variance Gammma and normal inverse Gaussian (NIG) processes are obtained by a subordination of a Brownian motion and a tempered $\alpha$-stable process; variance Gamma and NIG processes correspond to $\alpha=0$ and $\alpha=1/2$, respectively. These models are widely used in finance because of the existence of probability density of the subordinator in a closed form for these values of $\alpha$ (see \cite{ConTan03}).

\subsubsection*{Variance Gamma process}

This is a pure discontinuous process of infinite activity and finite variation ($\int_{|x|\leq 1}|x|\nu(\ud x)<\infty$) that is widely used in the financial modeling. It has the following L\'{e}vy measure:
\begin{equation*}\nu \left(x\right)=\frac{1}{\kappa \left|x\right|}e^{Ax-B\left|x\right|} \text{ with } A=\frac{\theta}{\sigma^{2}} \text{ and } B=\frac{\sqrt{\theta^2+2\frac{\sigma^2}{\kappa}}}{\sigma^2}.
\end{equation*}
Here, $\sigma$ and $\theta$ are parameters related to the volatility and drift of the Brownian motion, respectively; $\kappa$ is the parameter related to the variance of the subordinator, which is a the Gamma process (see \cite{ConTan03}). The probability density is given by
\begin{equation*}
p_{t}(x)=Ce^{A x}|x|^{\frac{t}{k}}K_{\frac{t}{k}-\frac{1}{2}}(|x|),
\end{equation*}
where $K$ is the modified Bessel function of second kind. The characteristic function of $X_{t}+\gamma t$ is given by
\begin{equation*}
\Phi_{t}\left(u\right)=e^{itu\gamma}\phi_{t}\left(u\right)=e^{itu\gamma}\left(1+\frac{\sigma^2 u^2 \kappa}{2}-i\theta \kappa u\right)^{-t/\kappa},
\end{equation*} 
where $\gamma$ is determined by the martingale condition, and $\phi_{t}\left(u\right)$ is the characteristic function of $X_{t}$.
Moreover,  we have 
\begin{equation}
\mathbb{E}[e^{-rT}S_{T}|\mathbb{F}_{t}]=e^{-rt}S_{t},
\end{equation} 
where 
\begin{equation}
S_{t}=S_{0}e^{rt+\gamma t+X_{t}}
\end{equation}
is the risk-neutral process introduced in \cite{DPE98}. Therefore, $\gamma=\frac{1}{\kappa}log(1-\frac{\sigma^2 \kappa}{2}-\theta\kappa).$

\subsubsection*{Normal inverse Gaussian model}

The NIG process is a process of infinite activity and infinite variation without any Brownian component. It has the following L\'{e}vy measure \cite{ConTan03}

\begin{equation*}\nu\left(x\right)=\frac{C}{\left|x\right|}e^{Ax}K_{1}\left(B\left|x\right|\right)\end{equation*}

and \begin{equation*}C=\frac{\sqrt{\theta^{2}+\frac{\sigma^{2}}{\kappa}}}{2\pi\sigma\sqrt{\kappa}}, A=\frac{\theta}{\sigma^2}, B=\frac{\sqrt{\theta^{2}+\frac{\sigma^{2}}{\kappa}}}{\sigma^{2}},\end{equation*}
where $\theta$, $\sigma$, and $\kappa$ have the same meaning as in the Variance Gamma process.
The probability density is given by
\begin{equation*}
p_{t}(x)=Ce^{Ax}\frac{K_{1}(B\sqrt{x^2+\frac{t^2\sigma^2}{\kappa}})}{\sqrt{x^2+\frac{t^2\sigma^2}{\kappa}}}
\end{equation*}
where $K$ is the modified Bessel function of second kind.
The characteristic function is given by
\begin{equation}
\Phi_{t}\left(u\right)=e^{\frac{t}{\kappa}-\frac{t}{\kappa}\sqrt{1+u^2\sigma^2\kappa-2iu\theta\kappa}}.
\end{equation}
\subsubsection*{Generalized hyperbolic model}
The generalized hyperbolic model is a process of infinite variation without Gaussian part. It has the following characteristic function (see \cite{ConTan03}):
\begin{equation}
\phi_{t}(u)=e^{i\mu u}(\frac{\alpha^{2}-\beta^{2}}{\alpha^{2}-(\beta+iu)^{2}})^{\frac{t}{2\kappa}}\frac{K_{\frac{t}{\kappa}}(\delta\sqrt{\lambda^{2}-(\beta+iu)^{2}})}{K_{\frac{t}{\kappa}}(\delta\sqrt{\alpha^{2}-\beta^{2}})},
\end{equation} 
where $\delta$ is a scale parameter, $\mu$ is the shift parameter, and $\kappa$ has the same meaning as in the variance Gamma process. The parameters $\lambda$, $\alpha$, and $\beta$ determine the shape of the distribution.
The density function
\begin{equation*}
p_{t}(x)=C(\sqrt{\delta^{2}+(x-\mu)^{2}})^{\frac{t}{k}-\frac{1}{2}}K_{\frac{t}{\kappa}-\frac{1}{2}}(\alpha\sqrt{\delta^{2}-(x-\mu)^{2}})e^{\beta(x-\mu)},
\end{equation*}
where $K$ is the modified Bessel function and 
\begin{equation*} C=\frac{(\sqrt{\alpha^{2}-\beta^{2}})^{\frac{t}{k}}}{\sqrt{2\pi}\alpha^{\frac{t}{\kappa}-\frac{1}{2}}\delta^{\frac{t}{\kappa}}K_{\frac{t}{\kappa}}(\delta\sqrt{\alpha^{2}-\beta^{2}})}.
\end{equation*}
The variance Gamma process is obtained for $\mu=0$ and $\delta=0$. The NIG process corresponds to $\lambda=-\frac{1}{2}$.

\subsection{Admissible activity L\'evy measures}

This subsection presents the notion of an admissible activity L\'evy measure introduced by Cruz and \v{S}ev\v{c}ovi\v{c} \cite{NBS19, CruzSevcovic2020} for the one-dimensional case $n=1$, which was later extended by \v{S}ev\v{c}ovi\v{c} and Udeani \cite{vsevvcovivc2021multidimensional} for the multidimensional case $n\ge 1$.

\begin{definition}\cite[Definition 1]{vsevvcovivc2021multidimensional}
\label{def-admissiblemeasure}
A measure $\nu$ in $\R^n$ is called an admissible activity L\'evy measure if  there exists a nonnegative Lebesgue measurable function $h:\R^n\to \R$ such that $\nu(\ud z) = h(z) \ud z$ with 
\begin{equation}
0 \le  h(z)\le C_0 | z|^{-\alpha} e^{- D |z| - \mu |z|^{2}},
\label{growth_measure}
\end{equation}
for all $z\in\R^n$ and the shape parameters $\alpha, \mu\geq 0, D\in \R$ ($D>0$ if $\mu=0)$, where $C_0>0$ is a positive constant. 
\end{definition}
\begin{remark}
It is worth noting that the additional conditions $\int_{\mathbb{R}} \min(|z|^2,1) \nu (\ud z)<\infty$ and 
$\int_{\left|z\right|>1} e^{z} \nu (\ud z)<\infty$ are satisfied provided that $\nu$ is an admissible L\'evy measure with shape parameters $\alpha<3$, and either $\mu>0, D^\pm\in \mathbb{R}$, or $\mu=0$ and $D^-+1<0<D^+$. For the Merton model, we have $\alpha=0, D^\pm=0$ and $\mu=1/(2\delta^2)>0$. Meanwhile, for the Kou model, we have $\alpha=\mu=0, D^+=\lambda^-, D^-=-\lambda^+$. For the variance Gamma process, we have $\alpha=1, \mu=0, D^\pm=A\pm B$.
\end{remark}

\section{Multidimensional Linear and Nonlinear PIDE}
This section focuses on qualitative properties of solutions to the linear and nonlinear nonlocal parabolic PIDE of the form: 
\begin{eqnarray}
\frac{\partial u}{\partial \tau} &=& \frac{\sigma^2}{2}\Delta u 
+ \int_{\mathbb{R}^n}\left[ u(\tau, x+z)-u(\tau, x)- z\cdot \nabla_x u(\tau,x)
\right] \nu(\ud z) +  g(\tau, x, u, \nabla_x u), 
\label{PDE-u}
\\
&&u(0,x)=u_0(x), \quad x\in \R^n, \tau\in(0,T),
\nonumber
\end{eqnarray}
where $g$ is a given sufficiently smooth function; $\nu $ is a positive measure on $\mathbb{R}^n$ such that its Radon derivative is a nonnegative Lebesgue measurable function $h$ in $\R^n$, i.e., $\nu (\ud z) = h(z) \ud z$. Additionally, we will analyze the solution of the following generalization of the above PIDE, in which the shift function may depend on the variables $\tau>0, x,z\in\R$:
\begin{equation}
\frac{\partial u}{\partial \tau} = \frac{\sigma^2}{2} \Delta u 
+ \int_{\mathbb{R}^n}\left[ u(\tau, x+ \xi )-u(\tau, x)- \xi \cdot \nabla_x u(\tau,x)
\right] \nu(\ud z) +  g(\tau, x, u, \nabla_x u) ,
\label{PDE-u-general}
\end{equation}
where $\xi=\xi(\tau, x, z)$ is the shift function. An application of such a general shift function $\xi$ can be found in nonlinear option pricing models considering a large trader stock-trading strategy with the underlying asset price dynamic following the L\'evy process (c.f.,  Cruz and \v{S}ev\v{c}ovi\v{c} \cite{NBS19}). If $\xi(x, z)\equiv z$, then (\ref{PDE-u-general}) reduces to equation (\ref{PDE-u}). Fr example, the nonlinearity $g$ often arises from applications occurring in pricing XVA derivatives (c.f.,  Arregui \emph{et al.} \cite{NBS15, NBS17}) or applications of the penalty method for American option pricing under a PIDE model (c.f.,   Cruz and \v{S}ev\v{c}ovi\v{c} \cite{CruzSevcovic2020}).

\subsection{Existence and uniqueness results of PIDE}

In this section, we present the existence and uniqueness results for the general (\ref{PDE-u-general}) for a class of L\'evy measures using the theory of abstract semilinear parabolic equation in the scale of Bessel potential spaces. First, we rewrite the PIDE (\ref{PDE-u-general}) in high-dimensional space as follows:
\begin{eqnarray}
&&\frac{\partial u}{\partial \tau} + A u = 
 f(u) + g(\tau, x, u, \nabla_x u), \;\; u(0,x)=u_0(x), \; x\in \mathbb{R}^n, \tau\in(0,T),
\label{problem_transformed}
\end{eqnarray}
where  $A = -(\sigma^2/2) \Delta $. The linear nonlocal operator $f$ is defined by
\begin{equation}
f(u)(\cdot) =
\int_{\mathbb{R}^n}\left[ u(\cdot+\xi)-u(\cdot)- \xi \cdot \nabla_x u(\cdot)\, \right] \nu(\ud z),
\label{functional_f_def}
\end{equation}
where $\xi = \xi(\tau,x,z)$ is a given shift function. The function $g$ is assumed to be H\"older and Lipschitz continuous in the $\tau$ and other variables, respectively. Then, we employ the theory of abstract semilinear parabolic equations presented by Henry \cite{Henry1981} to establish the existence, continuation, and uniqueness of a solution. A solution to the PIDE (\ref{problem_transformed}) is constructed in the scale of the Bessel potential spaces ${\mathscr L}^p_{2\gamma}(\mathbb{R}^n),\gamma\ge 0$ in high-dimensional space, $n\geq 1$. These spaces can be viewed as a natural extension of the classical Sobolev spaces $W^{k,p}(\mathbb{R}^n)$ for non-integer values of order $k$. It is worth noting that nested scale of Bessel potential spaces allows for a finer formulation of existence and uniqueness results than the classical Sobolev spaces.

\begin{definition}\cite[Definition 1]{Henry1981}
An analytic semigroup is a family of bounded linear operators $\left\{S(t), t\geq 0\right\}$ in a Banach space $X$ satisfying the following conditions:
\begin{itemize}
\item[i)] $S(0)=I, S(t)S(s)=S(s)S(t)=S(t+s)$, for all $t,s\geq 0$;
\item[ii)] $S(t) u\rightarrow u$ when $t\rightarrow 0^{+}$ for all $u\in X$;
\item[iii)] $t\rightarrow S(t) u$ is a real analytic function on $0< t< \infty$ for each $u\in X$.
\end{itemize} 
The associated infinitesimal generator $A$ is defined as follows: $A u = \lim_{t\rightarrow 0^{+}} \frac{1}{t} (S(t)u-u)$ and its domain $D(A)\subseteq X$ consists of those elements $u\in X$ for which the limit exists in the space $X$. 
\end{definition}

\begin{definition}\cite{Henry1981}
Let $S_{a,\phi}=\left\{\lambda \in \mathbb{C}: \phi\leq \arg(\lambda-a)\leq 2\pi-\phi \right\}$ be a sector of complex numbers. A closed densely defined linear operator $A: D(A)\subset X \rightarrow X$ is called a sectorial operator if there exists a constant $M\geq 0$ such that $\Vert(A-\lambda )^{-1}\Vert\leq M/|\lambda -a|$ for all $\lambda\in S_{a,\phi} \subset \mathbb{C}\setminus\sigma(A)$.
\end{definition}

Next, we briefly recall the construction and basic properties of Bessel potential spaces. It is worth noting if $A$ is a sectorial operator in a Banach space $X$, then $-A$ is a generator of an analytic semigroup $\left\{e^{-A t}, t\geq 0\right\}$ acting on $X$ (c.f.,  \cite[Chapter I]{Henry1981}). For any $\gamma >0$, we can introduce the operator $A^{-\gamma}:X\to X$ as follows: $A^{-\gamma}=\frac{1}{\Gamma(\gamma)}\int_{0}^{\infty} \xi^{\gamma-1}e^{-A \xi} \ud \xi$. Then, the fractional power space $X^{\gamma} = D(A^{\gamma})$ is the domain of the operator $A^\gamma= (A^{-\gamma})^{-1}$, i.e., $X^{\gamma}=\left\{u\in X:\  \exists \varphi\in X,  u=A^{-\gamma}\varphi\right\}$. The norm is defined as follows: $\Vert u\Vert_{X^\gamma}=\Vert A^{\gamma}u\Vert_X=\Vert \varphi\Vert_X$. Furthermore, we have continuous embedding: $D(A)\equiv X^1 \hookrightarrow X^{\gamma_1} \hookrightarrow X^{\gamma_2} \hookrightarrow X^0\equiv X$, for $0\le \gamma_2\le \gamma_1\le 1$. 

Let us recall the convolution operator $(G*\varphi)(x)=\int_{\mathbb{R}^n} G(x-y)\varphi(y)\ud{y}$. According to \cite[Section 1.6]{Henry1981}, $ A= -(\sigma^2/2) \Delta$ is a sectorial operator in the Lebesgue space $X=L^{p}(\mathbb{R}^{n})$ for any $p\ge 1, n\ge 1$, and $D(A) \subset W^{2,p}(\mathbb{R}^{n})$. It follows from  \cite[Chapter 5]{Stein1970} that the space $X^\gamma, \gamma>0,$ can be identified with the Bessel potential space ${\mathscr L}^p_{2\gamma}(\mathbb{R}^n)$, where 
\[
{\mathscr L}^p_{2\gamma}(\mathbb{R}^n):=\{u\in X:\  \exists \varphi\in X, u=G_{2\gamma}*\varphi\}.
\]
Here, $G_{2\gamma}$ is the Bessel potential function, 
\[
G_{2\gamma}(x) = \frac{1}{(4\pi)^{n/2}\Gamma(\gamma)} \int_0^\infty  y^{-1+\gamma-n/2} e^{-(y +|x|^2/(4y))}\ud y .
\]
The norm of $u=G_{2\gamma}*\varphi$ is given by $\Vert u\Vert_{X^\gamma}=\Vert \varphi\Vert_{L^p}$. The space $X^\gamma$ is continuously embedded in the fractional Sobolev--Slobodeckii space $W^{2\gamma,p}(\R^n)$ (c.f.,  \cite[Section 1.6]{Henry1981}).

In what follows, we denote $C_0>0$ as a generic constant, which is independent of the solution $u$; however, it may depend on the model parameters, e.g., $n\ge 1, p\ge 1, \gamma\in[0,1)$.

\begin{proposition} \cite[Proposition 1]{vsevvcovivc2021multidimensional}
\label{prop_pointwise_est-2}
Let us define the mapping $Q(u,\xi)$ as follows:
\[
Q(u,\xi) = u(x +\xi(x))-\xi(x)\cdot\nabla_x u(x), \quad x\in\R^n.
\]
Then, there exists a constant $\hat{C}>0$ such that, for any vector valued functions  $\xi_1, \xi_2\in (L^\infty(\R^n))^n$, and $u$ such that $\nabla_x u \in (X^{\gamma-1/2})^n$, $1/2\le \gamma<1$, the following estimate holds:
\[
\Vert Q(u,\xi_1) - Q(u,\xi_2)  \Vert_{L^p({\R^n})}
\leq \hat{C} \Vert\xi_1-\xi_2\Vert_\infty^{2\gamma-1} (\Vert\xi_1\Vert_\infty+\Vert\xi_2\Vert_\infty) \Vert\nabla_x u\Vert_{X^{\gamma-1/2}}.
\]
\end{proposition}

\begin{proof}
Let $u\in X$ be such that  $\nabla_x u\in (X^{\gamma-1/2})^n$, i.e., $\partial_{x_i}u\in X^{\gamma-1/2} ~ \text{for each} ~ i = 1, \cdots, n$. Then, $\nabla_x u= A^{-(2\gamma-1)/2}\varphi  = G_{2\gamma-1} * \varphi$ for some $\varphi\in (L^p(\R^n))^n$, and $\Vert \nabla_x u\Vert_{X^{\gamma-1/2}} = \Vert A^{(2\gamma-1)/2}\nabla_x u\Vert_X = \Vert \varphi\Vert_{L^p}$. 
Here, $\varphi = (\varphi_1, \cdots, \varphi_n)$ and $ \partial_{x_i}u =  G_{2\gamma-1} * \varphi_i$. Let $x, \xi\in \R^n$. Then, 
\[
\nabla_x u(x+\xi) = G_{2\gamma-1}(x+\xi - \cdot)* \varphi(\cdot), \qquad  \nabla_x u(x)
= G_{2\gamma-1}(x - \cdot)* \varphi(\cdot).
\]
Recall that the following inequality holds for convolution operator:
\[
\Vert \psi*\varphi\Vert_{L^p(\R^n)}\le \Vert \psi\Vert_{L^q(\R^n)} \Vert \varphi\Vert_{L^r(\R^n)},
\]
where  $p,q,r\ge 1$ and $1/p + 1 = 1/q + 1/r$ (see \cite[Section 1.6]{Henry1981}). In particular, for $q=1$, we have $\Vert \psi*\varphi\Vert_{L^p}\le \Vert \psi\Vert_{L^1} \Vert \varphi\Vert_{L^p}$. 
The following estimate holds for the modulus of continuity of the Bessel potential function $G_{\alpha}, \alpha\in(0,1)$:
\[
\Vert G_{\alpha}(\cdot + h) - G_{\alpha}(\cdot) \Vert_{L^1}
\le C_0 |h|^{\alpha},
\]
for any $h\in \R^n$ (c.f.,  \cite[Chapter 5.4, Proposition 7]{Stein1970}).
Let $\xi_1, \xi_2$ be bounded vector valued  functions, i.e., $\xi_1, \xi_2\in (L^\infty(\R^n))^n$. Then, for any $x\in \R^n ~ \text{and} ~\theta\in [0,1]$, we have
\begin{eqnarray*}
&&
u(x+\xi_1(x)) - u(x+\xi_2(x)) - (\xi_1(x) -\xi_2(x)) \cdot\nabla_x u(x)
\\
&=& u(x+\xi_1(x)) - \nabla_x u(x) - \xi_1(x)\cdot \nabla_x u(x) 
\\
&& - [u(x+\xi_2(x))- \nabla_x u(x) -\xi_2(x)\cdot \nabla_x u(x)]
\\
&=& (\xi_1(x) -\xi_2(x))\int_0^1\nabla_x u(x+\theta \xi_1(x)) - \nabla_x u(x)\ud{\theta}  
\\
&&
+ \int_0^1 \nabla_x u(x+\theta \xi_1(x))-  \nabla_x u(x+\theta \xi_2(x))\ud{\theta}
\ .
\end{eqnarray*}
Now,
\begin{eqnarray*}
&&
\Vert Q(u,\xi_1) - Q(u,\xi_2)  \Vert^p_{L^p({\R^n})}
\\
&=& \int_{\R^n} |u(x+\xi_1(x)) - u(x+\xi_2(x)) - (\xi_1(x) - \xi_2(x)) \cdot\nabla_x u(x)|^p\ud{x}
\\
&\le&  \int_{\R^n} \left|(\xi_1(x)-\xi_2(x)) \int_0^1 \nabla_x u(x+\theta \xi_1(x) )-  \nabla_x u(x)\ud{\theta} \right|^p \ud{x}
\\
&& + \int_{\R^n} \left|\xi_2(x) \int_0^1 \nabla_x u(x+\theta \xi_1(x))- \nabla_x u(x+\theta\xi_2(x))\ud{\theta} \right|^p \ud{x}
\\
&\le& \Vert\xi_1 - \xi_2\Vert_{\infty}^p\int_0^1 \int_{\R^n} |\nabla_x u(x+\theta\xi_1(x))  - \nabla_x u(x)|^p dx d\theta 
\\
&& + \Vert\xi_2\Vert_{\infty}^p\int_0^1 \int_{\R^n} |\nabla_x u(x+\theta\xi_1(x))  - \nabla_x u(x+\theta\xi_2(x)|^p dx d\theta 
\\
&\le&  \Vert\xi_1 - \xi_2\Vert_{\infty}^p\int_0^1 \Vert\left( G_{2\gamma-1}(\cdot + \theta\xi_1) - G_{2\gamma-1}(\cdot)\right)*\varphi\Vert_{L^p}^p d\theta 
\\
&& + \Vert\xi_2\Vert_{\infty}^p\int_0^1 \Vert\left( G_{2\gamma-1}(\cdot + \theta\xi_1) - G_{2\gamma-1}(\cdot+\theta\xi_2)\right)*\varphi\Vert_{L^p}^p d\theta 
\\
&\le&  \Vert\xi_1 - \xi_2\Vert_{\infty}^p\int_0^1  \Vert G_{2\gamma-1}(\cdot + \theta\xi_1 ) - G_{2\gamma-1}(\cdot)\Vert^p_{L^1} d\theta \Vert\varphi\Vert_{L^p}^p 
\\
&& +  \Vert \xi_2\Vert_{\infty}^p \int_0^1 \Vert G_{2\gamma-1}(\cdot + \theta\xi_1) - G_{2\gamma-1}(\cdot+\theta\xi_2)\Vert^p_{L^1} d\theta \Vert\varphi\Vert_{L^p}^p
\\
&\le&  \left(\Vert\xi_1 - \xi_2\Vert_{\infty}^p \Vert\xi_1 \Vert^{(2\gamma-1)p}_{\infty}
+
\Vert \xi_2\Vert_{\infty}^p \Vert\xi_1-\xi_2 \Vert^{(2\gamma-1)p}_{\infty}\right)C_0^p\Vert\nabla_x u\Vert^p_{X^{\gamma-1/2}} 
\\
&\le&  \Vert\xi_1-\xi_2 \Vert^{(2\gamma-1)p}_{\infty}\left(\Vert\xi_1\Vert_{\infty}^p +\Vert\xi_2\Vert^{(2-2\gamma)p}_{\infty} \Vert\xi_1 \Vert^{(2\gamma-1)p}_{\infty}
+
\Vert \xi_2\Vert_{\infty}^p \right)C_0^p\Vert\nabla_x u\Vert^p_{X^{\gamma-1/2}} 
\ .
\end{eqnarray*}
By Young's inequality, we have $ab\leq \frac{a^\alpha}{\alpha}+ \frac{b^\beta}{\beta}$ for any $a, b\geq 0$, and $\alpha, \beta>1 $ with $1/\alpha + 1/\beta=1$  (c.f., \cite{brezis2010functional}). Set $ \alpha = 1/(2- 2\gamma), \beta =1/(2\gamma -1)$. Then, $1/\alpha+1/\beta=1$, and  we obtain 
$\Vert\xi_2\Vert^{(2-2\gamma)p}_{\infty} \Vert\xi_1 \Vert^{(2\gamma-1)p}_{\infty}\leq (2-2\gamma)\Vert \xi_2\Vert_{\infty}^p + (2\gamma -1)\Vert \xi_1\Vert_{\infty}^p\leq 2\Vert\xi_2\Vert_{\infty}^p + \Vert\xi_1\Vert_{\infty}^p$. Therefore, 
\begin{eqnarray*}
\Vert Q(u,\xi_1) - Q(u,\xi_2)  \Vert_{L^p({\R^n})}^p
&\le&  2\Vert\xi_1-\xi_2 \Vert^{(2\gamma-1)p}_{\infty}\left(\Vert\xi_1\Vert_{\infty}^p +
\Vert \xi_2\Vert_{\infty}^p \right)C_0^p\Vert\nabla_x u\Vert^p_{X^{\gamma-1/2}} 
\\
&\leq& 2 C_0^p \Vert\xi_1-\xi_2\Vert_\infty^{(2\gamma-1)p} (\Vert\xi_1\Vert_\infty+\Vert\xi_2\Vert_\infty)^p \Vert\nabla_x u\Vert^p_{X^{\gamma-1/2}}.
\end{eqnarray*}
Hence, the pointwise estimate holds with the constant $\hat{C} = 2^{1/p}C_0 >0$. 
\end{proof}

Applying  Proposition \ref{prop_pointwise_est-2} with $\xi_1=\xi ~ \text{and} ~ \xi_2=0$, we obtain the following corollary. 

\begin{corollary}\cite[Corollary 1]{vsevvcovivc2021multidimensional}
\label{prop_pointwise_est}
Let $u$ be such that $\nabla_x u \in (X^{\gamma-1/2})^n$ where $1>\gamma\ge 1/2$. Then, for any  $\xi\in\R^n$, the following pointwise estimate holds:
\[
\Vert Q(u,\xi) \Vert_{L^p({\R^n})}\leq C_0|\xi|^{2\gamma} \Vert\nabla_x u\Vert_{X^{\gamma-1/2}}.
\]
\end{corollary}

Next, we consider the case when the nonlocal integral term depends on $x$ and $z$ variables. It is a generalization of the result \cite[Lemma 3.4]{NBS19} due to Cruz and \v{S}ev\v{c}ovi\v{c} proven for the case where $\xi(x,z)\equiv z$.

\begin{proposition}\cite[Proposition 2]{vsevvcovivc2021multidimensional}
\label{prop-f} 
Suppose that the shift mapping $\xi=\xi(x,z)$ satisfies $\sup_{x\in\R} |\xi(x,z)|  \le C_0 |z|^\omega (1+ e^{D_0 |z|})$ for some constants $C_0>0, D_0\ge 0, \omega>0$ and any $z\in\R$. Assume $\nu$ is a L\'evy measure with the shape parameters $\alpha, D,$ and either $\mu>0, D\in\R$, or $\mu=0$ and $D>D_0\ge 0$.  Assume $1/2\le \gamma <1$, and $\gamma> (\alpha-n)/(2\omega)$. Then there exists a constant $C_0>0$ such that
\[
\Vert f(u)\Vert_{L^p} \le C_0 \Vert \nabla_x u\Vert_{X^{\gamma-1/2}},
\]
provided that $\nabla_x u \in (X^{\gamma-1/2})^n$. If $u\in X^\gamma$ then $\Vert f(u)\Vert_{L^p} \le C \Vert u\Vert_{X^\gamma}$, i.e., $f:X^\gamma\to X$ is a bounded linear operator.
\end{proposition}

\begin{proof}
The L\'evy measure $\nu(\ud z)$ is given by   $\nu(\ud z)= h(z) \ud z$. Let us denote the auxiliary function $\tilde h(z)=|z|^\alpha h(z)$. Then, $0\le \tilde h(z)\le C_0 e^{- D |z| - \mu |z|^{2}}$. Since $h(z) =  |z|^{-\alpha} \tilde h(z) = h_1(z) h_2(z)$, where $h_1(z)=|z|^{-\beta} \tilde h(z)^\frac12$ and $h_2(z)=|z|^{\beta-\alpha} \tilde h(z)^\frac12$.
Applying  Proposition~\ref{prop_pointwise_est} with $\xi_1=\xi, \xi_2=0$, and using the H\"older inequality, we obtain
\begin{eqnarray*}
\Vert f(u)\Vert_{L^p}^p 
&=& \int_{\R^n} \left|\int_{\R^n} ( u(x+\xi(x,z)) - u(x) - \xi(x,z)\cdot \nabla_x u(x)) h(z)\ud{z}  \right|^p\ud{x}
\\
&\le& \int_{\R^n} \int_{\R^n} \left| u(x+\xi(x,z)) - u(x) - \xi(x,z) \cdot \nabla_x u(x)\right|^p h_1(z)^p \ud{z} 
\\
&& \quad\times \left(\int_{\R^n} h_2(z)^q \ud{z}\right)^{p/q} \ud{x}
\\
&=& \int_{\R^n} \left(\int_{\R^n} \left| u(x+\xi(x,z)) - u(x) - \xi(x,z) \cdot\nabla_x u(x)\right|^p \ud{x}\right) h_1(z)^p \ud{z} 
\\
&& \quad\times \left(\int_{\R^n} h_2(z)^q \ud{z}\right)^{p/q} 
\\
&\le & 
C_0^p  \Vert \nabla_x u\Vert_{X^{\gamma-1/2}}^p
\int_{\R^n} |\xi(x,z)|^{2\gamma p} |z|^{-\beta p} \tilde h(z)^{p/2} \ud{z} \left(\int_{\R^n} h_2(z)^q \ud{z}\right)^{p/q}
\\
&\le & 
C_0^p  \Vert \nabla_x u\Vert_{X^{\gamma-1/2}}^p
\int_{\R^n} |z|^{(2\gamma\omega -\beta) p} \tilde h(z)^{p/2} \ud{z} \left(\int_{\R^n} h_2(z)^q \ud{z}\right)^{p/q}.
\end{eqnarray*}
Assuming $p,q\ge 1, 1/p +1/q=1$ are such that 
\[
(2\gamma\omega -\beta) p > -n, \qquad (\beta-\alpha) q = (\beta-\alpha) \frac{p}{p-1} >-n,
\]
then, the integrals 
$\int_{\R^n} |z| ^{(2\gamma\omega -\beta) p} \tilde h(z)^{p/2} \ud{z}$ and 
$\int_{\R^n} h_2(z)^q \ud{z} = \int_{\R^n} |z| ^{(\beta-\alpha) q} \tilde h(z)^{q/2} \ud{z}$
are finite, provided that the shape parameters satisfy: either $\mu>0, D\in\mathbb{R}$, or $\mu=0, D>D_0\ge 0$. As $\gamma>(\alpha-n)/(2\omega)$, there exists $\beta>1$
satisfying 
\[
\alpha-n+n/p < \beta < 2\gamma\omega +n/p.
\]
Therefore, there exists $C_0>0$ such that $\Vert f(u)\Vert_{L^p} \le C_0 \Vert \nabla_x u\Vert_{X^{\gamma-1/2}}$.
\end{proof}

Let $C([0,T],X^{\gamma})$ be the Banach space consisting of continuous functions from $[0,T]$ to $X^\gamma$ with the maximum norm. The following proposition is due to Henry \cite{Henry1981} (see also Cruz and \v{S}ev\v{c}ovi\v{c} \cite{CruzSevcovic2020}).

\begin{proposition}\cite[Proposition 3.5]{Henry1981}
\label{semilinear_general_existence_result}
Suppose that the linear operator $-A$ is a generator of an analytic semigroup $\left\{e^{-At},t\geq 0\right\}$ in a Banach space $X$. Assume the initial condition $U_0$ belongs to the space $X^{\gamma}$ where $0\leq \gamma <1$.  Suppose that the mappings $F:[0,T]\times X^{\gamma}\to X$ and $h:(0,T]\to X$ are H\"older continuous in the $\tau$ variable, $\int_0^T \Vert h(\tau)\Vert_X \ud \tau <\infty$, and $F$ is Lipschitz continuous in the $U$ variable. Then, for any $T>0$, there exists a unique solution to the abstract semilinear evolution equation: $\partial_\tau  U+A U=F(\tau, U) +h(\tau)$ 
such that $U\in C([0,T],X^{\gamma}), U(0)=U_{0}, \partial_\tau U(\tau) \in X, U(\tau)\in D(A)$ for any $\tau\in (0,T)$. The function $U$ is a solution in the mild (integral) sense, i.e., 
$U(\tau) = e^{-A \tau} U_0 + \int_0^\tau e^{-A (\tau-s)} (F(s, U(s)) + h(s) ) \ud{s}$, $\tau\in[0,T]$.
\end{proposition}

Applying Propositions~\ref{prop-f} and \ref{semilinear_general_existence_result}, we can state  the following result which is a nontrivial generalization of the result shown by \v{S}ev\v{c}ovi\v{c} and Cruz \cite{CruzSevcovic2020} for $n=1$. 

\begin{theorem} \cite[Theorem 1]{vsevvcovivc2021multidimensional}
\label{semilinear_existence_result}
Suppose that the shift mapping $\xi=\xi(x,z)$ satisfies $\sup_{x\in\R} |\xi(x,z)|  \le C_0 |z|^\omega (1+ e^{D_0 |z|}), z\in\R^n$, for some constants $C_0>0, D_0\ge 0, \omega>0$. Assume $\nu$ is an admissible activity L\'evy measure with the shape parameters $\alpha, D$, and, either $\mu>0, D\in\R$, or $\mu=0, D>D_0\ge 0$. Assume $1/2\le \gamma < 1$ and $\gamma>(\alpha-n)/(2\omega)$, $n\geq 1$. Suppose that $g(\tau, x, u,\nabla_x u)$ is H\"older continuous in the $\tau$ variable and  Lipschitz continuous in the remaining variables, respectively. Assume $u_0\in X^\gamma$, and $T>0$. Then, there exists a unique mild solution $u$ to PIDE (\ref{PDE-u-general}) satisfying $u\in C([0,T],X^{\gamma})$.
\end{theorem}

\subsection{Maximal monotone operator technique for solving nonlinear parabolic equations}

This section presents the existence and uniqueness results of a fully nonlinear parabolic equation using the monotone operator technique. We consider the HJB equation arising from portfolio optimization selection, where the goal is to maximize the conditional expected value of the terminal utility of the portfolio. Such a fully nonlinear HJB equation presented in an abstract setting can be viewed as a PIDE in some sense. First, we employ the so-called Riccati transformation method to transform the fully nonlinear HJB equation into a quasilinear parabolic equation, which can be viewed as the porous media type of equation with source term. Then, we showed that the underlying operator is maximally monotone in some Sobolev spaces. Next, we employed the Banach's fixed point theorem and Fourier transform technique to obtain the existence and uniqueness of a solution to the general form of the transformed parabolic equation in an abstract setting in high-dimensional spaces. Furthermore, as a crucial requirement for solving the Cauchy problem, we obtain that the diffusion function to the quasilinear parabolic equation is globally Lipschitz continuous under some assumptions.

We consider the Cauchy problem for the nonlinear parabolic PDE of the following form: 
\begin{eqnarray}
\label{generalPDE}
&&\partial_{\tau}\varphi -\Delta \alpha(\tau, \varphi) = g_0(\tau, \varphi) + \nabla\cdot \bm{g}_1(\tau, \varphi),
    \\
&&	\varphi(\cdot, 0) =\varphi_{0},
\end{eqnarray}
where $\tau\in(0,T), x\in\mathbb{R}^d, d\ge 1$. The solution $\varphi=\varphi(x,\tau)$ to such a nonlinear parabolic equation is established in some Sobolev spaces in high-dimensional spaces (see, \cite{udeani2021application}). To achieve such results, we assumed that the diffusion function $\alpha=\alpha(x,\tau,\varphi)$ is globally Lipschitz continuous and strictly increasing in the $\varphi$-variable. An example of such a Lipschitz continuous function $\alpha(x,\tau,\varphi)$ is the value function of the following parametric optimization problem:
\begin{equation}
\alpha(x,\tau,\varphi) = \min_{ {\bm{\theta}} \in \triangle} 
\left(
-\mu(x,t,{\bm{\theta}}) +  \frac{\varphi}{2}\sigma(x,t,{\bm{\theta}})^2\right), \quad \tau\in(0,T), x\in\mathbb{R}^d, \varphi>\varphi_{min}\,,
\label{eq_alpha_general}
\end{equation}
where $\mu$ and $\sigma^2$ are given $C^1$ functions, and $\triangle\subset\mathbb{R}^n$ is a compact decision set. The properties of the value function depend on the structure of the decision set $\triangle$. It is smooth if $\triangle$ is a convex set; meanwhile, it can only be $C^{0,1}$ smooth if $\triangle$ is not connected.

\subsubsection{Existence and uniqueness of a solution to the Cauchy problem}

First, we define our underlying function spaces. Let $V\hookrightarrow H \hookrightarrow V^\prime$ be a Gelfand triple, where 
\[
H = L^2(\mathbb{R}^d) = \{ f:\mathbb{R}^d\to\mathbb{R}, \;\; \; \Vert f\Vert_{L^2}^2 = \int_{\mathbb{R}^d} |f(x)|^2 dx < \infty \}
\]
is a Hilbert space endowed with the inner product $(f,g)= \int_{\mathbb{R}^d} f(x) g(x) dx$. The Banach spaces $V$ and $V^\prime$ are defined as follows: 
\[
\quad V = H^1(\mathbb{R}^d), \quad V^\prime = H^{-1}(\mathbb{R}^d),
\]
where the Sobolev spaces $H^s(\mathbb{R}^d)$ are defined by means of the Fourier transform
\[
\hat f(\xi) = \frac{1}{(2\pi)^{d/2}} \int_{\mathbb{R}^d} e^{-i x\cdot\xi} f(x) dx, \quad \xi = (\xi_1, \xi_2, ..., \xi_d)^T\in\mathbb{R}^d,
\]
\[
H^s(\mathbb{R}^d) = \{ f:\mathbb{R}^d\to\mathbb{R}, (1+|\xi|^2)^{s/2} \hat f(\xi) \in L^2(\mathbb{R}^d) \}, \; s\in\mathbb{R}
\]
endowed with the norm $\Vert f\Vert_{H^s}^2 = \int_{\mathbb{R}^d} (1+|\xi|^2)^s |\hat f(\xi)|^2 d\xi$, and $ |\xi| = (\xi_1^2+ \cdots + \xi_d^2)^{1/2}$.
Let the linear operator $A : V\to V^\prime$ be defined as follows:
\[
A \psi = \psi  - \Delta \psi.
\]
It is worth nothing that the operator $A$ is self-adjoint in the Hilbert space $H=L^2(\mathbb{R}^d)$ with the following Fourier transform representation:
\[
\widehat{A\psi}(\xi) = (1+|\xi|^2)\hat\psi(\xi).
\]
The fractional powers of $A$ is defined by $\widehat{A^s\psi}(\xi) = (1+|\xi|^2)^s\hat\psi(\xi), \; s\in \mathbb{R}$. In particular, 
\[
\widehat{A^{\pm1/2}\psi}(\xi) = (1+|\xi|^2)^{\pm1/2}\hat\psi(\xi), \quad \]
and $A^{-1/2}$ is a self-adjoint operator in the Hilbert space $H=L^2(\mathbb{R}^d)$. Moreover, $A^{-1}= A^{-1/2} A^{-1/2}$.

In the sequel, we denote the duality pairing between the spaces $ V$ and $ V^\prime$ by $ \langle .,.\rangle$, i.e., the value of a functional $ F\in V^\prime$ at $u\in V$ is denoted by $ \langle F, u \rangle$. We have the following definitions.
\begin{definition}
\cite{udeani2021application, Barbu}  An operator (in general, nonlinear) $ B: V\to V^\prime $ is said to be
\begin{itemize}
    \item  [(i)] monotone if $$ \langle B(u) - B(v), u-v\rangle \geq 0, \; \forall \; u,v\in V,$$
    \item [(ii)] strongly monotone if there exists a constant $C>0$ such that
    $$ \langle B(u) - B(v), u-v\rangle \geq C\Vert u-v\Vert_V^2,  \; \forall \; u,v\in V,$$ 
    \item  [(iii)] hemicontinuous if for each $ u, v\in V$, the real-valued function $ t\mapsto B(u+tv)(v) $ is continuous.
\end{itemize}
\end{definition}

\begin{theorem}\cite{Barbu, Showalter}
\label{th:Showalter}
Let $V$ be a separable reflexive Banach space, dense, and continuous in a Hilbert space $H$, which is identified with its dual, so $V \hookrightarrow H\hookrightarrow V^\prime$. Let $p\geq 2$ and set $\mathcal{V} = L^p ((0,T); V).$ Assume a family of operators $\mathcal{A}(\tau,.): V\to V^\prime, 0\leq \tau< T$, is given such that
\begin{itemize}
\item [(i)] for each $\varphi \in V $, the function $\mathcal{A}(.,\varphi): [0,T]\to V^\prime $ is measurable,
\item [(ii)] for a.e $\tau\in [0,T]$, the operator $\mathcal{A}(\tau,.): V\to V^\prime$ is monotone, hemicontinuous, and bounded by
$\Vert \mathcal{A}(\tau,\varphi)\Vert \leq C ( \Vert \varphi\Vert^{p-1} + k(\tau)), \varphi\in V, 0\leq \tau< T,$ 
where $ k\in L^{p'}(0,T) $, 
\item [(iii)] and there exists $\lambda >0$ such that 
$\langle \mathcal{A}(\tau,\varphi), \varphi\rangle  \geq \lambda\Vert \varphi\Vert^{p} - k(\tau), \varphi\in V, 0\leq \tau< T.$
\end{itemize}
Then, for each $\hat f\in \mathcal{V^\prime}$ and $ \varphi_{0}\in H$, there exists a unique solution $\varphi\in\mathcal{V}$ of the Cauchy problem
\[
\partial_\tau\varphi(\tau) + \mathcal{A}(\tau,\varphi(\tau)) = \hat f(\tau) ~ \text{in}~ \mathcal{V^\prime},\ \  \varphi(0) =\varphi_{0}.
\]
\end{theorem}

\medskip
Consider the spaces $\mathcal{V} = L^2 ((0,T);V)$, $\mathcal{H} = L^{2}((0,T);H) $, and  $\mathcal{V^\prime} = L^2 ((0,T);V^\prime)$, i.e., $p=2$. Thus, we have that these spaces satify the Gelfand triple, i.e., $\mathcal{V}\hookrightarrow \mathcal{H} \hookrightarrow \mathcal{V^\prime}$, where $ \mathcal{H}$ is a Hilbert space endowed with the norm 
\[
\Vert \varphi\Vert^{2}_{\mathcal{H}} = \int_{0}^{T}\Vert \varphi(\tau)\Vert^{2}_{H}d\tau, \; \forall\varphi\in\mathcal{H}.
\]  
For a given value $\varphi_{min}$, we denote ${\mathcal D}= \mathbb{R}^d\times(0,T)\times (\varphi_{min},\infty)$. 

\begin{theorem}\cite[Theorem 2]{udeani2021application}
\label{th:alpha-existence}
Assume that the above settings on $H$ and $V$ hold. Let $ g_0, g_{1j}:[0,T]\times H\to H, j=1,\cdots, n,$ be globally Lipschitz continuous functions. Suppose $\alpha\in C^{0,1}(\mathcal{D})$ is such that there exist constants $\omega, L, L_0>0$ such that $0 < \omega\le  \alpha^\prime_\varphi(x,\tau,\varphi)\le L$, $|\nabla_x \alpha(x,\tau,\varphi)|\le p(x,\tau) + L_0  |\varphi|$, $\alpha(x,\tau,0)=h(x,\tau)$ for  a.e. $(x,\tau,\varphi)\in\mathcal{D}$ and $p, h\in L^{\infty}((0,T);H)$.
Then, for any $T>0 \; \text{and} \;\varphi_{0}\in H,$ there exists a unique solution $\varphi\in{\mathcal V}$ of the Cauchy problem 
\begin{equation}
	\partial_{\tau}\varphi + A\alpha(\cdot,\tau, \varphi) = g_0(\tau, \varphi) + \nabla \cdot \bm{g}_1(\tau, \varphi), \qquad 
	\varphi(0) =\varphi_{0}. 
	\label{equ:g_0,g_1}
\end{equation}

\end{theorem}
 \noindent We remark here that the above result and its proof are contained in our recent paper \cite[Theorem 2]{udeani2021application}.\\
\noindent P r o o f: Recall that $ H = L^{2}(\mathbb{R}^d) $ and $ V = H^{1}(\mathbb{R}^d)$, its dual space being $ V^\prime = H^{-1}(\mathbb{R}^d)$. Let the scalar products in $V$ and $ V^{'}$ be respectively defined by
\[
(f, g)_V =(A^{1/2}f, A^{1/2} g)_{H}=(Af, g)_{H}, \ 
(f, g)_{V^\prime} =(A^{-1/2}f, A^{-1/2} g)_{H}=(A^{-1}f, g)_{H}.
\] Let us define the operator $\mathcal{A}(\tau, \cdot) : V\to V^\prime$ by 
\[
\langle\mathcal{A}(\tau, \varphi), \psi\rangle = ( A^{-1}A\alpha(\cdot,\tau, \varphi), \psi)_{H} =(\alpha(\cdot, \tau, \varphi), \psi)_{H}.
\]
Therefore, we conclude that the mapping $\varphi\mapsto\alpha(\cdot,\tau,\varphi)$ maps $V$ into $V$ under the assumption made on the function $\alpha$. Indeed, if $\varphi\in V$ and $\eta=\alpha(\cdot,\tau,\varphi)$, then $\eta(x) = \alpha(x,\tau,\varphi(x)) - \alpha(x,\tau,0) + \alpha(x,\tau,0)$, and so
\[
|\eta(x)| \le (\max_\varphi \alpha^\prime_\varphi(x,\tau,\varphi)) |\varphi(x)|  + |h(x,\tau)| 
\le L |\varphi(x)| + |h(x,\tau)|. 
\]
Thus, $\int_{\mathbb{R}^d} |\eta(x)|^2 dx \le  2 \int_{\mathbb{R}^d} L^2 |\varphi(x)|^2  +  |h(x,\tau)|^2 dx \le 2L^2\Vert\varphi\Vert^2_H + 2 \Vert h(\cdot,\tau)\Vert^2_H$. Since $\nabla\eta(x) = \nabla_x\alpha(x,\tau, \varphi(x)) + \alpha^\prime_\varphi(x,\tau, \varphi(x))\nabla\varphi(x)$, we have
 \begin{equation*}
\begin{split}
 \Vert \eta\Vert_V^2=&\int_{\mathbb{R}^d} |\eta(x)|^2 + |\nabla\eta(x)|^2 dx 
\\ \leq & 2 \int_{\mathbb{R}^d} L^2 |\varphi(x)|^2  +  |h(x,\tau)|^2 dx + 2 \int_{\mathbb{R}^d} |p(x,\tau)|^2 + L_0^2|\varphi(x)|^2 dx  +  2\int_{\mathbb{R}^d} L^2|\nabla \varphi(x)|^2 dx \\ \leq & 2(L^2 \Vert \varphi \Vert_V^2 + \Vert h(\cdot, \tau)\Vert_H^2 + \Vert p(\cdot ,\tau)\Vert_H^2 + L_0^2 \Vert \varphi\Vert_H^2)<\infty,
\end{split}
\end{equation*}
because $p,h\in L^{\infty}((0,T);H)$. Consequently, $\eta\in V$, as claimed.  

Next, we show that the operator $\mathcal{A}$ is monotone in the space $ V^\prime$. According to (\ref{lipschitz}), we have $(\alpha(x,\tau,\varphi_{1}) - \alpha(x,\tau,\varphi_{2}))(\varphi_{1} - \varphi_{2})\geq \omega (\varphi_{1} - \varphi_{2})^2$, for any $\varphi_1, \varphi_2\ge\varphi_{min}, x\in\mathbb{R}, \tau\in[0,T]$. 
 \begin{equation*}
\begin{split}
\langle\mathcal{A}(\tau,\varphi_{1}) -\mathcal{A}(\tau,\varphi_{2}), \varphi_{1} -\varphi_{2}\rangle & = (\alpha (\cdot,\tau,\varphi_{1}) - \alpha (\cdot,\tau, \varphi_{2}), \varphi_{1} - \varphi_{2}) 
\\ & = \int_{\mathbb{R}^d}(\alpha (x,\tau,\varphi_{1}(x)) - \alpha (x,\tau, \varphi_{2}(x)))( \varphi_{1}(x) - \varphi_{2}(x)) dx 
\\ & \geq \int_{\mathbb{R}^d}\omega|\varphi_{1}(x) - \varphi_{2}(x)|^{2} dx  = \omega \Vert \varphi_{1} - \varphi_{2}\Vert^{2}_{H}.
\end{split}
\end{equation*}
This implies that the operator $ \mathcal{A}(\tau,\cdot)$ is strongly monotone. 

For a given $\tilde\varphi\in {\mathcal H}$, we have $\hat{f}\in\mathcal{V}^\prime$, where $\hat{f}(\tau) =  g_0(\tau,\tilde \varphi(\cdot, \tau)) + \nabla\cdot \bm{g}_1(\tau, \tilde\varphi(\cdot, \tau))$, because $g_0, g_{1j}: [0,T]\times H\to H $ are globally Lipschitz continuous, $H\hookrightarrow V^\prime$, and the operator $\nabla$ maps $H$ into $V^\prime$. The hemicontinuity, boundedness, and coercivity of the operator $\mathcal{A}$ follow from the assumption that the function $\alpha $ is globally Lipschitz continuous and strictly increasing.

Applying Theorem~\ref{th:Showalter}, we deduce the existence of a solution $ \varphi \in {\mathcal V} $ such that 
\begin{equation}
	\partial_{\tau}\varphi + \mathcal{A}(\tau,\varphi) = \hat{f}(\tau), \qquad \varphi_{0}\in H, 
	\label{eq:hatf}
\end{equation}
where $ \mathcal{A}(\tau,\varphi) = A\alpha(\cdot,\tau,\varphi)$. Next, we multiply (\ref{eq:hatf}) by $A^{-1}$ to obtain
\begin{equation}
\partial_{\tau}A^{-1}\varphi + \alpha(\cdot,\tau, \varphi) = f,
\label{eq:f}
\end{equation}
where $f=f(\tau,\tilde\varphi) = A^{-1}\hat{f}(\tau)$. For $\tau\in[0,T]$, we denote $\tilde{f}(\tilde\varphi) = A^{-1/2} \hat{f}(\tau) =  A^{-1/2} g_0(\tau,\tilde\varphi) + A^{-1/2}\sum_{j=1}^d \partial_{x_j}g_{1j}(\tau,\tilde\varphi)$. The Fourier transform of $\tilde f$ is defined by
\[
\widehat{\tilde{f}(\tilde\varphi)}(\xi) = \frac{1}{(1+ |\xi|^2)^{1/2}}\widehat{g_0(\tau, \tilde\varphi})(\xi) + \sum_{j=1}^d \frac{(-i\xi_j)}{(1+ |\xi|^2)^{1/2}}\widehat{g_{1j}(\tau, \tilde\varphi})(\xi).
\]
Let $\beta>0$ be the Lipschitz constant of the mappings $g_0, g_{1j}, j=1, \cdots, d$. Using Parseval's identity and Lipschitz continuity of $g_0, g_{1j}$ in $H$, we obtain, for $\tilde\varphi_1, \tilde\varphi_2\in\mathcal{H}$, 
\begin{equation*}
\begin{split}
\Vert\tilde{f}(\tilde\varphi_{1}) - \tilde{f}(\tilde\varphi_{2})\Vert^{2}_{H} & = \Vert\widehat{\tilde{f}(\tilde\varphi_{1})} - \widehat{\tilde{f}(\tilde\varphi_{2})}\Vert^{2}_{H} = \int_{\mathbb{R}^d} |\widehat{\tilde{f}(\tilde\varphi_{1})}(\xi)- \widehat{\tilde{f}(\tilde\varphi_{2})}(\xi)|^{2}d\xi \\& 
\le
2 \int_{\mathbb{R}^d} \frac{1}{1+|\xi|^{2}}|\widehat{g_0(\tau,\tilde\varphi_{1})}(\xi) - \widehat{g_0(\tau,\tilde\varphi_{2})}(\xi)|^{2} 
\\ &
\quad + \sum_{j=1}^d \frac{|\xi|^2}{1+|\xi|^{2}}|\widehat{g_{1j}(\tau,\tilde\varphi_{1})}(\xi) - \widehat{g_{1j}(\tau,\tilde\varphi_{2})}(\xi)|^{2}
d\xi \\ & 
\leq 2\Vert \widehat{g_0(\tau,\tilde\varphi_{1})} - \widehat{g_0(\tau,\varphi_{2})}\Vert^{2}_{H}  
+ 2\sum_{j=1}^d \Vert \widehat{g_{1j}(\tau,\tilde\varphi_{1})} - \widehat{g_{1j}(\tau,\varphi_{2})}\Vert^{2}_{H}
\\ &
= 2 \Vert g_0(\tau,\tilde\varphi_{1}) - g_0(\tau,\tilde\varphi_{2})\Vert^{2}_{H}
+ 2 \sum_{j=1}^d \Vert g_{1j}(\tau,\tilde\varphi_{1}) - g_{1j}(\tau,\tilde\varphi_{2})\Vert^{2}_{H}  
\\ &
\leq \tilde \beta^{2} \Vert \tilde\varphi_{1} - \tilde\varphi_{2}\Vert^{2}_{H},
\end{split}
\end{equation*} where $\tilde \beta^2= 2(1+d)\beta^2$. Hence, we obtain 
\begin{equation}
\Vert \tilde{f}(\tilde\varphi_{1}) - \tilde{f}(\tilde\varphi_{2})\Vert_{H} \leq \tilde\beta \Vert\tilde\varphi_{1} - \tilde\varphi_{2}\Vert_{H}.
	\label{eq:liphatf}
\end{equation}
\noindent 
Suppose $\varphi_{1}, \varphi_{2} \in \mathcal{H} $ are such that $\varphi_{1}=F(\tilde{\varphi_{1}})$ and $\varphi_{2} = F(\tilde{\varphi_{2}}).$ Here, the map  
$ F : \mathcal{H} \to \mathcal{H} $ is defined by $ \varphi = F(\tilde{\varphi}), $ where $ \varphi$ is a solution to the Cauchy problem
\begin{equation*}
\partial_{\tau}A^{-1}\varphi + \alpha(\cdot,\tau,\varphi) = f(\tau,\tilde{\varphi}), \qquad \varphi(0) =\varphi_{0}. 
\end{equation*}
Letting $ \varphi = \varphi_{1} -\varphi_{2} = F(\tilde{\varphi_{1}}) - F(\tilde{\varphi_{2}})$, we obtain
\begin{equation}
\partial_{\tau}A^{-1}(\varphi_{1} - \varphi_{2}) + \alpha(\cdot,\tau,\varphi_{1}) -\alpha(\cdot,\tau,\varphi_{2}) = f(\tilde{\varphi}_{1}) -f(\tilde{\varphi}_{2}).
	\label{eq:dsol}
\end{equation}
Next, multiplying (\ref{eq:dsol})  by $ \varphi_{1} -\varphi_{2}$ and taking the scalar product in the space $H$, we obtain
\begin{eqnarray}
 (\partial_{\tau}A^{-1}(\varphi_{1} - \varphi_{2}), \varphi_{1} - \varphi_{2}) &+& (\alpha(\cdot,\tau,\varphi_{1}) -\alpha(\cdot,\tau,\varphi_{2}), \varphi_{1} - \varphi_{2}) \nonumber 
\\
&=& ( f(\tau, \tilde\varphi_{1}) -f(\tau, \tilde{\varphi}_{2}), \varphi_{1} - \varphi_{2}).
\label{eq:scph}
\end{eqnarray}
Using (\ref{eq:liphatf}) and the fact that $ A^{-1/2}$ is self-adjoint in $H$, then (\ref{eq:scph}) gives

\begin{equation*}
\begin{split}
\frac{1}{2}\frac{d}{d\tau} & \Vert A^{-1/2}(\varphi_{1}-\varphi_{2})\Vert^{2}_{H} + \omega\Vert\varphi_{1}- \varphi_{2}\Vert^{2}_{H}
\\
& \leq \langle f(\tau, \tilde{\varphi}_{1})-f(\tau, \tilde{\varphi}_{2}), \varphi_{1}-\varphi_{2}\rangle  = \langle A^{1/2}(f(\tau, \tilde\varphi_{1}) - f(\tau, \tilde\varphi_{2})), A^{-1/2}(\varphi_{1} -\varphi_{2})\rangle \\& \leq \Vert A^{1/2}(f(\tau, \tilde\varphi_{1}) - f(\tau, \tilde\varphi_{2}))\Vert_{H} \Vert \varphi_{1} -\varphi_{2}\Vert_{V^\prime}  = \Vert \tilde{f}(\tilde\varphi_{1}) - \tilde{f}(\tilde\varphi_{2})\Vert_{H} \Vert \varphi_{1} -\varphi_{2}\Vert_{V^\prime} \\& \leq \tilde \beta \Vert\tilde{\varphi_{1}} - \tilde{\varphi_{2}}\Vert_{H} \Vert \varphi_{1} -\varphi_{2}\Vert_{V^\prime}.  
\end{split}
\end{equation*}

\noindent This implies 
\[
\frac{1}{2}\frac{d}{d\tau}\Vert \varphi_{1}-\varphi_{2}\Vert^{2}_{V^\prime} + \omega\Vert\varphi_{1}- \varphi_{2}\Vert^{2}_{H} \leq \tilde\beta \Vert\tilde{\varphi_{1}} - \tilde{\varphi_{2}}\Vert_{H} \Vert \varphi_{1} -\varphi_{2}\Vert_{V^\prime}.
\]
Then, integrating on a small time interval $ [0,T]$ from $ 0$ to $t$ and noting that $ \varphi_{1}(0) = \varphi_{2}(0) =\varphi_{0}$, we obtain

\begin{equation*}
\begin{split}
\frac{1}{2}\Vert \varphi_{1}(\tau)-\varphi_{2}(\tau)\Vert^{2}_{V^\prime} & + \omega \int_{0}^{\tau} \Vert\varphi_{1}(s)- \varphi_{2}(s)\Vert^{2}_{H}ds 
\\ & 
\leq \tilde\beta \int_{0}^{\tau}\Vert\tilde{\varphi_{1}}(s) - \tilde{\varphi_{2}}(s)\Vert_{H} \Vert \varphi_{1}(s) -\varphi_{2}(s)\Vert_{V^\prime}ds 
\\ & 
\leq \tilde\beta \max\limits_{\tau\in[0,T]} ~ \Vert\varphi_{1}(\tau) - \varphi_{2}(\tau)\Vert_{V^\prime}\int_{0}^{T}\Vert \tilde{\varphi_{1}}(\tau) - \tilde{\varphi_{2}}(\tau)\Vert_{H}d\tau.
\end{split}
\end{equation*}
Taking the maximum over $ \tau \in [0,T]$ and using the fact that for any $ a,b\in\mathbb{R}, ~ ab \leq \frac{1}{2} a^{2} + \frac{1}{2}b^{2}$, we obtain 
\begin{equation*}
\begin{split}
\frac{1}{2}(\max\limits_{\tau\in[0,T]} ~ &\Vert\varphi_{1}(\tau) - \varphi_{2}(\tau)\Vert_{V^\prime})^{2}  + \omega \int_{0}^{T} \Vert\varphi_{1}(\tau)- \varphi_{2}(\tau)\Vert^{2}_{H}d\tau \\& \leq \tilde\beta \max\limits_{\tau\in[0,T]} ~ \Vert\varphi_{1}(\tau) - \varphi_{2}(\tau)\Vert_{V^\prime}  \int_{0}^{T}\Vert \tilde{\varphi_{1}}(\tau) - \tilde{\varphi_{2}}(\tau)\Vert_{H}d\tau
\\& 
\leq \frac{1}{2}(\max\limits_{\tau\in[0,T]} ~ \Vert\varphi_{1}(\tau) - \varphi_{2}(\tau)\Vert_{V^\prime})^{2}  + \frac{\tilde\beta^{2}}{2} (\int_{0}^{T}\Vert \tilde{\varphi_{1}}(\tau) - \tilde{\varphi_{2}}(\tau)\Vert_{H}d\tau)^{2}.
\end{split}
\end{equation*}
Using the Cauchy--Schwartz inequality, we obtain 
$\omega \int_{0}^{T} \Vert\varphi_{1}(\tau)- \varphi_{2}(\tau)\Vert^{2}_{H}d\tau \leq \frac{\tilde\beta^{2}}{2} \int_{0}^{T}d\tau \int_{0}^{T}\Vert \tilde{\varphi_{1}}(\tau) - \tilde{\varphi_{2}}(\tau)\Vert^{2}_{H}d\tau = \frac{\tilde\beta^{2} T}{2}\int_{0}^{T}\Vert \tilde{\varphi_{1}}(\tau) - \tilde{\varphi_{2}}(\tau)\Vert^{2}_{H}d\tau.
$
This implies that 
\[
\Vert F(\tilde{\varphi_{1}}) - F(\tilde{\varphi_{2}})\Vert^{2}_{\mathcal{H}} \leq \frac{\tilde\beta^{2} T}{2\omega}\Vert \tilde{\varphi_{1}} -\tilde{\varphi_{2}}\Vert^{2}_{\mathcal{H}}. 
\]
Thus, for sufficiently small value of $T$ such that $\frac{\tilde\beta^{2} T}{2\omega} <1$, the operator $F$ is a contraction on the space $ \mathcal{H}$. Therefore, by the Banach fixed point theorem, $F$ has a unique fixed point in $\mathcal{H}$. It is worth noting that $\tilde\beta$ and $ \omega$ are given such that they are independent of $T$. If $T>0$ is arbitrary, then we can apply a simple continuation argument. In other words, if the solution exists in $(0,T_0)$ interval with $\frac{\tilde \beta^{2} T_0}{2 \omega} <1$, then starting from the initial condition $\varphi_0=\varphi(T_0/2)$, we can continue the solution $\varphi$ from the interval $(0,T_0)$ over the interval $(0, T_0)\cup (T_0/2, T_0/2 +T_0) \equiv (0, 3T_0/2)$. Continuing in this manner, we obtain the existence and uniqueness of a solution $\varphi\in{\mathcal H}$ defined on the time interval $(0,T)$. 

Finally, the solution belongs to the space $\mathcal V$ because the right-hand side, i.e., the function  $\hat f(\tau)=g_0(\tau, \varphi(\cdot, \tau)) + \nabla\cdot\bm{g}_1(\tau, \varphi(\cdot, \tau))$ belongs to ${\mathcal V}'$. Applying Theorem~\ref{th:Showalter} we conclude  $\varphi\in{\mathcal V}$, as claimed.
\hfill
$\diamondsuit$

\medskip

The following result demonstrates the absolute continuity and a-priori energy estimate property of the solution. Based on the assumption of the previous theorem, we have $\alpha(\cdot, 0), g_0(\cdot, 0), g_{1j}(\cdot, 0) \in \mathcal{H}$. Here, the space $\mathcal{X} = L^{\infty} ((0,T);V^\prime)$ is endowed with the norm 
\[
\Vert \varphi\Vert^{2}_{\mathcal{X}} = \sup_{\tau\in[0,T]}\Vert \varphi(\tau)\Vert^{2}_{V^\prime}, \; \forall\varphi\in\mathcal{X}.
\] 
Again, the following result and its are contained in our recent paper \cite[Theorem 3]{udeani2021application}.
\begin{theorem}\cite[Theorem 3]{udeani2021application}
\label{cor:abs-continuous}
Suppose that the functions $\alpha, g_0, g_{1j}$ satisfy the assumptions of Theorem~\ref{th:alpha-existence}. Then, the unique solution $\varphi \in\mathcal{V}$ to the Cauchy problem $(\ref{eq:hatf})$ is absolutely continuous, i.e., $ \varphi\in C([0,T];H)$. Moreover, there exists a constant $\tilde{C}>0$, such that the unique solution satisfies the following inequality:
\begin{equation}
\Vert \varphi \Vert^2_{\mathcal{X}} 
+ \Vert \varphi\Vert^2_{\mathcal{H}} 
\leq \tilde{C} \bigl(
\Vert \varphi_0\Vert^2_{V^\prime}
+ \Vert \alpha(\cdot, 0)\Vert^2_{\mathcal{H}}
+ \Vert g_0(\cdot, 0)\Vert^2_{\mathcal{H}}
+ \sum_{j=1}^d \Vert g_{1j}(\cdot, 0)\Vert^2_{\mathcal{H}}
\bigr).
\label{prior_est}
\end{equation}
\end{theorem}

\begin{proof}
Since $\hat{f}\in\mathcal{V}^\prime$, where $\hat{f} = g_0 + \nabla \cdot \bm{g}_{1}$ and $\mathcal{A}(\tau,\varphi)\in\mathcal{V}^\prime$, then $\partial_{\tau}\varphi\in\mathcal{V}^\prime $. Therefore, for each $\varphi_0 \in H$, we have $\varphi\in W$, where $W$ is the Banach space  $W=\{\varphi,  \varphi\in\mathcal{V}, \partial_\tau\varphi \in \mathcal{V}'\}$. According to \cite[Proposition \;1.2]{Showalter}, we have $W \hookrightarrow C([0,T];H)$.
Hence, the unique solution $\varphi$ to the Cauchy problem  (\ref{th:alpha-existence}) belongs to the space $ C([0,T];H)$, as claimed.

\medskip

Next, we show that the unique solution satisfies a-priori energy estimate (\ref{prior_est}). Let $\varphi$ be a unique solution to the Cauchy problem (\ref{equ:g_0,g_1}). Multiply (\ref{eq:f}) by $\varphi$ and take the scalar product in $H$ to obtain
\begin{equation}
(\partial_{\tau}A^{-1}\varphi, \varphi)_{H} + (\alpha(\cdot,\tau,\varphi), \varphi)_{H} =(A^{-1}g_0(\tau, \varphi) + A^{-1} \nabla\cdot \bm{g}_1(\tau, \varphi), \varphi).
\label{equ:scalarH}
\end{equation}
Using the Lipschitz continuity of $g_0, \bm{g}_1$, and strong monotonicity of $\alpha$,  we obtain
\begin{eqnarray*}
\frac{1}{2}\frac{d}{d\tau} \Vert \varphi\Vert^2_{V'} + \omega \Vert \varphi\Vert^2_{H} 
&=& 
(\partial_{\tau}A^{-1}\varphi, \varphi) +  \omega \Vert \varphi\Vert^2_{H} 
\\
&\le&
(\partial_{\tau}A^{-1}\varphi, \varphi) + (\alpha(\cdot,\varphi) - \alpha(\cdot, 0), \varphi)  
\\
&=& 
(A^{-1} ( g_0(\cdot, \varphi) + \nabla\cdot \bm{g}_1(\tau, \varphi) ) - \alpha(\cdot, 0), \varphi)
\\
&=& 
(A^{-1} ( g_0(\cdot, \varphi) - g_0(\cdot, 0)  + \nabla\cdot \bm{g}_1(\cdot, \varphi) - \nabla\cdot \bm{g}_1(\cdot, 0) ), \varphi)
\\
&& + (A^{-1} ( g_0(\cdot, 0)  + \nabla\cdot \bm{g}_1(\cdot, 0) ), \varphi) - (\alpha(\cdot, 0), \varphi)
\\
&=& (A^{-1/2} ( g_0(\cdot, \varphi) - g_0(\cdot, 0)  + \nabla\cdot \bm{g}_1(\cdot, \varphi) - \nabla\cdot \bm{g}_1(\cdot, 0) ), A^{-1/2} \varphi)
\\
&& + (A^{-1/2} ( g_0(\cdot, 0)  + \nabla\cdot \bm{g}_1(\cdot, 0) ), A^{-1/2} \varphi) - (\alpha(\cdot, 0), \varphi)
\\
&\le& 
\beta (1+d) \Vert \varphi\Vert_H \Vert\varphi\Vert_{V'}
+ \Vert A^{-1/2} ( g_0(\cdot, 0)  + \nabla\cdot \bm{g}_1(\cdot, 0) )\Vert_H \Vert\varphi\Vert_{V'} 
\\
&& + \Vert\alpha(\cdot,0)\Vert_{H}\Vert\varphi\Vert_{H}
\\
&\le& 
\frac{\omega}{4} \Vert \varphi\Vert^2 _H + \frac{\beta^2 (1+d)^2}{\omega} \Vert\varphi\Vert^2_{V'}
+ \frac{1}{2}\Vert A^{-1/2} ( g_0(\cdot, 0)  + \nabla\cdot \bm{g}_1(\cdot, 0) )\Vert^2_H  
\\&& 
+ \frac{1}{2}\Vert\varphi\Vert^2_{V'}  
+ \frac{1}{\omega}\Vert\alpha(\cdot,0)\Vert^2_{H} + \frac{\omega}{4}\Vert\varphi\Vert^2_{H}.
\end{eqnarray*}
Hence, there exist constants $C_0, C_1>0$ such that 
\begin{eqnarray*}
\frac{d}{d\tau} \Vert \varphi\Vert^2_{V'} + \omega \Vert \varphi\Vert^2_{H} 
&\le& 
 C_1 \Vert\varphi\Vert^2_{V'} + C_0 \bigl( \Vert  g_0(\cdot, 0)\Vert^2_H + \sum_{j=1}^d  \Vert  g_{1j}(\cdot, 0)\Vert^2_H  +  \Vert \alpha(\cdot, 0)\Vert^2_H \bigr).
\end{eqnarray*}
Solving the differential inequality $y'(\tau) \le C_1 y(\tau) + r(\tau) $, where $y(\tau) = \Vert\varphi(\cdot,\tau)\Vert^2_{V'}$ and $r(\tau) = C_0\bigl( \Vert  g_0(\cdot, \tau, 0)\Vert^2_H + \sum_{j=1}^d  \Vert  g_{1j}(\cdot, \tau, 0)\Vert^2_H +  \Vert \alpha(\cdot,\tau, 0)\Vert^2_H \bigr)$, we obtain 
\[
y(\tau) \le e^{C_1 T} \bigl(y(0) + \int_0^T r(s)ds \bigr),
\]
and the proof of the Theorem follows. 
\end{proof}

\section{Applications to Option Pricing}

The classical linear Black--Scholes model and its multidimensional generalizations have been widely used in the financial market analysis. It is well-known that the price $V=V(t,S)$ of an option on an underlying asset price $S$ at time $t\in[0,T]$ can be obtained as a solution to the linear Black--Scholes parabolic equation of the form (\ref{eqBS}). Generally, th underlying asset price is assumed to follow the geometric Brownian motion $dS/S = \mu dt + \sigma dW$. Here, $\{W_t, t\ge 0\}$ is the standard Wiener process. The terminal condition $\Phi(S)$ represents the payoff diagram at maturity $t = T$, $\Phi(S)= (S-K)^+$ (call option case) or $\Phi(S)= (K-S)^+$ (put option case). 

For the multidimensional case, where the option price $V(t, S_1, \cdots, S_n)$ depends on the vector of $n$ underlying stochastic assets $S=(S_1, \cdots, S_n)$ with the volatilities $\sigma_i$ and mutual correlations $\varrho_{ij}, i,j=1,\cdots, n$, the Black--Scholes pricing equation can be expressed as follows:
\begin{equation}
\frac{\partial V}{\partial t} + \frac{1}{2} \sum_{i = 1}^n
\sum_{j = 1}^n \rho_{ij} \sigma_i \sigma_j S_i S_j
\frac{\partial^2 V}{\partial S_i \partial S_j} + r \sum_{i = 1}^n
S_i \frac{\partial V}{\partial S_i} - rV =0,\quad V(T,S)=\Phi(S).
\label{index-bs}
\end{equation}
Equations (\ref{eqBS}) and (\ref{index-bs}) can be transformed into equation (\ref{PDE-u}) defined on the whole space $\R^n$ (c.f.,  \v{S}ev\v{c}ovi\v{c}, Stehl\'{\i}kov\'a, Mikula \cite[Chapter 4, Section 5]{SSMbook}).

According to stock markets observations, the models (\ref{eqBS}) and (\ref{index-bs}) were derived under some restrictive assumptions, e.g., completeness and frictionless of the financial market, perfect replication of a portfolio and its liquidity, and absence of transaction costs. However, these assumptions are often violated in financial markets. In the past few decades, several attempts have been made to investigate the effects of nontrivial transaction costs \cite{NBS5,NBS7, leland1985option,sevcoviczitnanska}. For instance, Sch\"onbucher and Willmott \cite{NBS13}, Frey and Patie \cite{NBS11}, Frey and Stremme \cite{NBS10} investigated the feedback and illiquid market effects due to large traders choosing given stock-trading strategies. Janda\v{c}ka and \v{S}ev\v{c}ovi\v{c} recently investigated the effects of the risk arising from an unprotected portfolio. Barles and Soner \cite{barles} analyzed the option pricing models based on utility maximization. The common feature of these generalizations of the linear Black--Scholes equation (\ref{eqBS}) is that the constant volatility $\sigma$ is replaced by a nonlinear function depending on the second derivative $\partial _S^2 V$ of the option price $V$. Among these generalizations, Frey and Stremme \cite{NBS1} derived a nonlinear Black--Scholes model by assuming that the underlying asset dynamics takes into account the presence of feedback effects due to the influence of a large trader choosing particular stock-trading strategy (see also \cite{NBS13, NBS11, Frey98}). 

Recently, Cruz and \v{S}ev\v{c}ovi\v{c} \cite{NBS19} generalized the Black--Scholes equation in two important directions. First, they employed the ideas of Frey and Stremme \cite{NBS1} to incorporate the effect of a large trader into the model. Second, they relaxed the assumption on liquidity of market by assuming that the underlying asset price follows a L\'evy stochastic process with jumps to obtain the following nonlinear PIDE: 

\begin{eqnarray}
0 &=&\frac{\partial V}{\partial t}+\frac{1}{2}\frac{\sigma^2 S^2}{\left(1-\varrho S\partial_S \phi\right)^{2}} \frac{\partial^2 V}{\partial S^2 } +r S\frac{\partial V}{\partial S}-rV
\nonumber
\\
&+&\int_{\mathbb{R}} \left( V(t,S+H)-V(t,S)-H \frac{\partial V}{\partial S}\right) \nu(\ud z),
\label{nonlinearPIDE_one-intro-nonlin}
\end{eqnarray}
where the shift function $H=H(\phi,S,z)$ depends on the large investor stock-trading strategy function $\phi=\phi(t,S)$. Moreover, this shift function is a solution to the following implicit algebraic equation:
\begin{equation}
H =\rho S ( \phi(t, S+H)-\phi(t,S) ) + S(e^{z}-1).
\label{Hfunction}
\end{equation}
The large trader strategy function $\phi$ may depend on the derivative $\partial_S V$ of the option price $V$, e.g., $\phi(t,S)=\partial_S V(t,S)$. However, in our application, we assume the trading strategy function $\phi(t,S)$ is prescribed and globally H\"older continuous. Next, we present the analysis of this equation depending the behavior of the parameter $\rho=0$.

If $\rho=0$, then $H =S(e^{z}-1)$. Thus, equation \eqref{nonlinearPIDE_one-intro-nonlin} can be reduced to a linear PIDE of the form (\ref{PDE-u}) in the one-dimensional space ($n=1$). This is obtained using the standard transformation $\tau=T-t,x=\ln(\frac{S}{K})$ and setting $V(t,S)=e^{-r\tau} u(\tau,x)$.

However, if $\rho >0$, then \eqref{nonlinearPIDE_one-intro-nonlin} can be transformed into a nonlinear parabolic PIDE. Indeed, suppose that the transformed large trader stock-trading strategy $\psi(\tau,x)=\phi(t,S)$. Then, $V(t,S)$ solves equation \eqref{nonlinearPIDE_one-intro-nonlin} if and only if the transformed function $u(\tau,x)$ is a solution to the following nonlinear parabolic equation:
\begin{eqnarray}
\frac{\partial u}{\partial \tau}&=&\frac{\sigma^2}{2}\frac{1}{(1-\rho \partial_x\psi)^2}
\frac{\partial^2 u}{\partial^2 x}
+\left(r-\frac{\sigma^2}{2}\frac{1}{(1-\rho \partial_x\psi)^2} -  \delta(\tau, x) \right)\frac{\partial u}{\partial x}
\nonumber
\\
&+&\int_{\mathbb{R}} \left( u(\tau,x+\xi)-u(\tau,x)- \xi \frac{\partial u}{\partial x}(\tau,x) \right) \nu(\ud z), 
\quad u(0,x) = \Phi(K e^{x}),
\label{nonlinearPIDEaltsimplified}
\end{eqnarray}
$\tau\in [0,T], x\in \mathbb{R}$. The shift function $\xi(\tau, x, z)$ is a solution to the following algebraic equation: 
\begin{equation}
e^\xi  = e^z + \rho (\psi(\tau, x+\xi)-\psi(\tau, x)),
\label{xifunction}
\end{equation}
and $\delta(\tau, x) =\int_{\mathbb{R}} (e^\xi -1 -\xi ) \nu(\ud z) 
=\int_{\mathbb{R}} (e^z -1 - \xi +\rho  (\psi(\tau, x+\xi)-\psi(\tau, x)) ) \nu(\ud z)$.

For small values of $0<\rho\ll 1$, we can construct the first order asymptotic expansion $\xi(\tau, x, z) = \xi_0(\tau, x, z) + \rho \xi_1(\tau, x, z)$. For $\rho=0$, we obtain $\xi_0(\tau, x, z)=z$. Hence,
\[
e^{z+\rho\xi_1}  = e^z + \rho (\psi(\tau, x+z+ \rho \xi_1)-\psi(\tau, x)).
\]
Taking the first derivative of the above implicit equation with respect to $\rho$ and evaluating it at the origin $\rho=0$, we obtain $\xi_1=e^{-z}(\psi(\tau, x+z)-\psi(\tau, x))$, i.e.,  
\begin{equation}
\xi(\tau, x, z) = z + \rho e^{-z} (\psi(\tau, x+z)-\psi(\tau, x)).
\label{xifunctionexplicit}
\end{equation}

Consequently, we obtain the following lemma. 

\begin{lemma}\cite[Lemma 1]{vsevvcovivc2021multidimensional}
\label{holdercontinuityxi}
Assume that the stock-trading strategy $\phi=\phi(t,S)$ is a globally $\omega$-H\"older continuous function, $0<\omega\le 1$. Then, the transformed function $\psi(\tau,x)=\phi(t,S)$ is $\omega$-H\"older continuous, and the first order asymptotic expansion $\xi(\tau, x, z)$ of the nonlinear algebraic equation (\ref{xifunction}) is $\omega$-H\"older continuous in all variables. Furthermore, there exists a constant $C_0>0$ such that  $\sup_{\tau,x}|\xi(\tau,x,z)| \le C_0 |z|^\omega(1+e^{|z|})$ for any $z\in\R$.
\end{lemma}

\subsection{Linearization of PIDE}

In what follows, we consider a simplified linear approximation of  (\ref{nonlinearPIDE_one-intro-nonlin}) by setting $\rho=0$ in the diffusion function, but we keep the shift function $H$ depending on the parameter $\rho$. Then, the transformed Cauchy problem for the solution $u$  with the first order approximation of the shift function $\xi$ is given as follows:
\begin{eqnarray}
\frac{\partial u}{\partial \tau}&=&\frac{\sigma^2}{2} \frac{\partial^2 u}{\partial^2 x}
+\left(r-\frac{\sigma^2}{2} + \delta(\tau, x) \right)\frac{\partial u}{\partial x}
\nonumber 
\\
&&+\int_{\mathbb{R}} \left( u(\tau,x+\xi)-u(\tau,x)- \xi \frac{\partial u}{\partial x}(\tau,x) \right) \nu(\ud z),
\label{transformedeq}
\end{eqnarray}
$\tau\in [0,T], x\in \mathbb{R}$, where $\xi(\tau, x, z) = z + \rho (\psi(\tau, x+z))-\psi(\tau, x))$. 

Note that the call/put option payoff functions  $\Phi(S)=\Phi(K e^x)=(S-K)^+ = K (e^x-1)^+$ / $\Phi(S)=\Phi(K e^x)=(K-S)^+=K(1-e^x)^+$ do not belong to the Banach space $X^\gamma$. According to \cite{CruzSevcovic2020}, the procedure on how to overcome this problem and formulate existence and uniqueness of a solution to the PIDE  (\ref{transformedeq}) is based on shifting the solution $u$ by $u^{BS}$. Here, $u^{BS}(\tau,x) =e^{r\tau} V^{BS}(T-\tau, Ke^x)$ is an explicitly given solution to the linear Black--Scholes equation without the PIDE part. In other words, $u^{BS}$ solves the following linear parabolic equation:
\begin{equation}
\frac{\partial u^{BS}}{\partial \tau} - \frac{\sigma^2}{2} \frac{\partial^2 u^{BS}}{\partial x^2} 
-  \left(r-\frac{\sigma^2}{2}\right)\frac{\partial u^{BS}}{\partial x} =0,
\quad u^{BS}(0,x)=\Phi(K e^{x}), \ \tau\in(0,T), x\in \mathbb{R}.
\label{PDE-uBS}
\end{equation}
Recall that $u^{BS}(\tau,x) = K e^{x+r \tau} N(d_1) - K N(d_2)$ (call option case), where $d_{1,2} = ( x+ (r\pm\sigma^2/2)\tau) /(\sigma\sqrt{\tau})$ (c.f.,  \cite{NBS5, SSMbook}). Here,  $N(d)=\frac{1}{\sqrt{2\pi}}\int_{-\infty}^d e^{-\xi^2/2} \ud \xi$ is the cumulative density function of the normal distribution. 
The next result and its proof are based on the recent paper by \v{S}ev\v{c}ovi\v{c} and Udeani \cite{vsevvcovivc2021multidimensional}
\begin{theorem}\cite[Theorem 2]{vsevvcovivc2021multidimensional}
\label{existence_linear_PIDE}
Assume the transformed stock-trading strategy function $\psi(\tau,x)$ is globally $\omega$-H\"older continuous in both variables. Suppose that $\nu$ is a L\'evy  measure with the shape parameters $\alpha<3, D\in\R$, where either $\mu>0$, or $\mu=0$ and $D>1$. Let $X^\gamma$ be the space of Bessel potentials space ${\mathscr L}^p_{2\gamma}(\mathbb{R})$, where  $\frac{\alpha-1}{2\omega}<\gamma < \frac{p+1}{2p}$ and  $\frac12 \le \gamma<1$. Let $T>0$. Then, the linear PIDE (\ref{transformedeq}) has a unique mild solution $u$ with the property that the difference $U=u-u^{BS}$ belongs to the space $C([0,T],X^{\gamma})$. 
\end{theorem}

\begin{proof}
We first outline the idea of the proof. The initial condition $u(0,\cdot)\not\in X^\gamma$ because of two reasons. It is not smooth for $x=0$, and it grows exponentially for $x\to\infty$ (call option) or $x\to-\infty$ (put option). The shift function $U=u-u^{BS}$ satisfies $U(0,\cdot)\equiv 0$, and so the initial condition $U(0,\cdot)$ belongs to $X^\gamma$. However, the shift function $u^{BS}$ enters the governing PIDE as it includes the term $f(u^{BS}(\tau,\cdot))$ in the right-hand side. Since $u^{BS}(0,x)$ is not sufficiently smooth for $x=0$, the shift term $f(u^{BS}(\tau,\cdot))$ is singular for $\tau\to 0^+$. Following the ideas of \cite{CruzSevcovic2020}, for the shift term $f(u^{BS}(\tau,\cdot))$, we can provide H\"older estimates, which are sufficient for proving the main result of this theorem (c.f.,  \cite[Lemma 4.1]{CruzSevcovic2020}). Furthermore, the exponential growth of the function $u^{BS}$ will be overcome since $\tilde f(e^x) = 0$, where $\tilde f(u) = f(u) - \delta(\tau,\cdot) \partial_x u$, i.e., 
\[
\tilde f(u)(x) =
\int_{\mathbb{R}}\left( u(x+\xi)-u(x)- (e^\xi -1) \partial_x u(x)\, \right) \nu(\ud z).
\]

Next, we present more details of the proof. The function $u^{BS}$ solves the linear PDE (\ref{PDE-uBS}). Thus, the difference $U=u-u^{BS}$ of the solution $u$ to (\ref{transformedeq}) and $u^{BS}$ satisfies the PIDE with the right-hand side:
\begin{eqnarray*}
\frac{\partial U}{\partial \tau}
&=&\frac{\sigma^2}{2} \frac{\partial^2 U}{\partial x^2} 
+ \left(r-\frac{\sigma^2}{2} - \delta(\tau,x)\right)\frac{\partial U}{\partial x} + f(U) + f(u^{BS}) - \delta(\tau,x) \frac{\partial u^{BS}}{\partial x}
\\
&=&\frac{\sigma^2}{2} \frac{\partial^2 U}{\partial x^2} + f(U) + g(\tau, x, \partial_x U) + h(\tau,\cdot),
\end{eqnarray*}
$U(0,x)= 0, \ x\in \mathbb{R}, \tau\in(0,T)$. Here $g(\tau, x, \partial_x U) = (r-\sigma^2/2 - \delta(\tau,x) )\partial_x U$, and $h(\tau,\cdot)=\tilde f(u^{BS}(\tau,\cdot))$. According to Proposition~\ref{prop-f}, $f:X^\gamma\to X$ is a bounded linear mapping. Consequently, it is  Lipschitz continuous, provided that $1/2\le \gamma<1$ and $\gamma>(\alpha-1)/(2\omega)$. Clearly, $\tilde f(e^x)=0$. Hence,
\[
\tilde f(u^{BS}) = \tilde f(u^{BS} - K e^{r\tau+x}), \quad \hbox{and}\ \ 
\partial_\tau \tilde f(u^{BS}) = \tilde f(\partial_\tau(u^{BS} - K e^{r\tau+x})).
\]
Now, it follows from \cite[Lemma 4.1]{CruzSevcovic2020} that the following estimate holds true: 
\[
\Vert  h(\tau_1, \cdot) - h(\tau_2, \cdot) \Vert_{L^p} 
= \Vert \tilde f(u^{BS}(\tau_1, \cdot)) - \tilde f(u^{BS}(\tau_2, \cdot)) \Vert_{L^p} 
\le C_0 |\tau_1-\tau_2|^{-\gamma +\frac{p+1}{2p}}, 
\]
\[
\Vert  h(\tau, \cdot)  \Vert_{L^p} 
= \Vert \tilde f(u^{BS})(\tau, \cdot)) \Vert_{L^p} 
\le C_0 |\tau^{-(2\gamma-1)\left(\frac{1}{2} - \frac{1}{2p}\right)}, 
\]
for any $0<\tau_1,\tau_2,\tau\le T$. The function $h:[0,T]\to X\equiv L^p(\mathbb{R})$ is $((p+1)/(2p)-\gamma)$-H\"older continuous because  $\gamma<\frac{p+1}{2p}$. Moreover, 
\[
\int_0^T \Vert h(\tau, \cdot ) \Vert_{L^p} d\tau=
\int_0^T \Vert \tilde f(u^{BS}(\tau, \cdot)) \Vert_{L^p}d\tau 
\le C_0 \int_0^T\tau^{-(2\gamma-1)\left(\frac{1}{2} - \frac{1}{2p}\right)}d\tau <\infty,
\]
because $(2\gamma-1)\left(\frac{1}{2} - \frac{1}{2p}\right)<1$. Recall that the crucial part of the proof of \cite[Lemma 4.1]{CruzSevcovic2020} was based on the estimates:
\[
\Vert \tilde f(u^{BS}(\tau, \cdot)) \Vert_{L^p} 
\le C_0 \Vert v(\tau,\cdot) \Vert_{X^{\gamma-1/2}},
\quad\text{and}\ 
\Vert \partial_\tau \tilde f(u^{BS}(\tau, \cdot)) \Vert_{L^p} 
\le C_0 \Vert \partial_\tau v(\tau,\cdot) \Vert_{X^{\gamma-1/2}}, 
\]
where $v(\tau,x)= \partial_x \left(u^{BS}(\tau,x) - K e^{r\tau+x}\right)
= K e^{r\tau+x}( N(d_1(\tau,x)) -1)$. This estimate is satisfied because of Proposition~\ref{prop-f} under the assumptions made on $\gamma$. The proof for the case of a put option is similar. The final estimate on the H\"older continuity of the mapping $h$ follows from careful estimates of the solution $u^{BS}$ derived in the proof of \cite[Lemma 4.1]{CruzSevcovic2020}. Now, the proof follows from Theorem~\ref{semilinear_existence_result} and Proposition~\ref{semilinear_general_existence_result}. 
\end{proof}

\section{Feedback effects under jump-diffusion asset price dynamics}

Suppose that a large trader uses a stock-holding strategy $\alpha_{t}$ and $S_{t}$ is a cadlag process (right continuous with limits to the left). In what follows, we will identify $S_{t}$ with $S_{t^-}$. We assume $S_t$ has the following  dynamics:
\begin{eqnarray}
&&\ud S_{t}=\mu S_{t}\ud t+\sigma S_{t} \ud W_{t}+\rho S_{t} \ud \alpha_{t}+ \int_{\mathbb{R}} S_{t}(e^{x}-1) J_{X}(\ud t, \ud x),
\label{Sdynamicsimplicit}
\end{eqnarray}
which can seen as a perturbation of the classical jump-diffusion model. For instance, if a large trader does not trade, then $\alpha_{t}=0$ or the market liquidity parameter $\rho$ is set to zero, then the stock price $S_t$ follows the classical jump-diffusion model.

We will assume the following structural hypothesis in this chapter:
\begin{assumption}\label{assumptionone}\cite[Assumption 1]{}
Assume the trading strategy $\alpha_{t}=\phi(t,S_{t})$ and the parameter $\rho\ge 0$ satisfy $\rho L<1$, where $L=\sup_{S>0} |S\frac{\partial \phi}{\partial S}|$.
\end{assumption}

Next, we show an explicit formula for the dynamics of $S_{t}$ satisfying (\ref{Sdynamicsimplicit}) under certain regularity assumptions made on the stock-holding function $\phi(t,S)$.
\begin{proposition}
Suppose that the stock-holding strategy $\alpha_{t}=\phi(t,S_{t})$ satisfies Assumption \ref{assumptionone}, where $\phi\in C^{1,2}([0,T]\times \mathbb{R^{+}})$. If the process $S_{t}, t\ge 0,$ satisfies the implicit stochastic equation  (\ref{Sdynamicsimplicit}), then the process $S_{t}$ is driven by the following SDE:
\begin{eqnarray}
\ud S_{t}=b(t,S_{t})S_{t}\ud t+v(t,S_{t})S_{t} \ud W_{t}+\int_{\mathbb{R}} H(t,x,S_{t}) J_{X}(\ud t, \ud x)
\label{Sdynamicsexplicit},
\end{eqnarray}
where
\begin{eqnarray}
&&b(t,S)=\frac{1}{1-\rho S\frac{\partial \phi}{\partial S}(t,S)}\left(\mu+\rho\left(\frac{\partial \phi}{\partial t}+\frac{1}{2}v(t,S)^2S^2 \frac{\partial^2 \phi}{\partial S^2}\right)\right),
\label{coeffcondb}
\\
&&v(t,S)=\frac{\sigma}{1-\rho S\frac{\partial \phi}{\partial S}(t,S)},
\label{coeffcondv}
\\
&&H(t,x,S)=S_{}(e^{x}-1)+\rho S\left[\phi(t,S+H(t,x,S))-\phi(t,S)\right].
\label{coeffcondone}
\end{eqnarray}

\end{proposition}
\begin{proof}
First, the SDE \eqref{Sdynamicsexplicit} can be expressed for $S_{t}$ as follows:
\begin{eqnarray}
\ud S_{t}&=&\left(b(t,S_{t}) S_{t} +\int_{|x|<1} H(t,x,S_{t})\nu(\ud x)\right)\ud t + v(t,S_{t}) S_{t} \ud W_{t}\nonumber
\\
&&+\int_{|x|\ge 1} H(t,x,S_{t}) J_{X}(\ud t, \ud x)+\int_{|x|<1} H(t,x,S_{t}) \tilde{J}_{X}(\ud t, \ud x).
\nonumber
\end{eqnarray}
Since the function $\phi(t,S)$ is smooth, by applying It\^{o} formula \eqref{itolemma} to the process $\phi(t,S_{t})$, we obtain
\begin{eqnarray}
\ud \alpha_{t}&=&\left(\frac{\partial \phi}{\partial t}+\frac{1}{2}v(t,S_{t})^2S_{t}^2 \frac{\partial^2 \phi}{\partial S^2}\right)\ud t+\frac{\partial \phi}{\partial S}\ud S_{t}
\\
&&+\int_{\mathbb{R}} \phi(t,S_{t}+H(t,x,S_{t}))-\phi(t,S_{t})-H(t,x,S_{t}) \frac{\partial \phi}{\partial S}(t,S_t) J_{X}(\ud t, \ud x).
\nonumber
\end{eqnarray}
Substituting the differential  $\ud \alpha_{t}$ into \eqref{Sdynamicsimplicit}, we obtain
\begin{eqnarray}
\ud S_{t}&=&\mu S_{t}\ud t+\sigma S_{t} \ud W_{t}+\int_{\mathbb{R}} S_{t}(e^{x}-1) J_{X}(\ud t, \ud x)+\rho S_{t}\frac{\partial \phi}{\partial S}\ud S_{t}\nonumber
\\
&&+\rho S_{t} \left(\frac{\partial \phi}{\partial t}+\frac{1}{2}v(t,S_{t})^2S_{t}^2 \frac{\partial^2 \phi}{\partial S^2}\right)\ud t\label{Sdynamicsexplicitone}
\\
&&+\rho S_{t}\int_{\mathbb{R}} \phi(t,S_{t}+H(t,x,S_{t}))-\phi(t,S_{t})- H(t,x,S_{t}) \frac{\partial \phi}{\partial S}(t,S_t) J_{X}(\ud t, \ud x).
\nonumber
\end{eqnarray}
Next, rearrange terms in (\ref{Sdynamicsexplicitone}) to obtain 
\begin{eqnarray}
&&(1- \rho S_{t}\frac{\partial \phi}{\partial S}(t,S_{t}))\ud S_{t}=(\mu S_{t}+\rho S_{t}(\frac{\partial \phi}{\partial t}+\frac{1}{2}v(t,S_{t})^2 S_{t}^2 \frac{\partial^2 \phi}{\partial S^2}))\ud t 
\nonumber
\\
&&+\sigma S_{t} \ud W_{t}-\rho S_{t}\int_{\mathbb{R}} H(t,x,S_{t}) \frac{\partial \phi}{\partial S}(t, S_t)  J_{X}(\ud t, \ud x)
\label{Sdynamicsexplicittwo}
\\
&&+\int_{\mathbb{R}} S_{t}(e^{x}-1)+\rho S_{t}\left( \phi(t,S_{t}+H(t,x,S_{t}))-\phi(t,S_{t})\right) J_{X}(\ud t, \ud x).
\nonumber
\end{eqnarray}
By comparing terms in (\ref{Sdynamicsexplicit}) and  (\ref{Sdynamicsexplicittwo}), we obtain the expressions \eqref{coeffcondb}, \eqref{coeffcondv}, and the implicit equation for the function $H$:
\begin{eqnarray}
H(t,x,S)&=&\frac{1}{1-\rho S\frac{\partial \phi}{\partial S}(t,S)}\left(S(e^{x}-1)+\rho S\left(\phi(t,S+H(t,x,S))-\phi(t,S)\right)\right)\nonumber
\\
&&-\frac{1}{1-\rho S\frac{\partial \phi}{\partial S}(t,S)}\rho S\frac{\partial \phi}{\partial S}(t,S) H(t,x,S).
\label{coeffcondH}
\end{eqnarray}
Simplifying this expression for $H$, we establish \eqref{coeffcondone}, as claimed.
\end{proof}
Since the function $H$ is given implicitly by equation \eqref{coeffcondone}, we can expand its solution in terms of a small parameter $\rho$ as follows:
\[H(t,x,S) = H^0(t,x,S) + \rho H^1(t,x,S) + O (\rho^2)~ as ~ \rho\to0,\]
to deduce the following proposition:

\begin{proposition}\label{prop-Happrox}
Assume $\rho$ is sufficiently small. Then, the first-order approximation of the function $H(t,x,S)$ is given as follows:
\begin{eqnarray}
&&H(t,x,S)= S(e^{x}-1)+\rho S\left(\phi(t,Se^{x})-\phi(t,S)\right)+O (\rho^2)\quad\text{as}\ \rho\to0.
\label{Hsmall}
\end{eqnarray}
\end{proposition}
\noindent The next proposition and its proof is based on the recent paper by Cruz and \v{S}ev\v{c}ovi\v{c} \cite{NBS19}.
\begin{proposition} \cite[Proposition 3.4] {NBS19}
Assume that the asset price process $S_{t}=e^{X_t +r t}$ satisfies SDE \eqref{Sdynamicsexplicit}, where the L\'evy measure $\nu$ is such that $\int_{|x|\ge 1}  e^{2x} \nu\left(\ud x\right)<\infty$. Let the price of a derivative security $V(t,S)$ be define by 
\begin{equation}
V(t,S)=\mathbb{E}\left[e^{-r(T-t)}\Phi(S_{T})|S_{t}=S\right]=e^{-r(T-t)}\mathbb{E}\left[\Phi(Se^{r(T-t)+X_{T-t}})\right].
\end{equation}
Suppose that the payoff function $\Phi$ is Lipschitz continuous and the function $\phi$ has a bounded derivative. Then, $V(t,S)$ is a solution to the PIDE of the form:
\begin{eqnarray}
&&0=\frac{\partial V}{\partial t}+\frac{1}{2}v(t,S)^2 S^2\frac{\partial^2 V}{\partial S^2}+rS\frac{\partial V}{\partial S}-rV\nonumber
\\
&&\qquad +\int_{\mathbb{R}} V(t,S+H(t,x,S))-V(t,S)-H(t,x,S)\frac{\partial V}{\partial S}(t,S)\nu(\ud x),
\label{nonlinearPIDE}
\end{eqnarray}
where $v(t,S)$ and $H(t,x,S)$ are given by $\eqref{coeffcondv}$ and $\eqref{coeffcondone}$, respectively.
\end{proposition}

\begin{proof}
 Recall that the asset price dynamics of $S_{t}$ under the measure $\mathbb{Q}$ is given by
\begin{eqnarray}
\ud S_{t}=rS_{t}\ud t + v(t,S_{t}) S_{t} \ud W_{t}+\int_{\mathbb{R}} H(t,x,S_{t}) \tilde{J}_{X}(\ud t, \ud x).\label{SQdynamics}
\end{eqnarray}
Applying the It\^{o}'s lemma to $V(t,S_{t})$, we obtain $\ud(V(t,S_{t})e^{-rt})=a(t)\ud t+\ud M_{t}$, where 
\begin{eqnarray*}
a(t)& = &\frac{\partial V}{\partial t}+\frac{1}{2}v(t,S_t)^2S_{t}^2 \frac{\partial^2 V}{\partial S^2}+rS_{t}\frac{\partial V}{\partial S}-rV\nonumber
\\
&&+\int_{\mathbb{R}} V(t,S_{t}+H(t,x,S_{t}))-V(t,S_{t})-H(t,x,S_{t})\frac{\partial V}{\partial S}(t,S_t)\nu(\ud x),
\end{eqnarray*}
and 
\[
\ud M_{t} = e^{-rt}S_{t}v(t,S_{t})\frac{\partial V}{\partial S}\ud W_{t}+e^{-rt}\int_{\mathbb{R}} V(t,S_{t}+H(t,x,S_{t}))-V(t,S_{t})\tilde{J}_{X}(\ud t, \ud x).
\]
Our aim is to show that $M_{t}$ is a martingale. Consequently, we have $a\equiv 0$ a.s., and  $V$ is a solution to \eqref{nonlinearPIDE} (\cite[Proposition 8.9]{ConTan03}). To show that the term  $\int_{0}^{T}e^{-rt}\int_{\mathbb{R}} V(t,S_{t}+H(t,x,S_{t}))-V(t,S_{t})\tilde{J}_{S}(\ud t, \ud y)$ is a martingale, it is suffices to show that 
\begin{eqnarray}
\mathbb{E}\left[\int_{0}^{T}e^{-2rt} \left(\int_{\mathbb{R}}V(t,S_{t}+H(t,x,S_{t}))-V(t,S_{t})\nu(\ud x)\right)^2\ud t\right]<\infty.
\end{eqnarray}
Since $\sup_{0\le t\le T}\mathbb{E}\left[e^{X_{T-t}}\right] <\infty$, and the payoff function $\Phi$ is Lipschitz continuous, $V(t,S)$ is Lipschitz continuous with the Lipschitz constant $C>0$. Moreover, since the function $\phi(t,S)$ has bounded derivatives, we have
\[
S \left|\phi(t,S+H(t,x,S))-\phi(t,S)\right|
\leq S \left|\frac{\partial \phi}{\partial S}\right| |H(t,x,S)|
\leq L |H(t,x,S)|.
\]
(see Assumption \ref{assumptionone}). Since $H(t,x,S)=S (e^{x}-1)+\rho S (\phi(t,S+H(t,x,S))-\phi(t,S) )$, we obtain $|H(t,x,S)|^{2}\leq S^{2}(e^{x}-1)^2/(1-\rho L)^2$. Again, since $V$ is Lipschitz continuous with the Lipschitz constant $C>0$, we have
\begin{eqnarray*}
&&\mathbb{E}\left[\int_{0}^{T}e^{-2rt} \left(\int_{\mathbb{R}}V(t,S_{t}+H(t,x,S_{t}))-V(t,S_{t})\nu(\ud x)\right)^2\ud t\right]
\\
&&\le \frac{C^2 }{(1-\rho L)^2} \mathbb{E}\left[\int_{0}^{T}\int_{\mathbb{R}} e^{-2rt}  |S_t|^2  (e^{x}-1)^2 \nu(\ud x)  \ud t \right] <\infty,
\end{eqnarray*}
since $\sup_{t\in \left[0,T\right]} \mathbb{E}\left[S_{t}^2\right]<\infty$. Here, $C_0 = \int_{\mathbb{R}} (e^{x}-1)^2 \nu(\ud x) <\infty$ because of the assumptions made on the measure $\nu$. Next, we show that $\int_{0}^{T}e^{-rt}S_{t}v(t,S_{t})\frac{\partial V}{\partial S}(t,S_t)\ud W_{t}$ is a martingale. Since $S \frac{\partial \phi}{\partial S}(t,S)$ is bounded, we have
\begin{eqnarray}
0<v(t,S)=\frac{\sigma}{1-\rho S\frac{\partial \phi}{\partial S}(t,S)}\leq \frac{\sigma}{1-\rho L}\equiv C_{1}<\infty .\nonumber
\end{eqnarray}
Therefore, $\mathbb{E}[\int_{0}^{T}e^{-2rt}(\frac{\partial V}{\partial S}(t,S_{t})v(t,S_{t}) S_{t})^{2}\ud t]
\leq C^{2} C_1^2 \int_{0}^{T}e^{-2rt}\mathbb{E}[S^{2}_{t}]\ud t<\infty$ because $S_t$ is a martingale. Thus, $M_t$ is also a martingale, and $a\equiv 0$. Hence, $V$ is a solution to PIDE (\ref{nonlinearPIDE}), as claimed.
\end{proof}

\begin{remark}
If $\rho=0$, then $H(t,x,S)=S (e^{x}-1)$. Moreover, equation (\ref{nonlinearPIDE}) becomes
\begin{equation}
\frac{\partial V}{\partial t}+\frac{\sigma^2}{2}S_{}^2\frac{\partial^2 V}{\partial S^2}+rS_{}\frac{\partial V}{\partial S}-rV+\int_{\mathbb{R}} V(t,S e^{x})-V(t,S_{})-S (e^{x}-1)\frac{\partial V}{\partial S}(t,S) \nu(\ud x)=0,
\label{classicalPIDE}
\end{equation}
which is the well-known classical PIDE. If there are no jumps ($\nu=0$) and a trader follows the delta hedging strategy, i.e.,  $\phi(t,S)=\partial_S V(t,S)$, then equation (\ref{nonlinearPIDE}) reduces to the Frey--Stremme option pricing model:
\begin{equation}
\frac{\partial V}{\partial t}+\frac{1}{2}\frac{\sigma^2}{\left(1-\varrho S\partial^2_S V\right)^{2}} S^2 \frac{\partial^2 V}{\partial S^2 } +r S\frac{\partial V}{\partial S}-rV =0
\label{Frey}
\end{equation}
(c.f. \cite{Frey98}). Finally, if $\rho=0$ and $\nu=0$, equation (\ref{nonlinearPIDE}) reduces to the classical linear Black--Scholes equation. 
\end{remark}

For simplicity, we assume the interest rate is zero, i.e., $r=0$. Then, the function $V(t,S)$ is a solution to the following PIDE:
\begin{eqnarray}
\frac{\partial V}{\partial t}&+&\frac{1}{2}v(t,S)^2 S^2 \frac{\partial^2 V}{\partial S^2}
\nonumber 
\\
&+&\int_{\mathbb{R}} V(t,S+H(t,x,S))-V(t,S)-H(t,x,S)\frac{\partial V}{\partial S}(t,S)\nu(\ud x)=0.
\label{equationsimple}
\end{eqnarray}
Next, let the tracking error of a trading strategy $\alpha_t=\phi(t,S_t)$ be
$e_{T}^{M}:=\Phi(S_T)-V_0=V(T, S_{T})-V_{0}-\int_{0}^{T} \alpha_{t} \ud S_{t}$. By applying It\^{o}'s formula to $V(t,S_{t})$ and using $\eqref{equationsimple}$, we obtain
\begin{eqnarray}
&&V(T,S_{T})-V_0= V(T,S_{T})-V(0,S_{0})=\int_{0}^{T} \ud V(t,S_{t})
\nonumber
\\
&&=\int_{0}^{T} \frac{\partial V}{\partial S}\ud S_{t}+\int_{0}^{T} \frac{\partial V}{\partial t}+\frac{1}{2}v(t,S_{t})^2 S_{t}^2\frac{\partial^2 V}{\partial^2 S} \ud t\nonumber
\\
&&\qquad +\int_{0}^{T}\int_{\mathbb{R}} V(t,S_{t}+H(t,x,S_{t}))-V(t,S_{t})-H(t,x,S_{t})\frac{\partial V}{\partial S} J_{X}(\ud t, \ud x)\nonumber
\\
&&=\int_{0}^{T} \frac{\partial V}{\partial S}\ud S_{t} - \int_{0}^{T} \int_{\mathbb{R}} V(t,S_{t}+H(t,x,S_{t}))-V(t,S_{t})-H(t,x,S_{t})\frac{\partial V}{\partial S}\nu(\ud x)\ud t\nonumber
\\
&&\qquad +\int_{0}^{T}\int_{\mathbb{R}} V(t,S_{t}+H(t,x,S_{t}))-V(t,S_{t})-H(t,x,S_{t})\frac{\partial V}{\partial S} J_{X}(\ud t, \ud x)\nonumber
\\
&&=\int_{0}^{T} \frac{\partial V}{\partial S}\ud S_{t}+\int_{0}^{T}\int_{\mathbb{R}} V(t,S_{t}+H(t,x,S_{t}))-V(t,S_{t})-H(t,x,S_{t})\frac{\partial V}{\partial S} \tilde{J}_{X}(\ud t, \ud x).
\nonumber
\end{eqnarray}
Using expression (\ref{SQdynamics}) for the dynamics of the asset price $S_t$ (with $r=0$), the tracking error $e_{T}^{M}$ can be expressed as follows:
\begin{eqnarray}
e_{T}^{M} &=& V(T, S_{T})-V_{0}-\int_{0}^{T} \alpha_{t} \ud S_{t} = \int_{0}^{T} \left(\frac{\partial V}{\partial S}(t,S_t) - \alpha_t\right)\ud S_{t}
\nonumber\\
&& +\int_{0}^{T}\int_{\mathbb{R}} V(t,S_{t}+H(t,x,S_{t}))-V(t,S_{t})-H(t,x,S_{t})\frac{\partial V}{\partial S} \tilde{J}_{X}(\ud t, \ud x)
\nonumber
\\
&=& \int_{0}^{T}  v(t, S_t) S_t \left(\frac{\partial V}{\partial S}-\alpha_t\right)\ud W_{t}
\label{trackingerr}
\\
&& +\int_{0}^{T}\int_{\mathbb{R}} V(t,S_{t}+H(t,x,S_{t}))-V(t,S_{t})-\alpha_t H(t,x,S_{t}) \tilde{J}_{X}(\ud t, \ud x).
\nonumber
\end{eqnarray}

\begin{remark}
For the delta hedging strategy $\alpha_t=\phi(t,S_{t})=\frac{\partial V}{\partial S}(t,S_{t})$, the tracking error function $e_{T}^{M}$ can be expressed as follows:
\[
e_{T}^{M}=\int_{0}^{T}\int_{\mathbb{R}} V(t,S_{t}+H(t,x,S_{t}))-V(t,S_{t})-H(t,x,S_{t})\frac{\partial V}{\partial S}(t, S_t) \tilde{J}_{X}(\ud t, \ud x).
\]
It is obvious that the tracking error for the delta hedging strategy need not be zero for $\nu\not\equiv0$. 
\end{remark}
The next proposition presents a criterion that can be used to find the optimal hedging strategy. The proposition and its proof is based on the recent paper by Cruz and \v{S}ev\v{c}ovi\v{c} \cite{NBS19}.

\begin{proposition} \cite[Proposition 3.5]{NBS19}
\label{prop-trackingerror-min}
The trading strategy $\alpha_t=\phi(t,S_t)$ of a large trader minimizing the variance $\mathbb{E}\left[(\epsilon_{T}^M)^{2}\right]$ of the tracking error is given by the implicit equation:
\begin{eqnarray}
\phi(t,S_t)&=&\beta^\rho(t,S_t)\bigl[v(t,S_{t})^{2}S_{t}^{2}\frac{\partial V}{\partial S}(t,S_{t})
\nonumber
\\
&& +\int_{\mathbb{R}} \left(V(t,S_{t}+H(t,x,S_{t}))-V(t,S_{t})) H(t,x,S_{t}\right)\nu(\ud x)\bigr],
\label{strategy}
\end{eqnarray}
where $\beta^\rho(t,S_t)=1/[v(t,S_{t})^{2}S_{t}^{2}+\int_{\mathbb{R}}H(t,x,S_{t})^2\nu(\ud x) ]$ and $H(t,x,S)= S (e^x-1) + \rho S [\phi(t,S + H(t,x,S) )-\phi(t,S)]$.
\end{proposition}

\begin{proof}
Using equation (\ref{trackingerr}) for the tracking error $\epsilon_{T}^M$ and It\^{o}'s isometry, we have 
\begin{eqnarray*}
\mathbb{E}\left[(\epsilon_{T}^M)^{2}\right]
&=&\mathbb{E}\left[\int_{0}^{T}v(t,S_{t})^{2}S_{t}^{2}\left(\frac{\partial V}{\partial S}(t,S_t) - \alpha_{t}\right)^{2}\ud t\right]\nonumber
\\
&&+\mathbb{E}\left[\int_{0}^{T}\int_{\mathbb{R}}\left(V(t,S_{t}+H(t,x,S_{t}))-V(t,S_{t})-\alpha_{t}H(t,x,S_{t})\right)^{2}\nu(\ud x)\ud t\right].
\end{eqnarray*}
The minimizer $\alpha_{t}$ of the above convex quadratic minimization problem satisfies the first-order necessary conditions 
$\ud (\mathbb{E}\left[\epsilon_{T}^{2}\right],\alpha_{t})=0$, that is, 
\begin{eqnarray}
0&=&-2\mathbb{E}\left[\int_{0}^{T}\left(v(t,S_{t})^{2}S_{t}^{2}\left(\frac{\partial V}{\partial S}(t,S_t) -\alpha_{t}\right)\right.\right.\nonumber
\\
&&\left.\left.+\int_{\mathbb{R}}H(t,x,S_{t})\bigl(V(t,S_{t}+H(t,x,S_{t}))-V(t,S_{t})-\alpha_{t}H(t,x,S_{t})\bigr)\nu(\ud x)\right)\omega_{t}\ud t\right]\nonumber
\end{eqnarray}
for any variation $\omega_{t}$. Thus, the tracking error minimizing strategy $\alpha_{t}$ is given by \eqref{strategy}.
\end{proof}

\begin{remark}
It is worth noting that the optimal trading strategy minimizing the variance of the tracking error need not satisfy the structural Assumption~\ref{assumptionone}. For instance, if $\nu=0$, then the tracking error minimizer is the delta hedging strategy $\phi=\partial_S V$. For a call or put option, its gamma, i.e., $\partial^2_S V(t,S)$, becomes infinite as $t\to T$ and $S=K$. However, given a level $L>0$, we can minimize the tracking error $\mathbb{E}\left[\epsilon_{T}^{2}\right]$ under the additional constraint $\sup_{S>0} |S\frac{\partial \phi}{\partial S}(t,S)|\le L$. In other words, we can solve the following convex constrained nonlinear optimization problem
\[
\min_{\phi} \ \ \mathbb{E}\left[\epsilon_{T}^{2}\right] \ \ s.t. \ \ |S\partial_S \phi|\le L
\]
instead of the unconstrained minimization problem proposed in Proposition~\ref{prop-trackingerror-min}.
\end{remark}

\begin{remark}
Note that if $\nu=0$ and $\rho\ge 0$, the trading strategy $\alpha_t$ reduces to the Black--Scholes delta hedging strategy, i.e., $\alpha_{t}=\frac{\partial V}{\partial S}(t,S_t)$. Meanwhile, if $\nu\not\equiv0$ and $\rho=0$, then the optimal trading strategy becomes $\alpha_t=\phi^{0}(t,S_t)$, where
\[
\phi^{0}(t,S_t)=\beta^0(t,S_t) \left(\sigma^{2}S_{t}^{2}\frac{\partial V}{\partial S}(t,S_t) + \int_{\mathbb{R}}S_{t}(e^{x}-1)\left(V(t,S_{t}e^{x})-V(t,S_{t})\right)\nu(\ud x)\right),
\]
where $\beta^0(t,S_t)=1/[\sigma^{2}S_{t}^{2}+\int_{\mathbb{R}}S^2_{t}(e^{x}-1)^2\nu(\ud x)]$. 
\end{remark}
We conclude this section with the following proposition providing the first-order approximation of the tracking error minimizing trading strategy for small values of the parameter $\rho\ll 1$. In what follows, we derive the first-order approximation of $\phi^\rho(t,S_{t})$ as  $\phi^\rho(t,S_{t})=\phi^{0}(t,S_t) + \rho\phi^1(t,S_t)+O(\rho^{2})$ as $\rho\to 0$. Clearly, the first-order Taylor expansion of the volatility function $v(t,S)$ is given by
\[
v(t,S)^2 = \frac{\sigma^2}{(1-\rho S\partial_S \phi)^2} = \sigma^2 + 2 \rho  \sigma^2 S \frac{\partial\phi^0}{\partial S}(t,S) + O(\rho^{2}), \ \ \text{as}\ \rho\to 0.
\]
According to Proposition~\ref{prop-Happrox} (see \eqref{Hsmall}), we have  $H(t,x,S) = H^0(t,x,S) + \rho H^1(t,x,S) + O(\rho^{2})$, where
\begin{equation}
H^0(t,x,S) = S (e^x-1), \qquad H^1(t,x,S) = S [\phi^0(t,S e^x) - \phi^0(t,S)].
\label{Hexpand}
\end{equation}
The function $\beta^\rho$ can be expanded as follows: $\beta^\rho(t,S) = \beta^0(t,S) + \rho\beta^{(1)}(t,S) + O(\rho^{2})$,
\begin{eqnarray}
\beta^0(t,S) &=& 1/[\sigma^{2}S^{2}+ S^2 \int_{\mathbb{R}}(e^{x}-1)^2\nu(\ud x)], \quad 
\label{betaexpand}
\\
\beta^{(1)}(t,S) &=& - (\beta^0(t,S))^2 
\left[2 \sigma^{2}S^{3} \frac{\partial\phi^0}{\partial S}(t,S) + 2 S^2\int_{\mathbb{R}}(e^{x}-1) [\phi^0(t,S e^x) - \phi^0(t,S)]\nu(\ud x)\right].
\nonumber
\end{eqnarray}
Using the first-order expansions of the functions $v^2, \beta^\rho$, and $H$, we obtain the following results.

\begin{proposition}\cite[Proposition 3.6]{NBS19}
For small values of the parameter $\rho\ll 1$, the tracking error variance minimizing strategy $\alpha_{t}=\phi^\rho(t,S_{t})$ is given by 
\begin{eqnarray}
&&\phi^\rho(t,S_{t})=\phi^{0}(t,S_t) + \rho\phi^{(1)}(t,S_t)+O(\rho^{2}), \ \ \text{as}\ \rho\to 0, 
\end{eqnarray}
where 
\begin{eqnarray*}
\phi^{(1)}(t,S) &=&\beta^0(t,S) \left[ 2 \sigma^2 S^3 \frac{\partial V}{\partial S}(t,S) \frac{\partial \phi^0}{\partial S}(t,S)\right.
\\
&& + \left. \int_{\mathbb{R}}\left(V(t,S e^{x})-V(t,S) +  \frac{\partial V}{\partial S}(t,S e^x) H^0(t,x,S) \right) H^1(t,x,S) \nu(\ud x)
\right] 
\\
&& + \beta^{(1)}(t,S) \left[ \sigma^2 S^2 \frac{\partial V}{\partial S}(t,S)  
 +  \int_{\mathbb{R}}\left(V(t,S e^{x})-V(t,S)\right) H^0(t,x,S) \nu(\ud x)
\right],
\end{eqnarray*}
and the functions $H^0, H^1, \beta^0$, and $\beta^{(1)}$ are defined as in \eqref{Hexpand} and \eqref{betaexpand}.
\end{proposition}

\subsection{Numerical simulation for the underlying PIDE}

In this section, we illustrate the behavior of the solutions to the linear PIDE with various L\'evy measures. Specifically, we consider European put options, i.e., $\Phi(S)=(K-S)^+$. The goal of the numerical simulation is to compare the solution to the linear Black--Scholes equation with solutions to the Merton and variance gamma PIDE models. The common model parameters were chosen as follows: $\sigma=0.23, K=100, T=1$, and $r\in\{0, 0.1\}$. For the underlying L\'evy process, we consider the variance gamma process with parameters $\theta=-0.43, \kappa=0.27$ and the Merton processes with parameters $\lambda=0.1$, $m=-0.2$, and $\delta=0.15$. First, we employ the finite difference discretization scheme proposed and analyzed by Cruz and \v{S}ev\v{c}ovi\v{c} \cite{NBS19} to calculate the numerical solution to the equation. Their scheme is based on a uniform spatial finite difference discretization with a spatial step $\Delta x = 0.01$, and implicit time discretization with a step $\Delta t = 0.005$. Then, we set the total number of spatial discretization steps as $N=400$ and the number of time discretization steps as $M=200$. We restricted the spatial computational domain to $x\in[-L,L ]$, where $L=4$. For more details about the discretization scheme, see the recent paper by Cruz and \v{S}ev\v{c}ovi\v{c} \cite{NBS19}.

Figure~\ref{fig-1} shows the comparison of European put option prices between PIDE and linear Black--Scholes models. Figures~\ref{fig-1} (a) and (b) show the plots of put option prices $V(0,S)$ for $S\in[80,125]$ for the interest rates $r=0$ and $r=0.1$, respectively. Table~\ref{tab1} summarizes the numerical values of option prices for two different values of the interest rate $r=0.1$ and $r=0$. The option price for the Merton and variance gamma models are higher than that of the classical Black-Scholes model. This is based on the idea prices of put or call options should be higher on underlying assets following stochastic processes with jumps than those following a continuous geometric Brownian motion.

\begin{figure}[!ht]
\centering
    \includegraphics[width=0.48\textwidth]{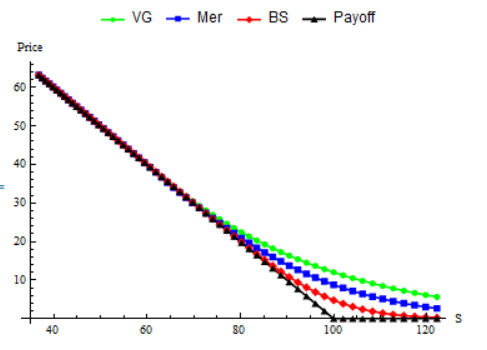}
\quad
    \includegraphics[width=0.48\textwidth]{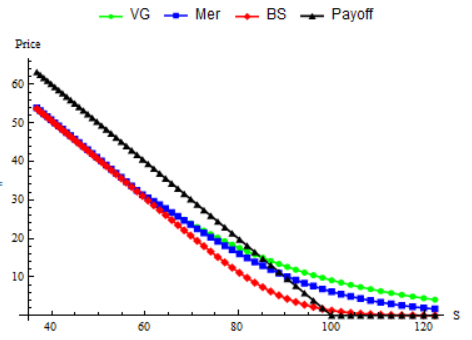}

    a) \hskip 7truecm b)

\caption{
Comparison of European put option prices for the Black--Scholes (BS) model and PIDE models under variance gamma (VG) and Merton (Mer) PIDE processes. Source: own computations based on the paper \cite{NBS19} by J. Cruz and D. \v{S}ev\v{c}ovi\v{c}.
} 
\label{fig-1}
\end{figure}

\begin{table}
\small

\centering
\caption{European put option prices $V(0,S)$ for the Black-Scholes and PIDE models under Variance Gamma and Merton processes for $r=0$ and $r=0.1$. Source: J. Cruz and D. \v{S}ev\v{c}ovi\v{c}, \cite{NBS19, CruzPhD2021}.}
\label{tab1}
\small

\medskip
\begin{tabular}{l|ll|ll|ll|l}
         & \multicolumn{2}{c|}{BS} &    \multicolumn{2}{c|}{PIDE-VG} & \multicolumn{2}{c|}{PIDE-Merton}  &  Payoff  \\ 
 $S$     &  $r=0$  & $r=0.1$ & $r=0$   & $r=0.1$ & $r=0$   & $r=0.1$ & \\
\hline\hline
 85.2144 & 15.2547 & 7.35166 & 19.2687 & 14.9855 & 17.1692 & 12.9056 & 14.7856\\
 88.692  & 12.2484 & 5.24145 & 17.2948 & 13.3899 & 14.8335 & 10.9901 & 11.308\\
 92.3116 & 9.42895 & 3.51944 & 15.428  & 11.8822 & 12.6423 & 9.21922 & 7.68837\\
 96.0789 & 6.90902 & 2.21106 & 13.674  & 10.4691 & 10.6201 & 7.61307 & 3.92106\\
 100.    & 4.78444 & 1.29196 & 12.0372 & 9.15576 & 8.78655 & 6.18483 & 0.\\
 104.081 & 3.1099  & 0.69843 & 10.52   & 7.94499 & 7.155   & 4.94044 & 0.\\
 108.329 & 1.88555 & 0.34773 & 9.12343 & 6.83762 & 5.73137 & 3.87864 & 0.\\
 112.75  & 1.0604  & 0.15881 & 7.84623 & 4.51403 & 5.83246 & 2.99166 & 0.\\
\hline
\end{tabular}

\end{table}

\newpage
\section{Hamilton-Jacobi-Bellman equation}
In this section, we present the motivation for studying the fully nonlinear parabolic equation (\ref{generalPDE}), which can be viewed as a parabolic PIDE in some sense. We also present the relationship between the nonlinear parabolic equation and the transformed equation using the so-called Ricatti transform. 

The motivation for studying the nonlinear parabolic equation of the form (\ref{generalPDE}) arises from dynamic stochastic programming for $d=1$. The fully nonlinear HJB equation describing optimal portfolio selection strategy is represented by the following nonlinear parabolic equation:
\begin{eqnarray}
&& \partial_t V + \max_{ {\bm{\theta}} \in \triangle} 
\left(
\mu(x,t,{\bm{\theta}})\, \partial_x V 
+ \frac{1}{2} \sigma(x,t,{\bm{\theta}})^2\, \partial_x^2 V \right) = 0\,,  
\label{eq_HJB}
\\
&& V(x,T)=u(x),  \label{init_eq_HJB}
\end{eqnarray}
where $x\in\mathbb{R}, t\in [0,T)$. A solution $V=V(x,t)$ to the parabolic equation (\ref{eq_HJB}) is subject to the terminal condition $V(x,T)=u(x)$. According to Kilianov\'a and {\v S}ev{\v c}ovi{\v c} \cite{KilianovaSevcovicANZIAM, KilianovaSevcovicKybernetika, KilianovaSevcovicJIAM}, such HJB equation of the form (\ref{eq_HJB}) arises from the dynamic stochastic programming, where the goal of an investor is to maximize the conditional expected value of the terminal utility of the portfolio:
\begin{equation}
\max_{{\bm{\theta}}|_{[0,T)}} \mathbb{E}
\left[u(x_T^{\bm{\theta}})\, \big| \, x_0^{\bm{\theta}}=x_0 \right],
\label{maxproblem}
\end{equation}
on a finite time horizon $[0,T]$. Here, $u: \mathbb{R} \to \mathbb{R}$ is an increasing terminal utility function, and $x_0$ is a given initial state condition of the process $\{x_t^{\bm{\theta}}\}$ at $t=0$. The underlying stochastic process $\{x_t^{{\bm{\theta}}}\}$ with a drift $\mu(x,t,{\bm{\theta}})$ and volatility $\sigma(x,t,{\bm{\theta}})$ is assumed to satisfy the following It\^{o}'s SDE:
\begin{equation}
\label{process_x}
d x_t^{\bm{\theta}} = \mu(x_t^{\bm{\theta}}, t, {\bm{\theta}}_t) dt + \sigma(x_t^{\bm{\theta}}, t, {\bm{\theta}}_t) dW_t\,,
\end{equation}
where the control process $\{{\bm{\theta}}_t\}$ is adapted to the process $\{x_t\}$, and $\{W_t\}$ is the standard one-dimensional Wiener process. The control parameter ${\bm{\theta}}$ belongs to a given compact subset $\triangle$ in $\mathbb{R}^n$. An example of such a subset is the compact convex simplex $\triangle\equiv\mathcal{S}^n = \{{\bm{\theta}} \in \mathbb{R}^n\  |\  {\bm{\theta}} \ge \mathbf{0}, \mathbf{1}^{T} {\bm{\theta}} = 1\} \subset \mathbb{R}^n$, where $\mathbf{1} = (1,\cdots,1)^{T} \in \mathbb{R}^n$. 

Consider the value function
\begin{equation}
V(x,t):= \sup_{  {\bm{\theta}}|_{[t,T)}} 
\mathbb{E}\left[u(x_T^{\bm{\theta}}) | x_t^{\bm{\theta}}=x \right].
\end{equation}
Then, according to Bertsekas \cite{Bertsekas}, the value function $V=V(x,t)$ solves the fully nonlinear HJB parabolic equation (\ref{eq_HJB}) and $V(x,T):=u(x)$. 

Several attempts have been made for solving the HJB equation (\ref{eq_HJB}). Macov{\'a} and {\v{S}}ev{\v{c}}ovi{\v{c}} \cite{MS} analyzed solutions to a fully nonlinear parabolic equation modeling the problem of optimal portfolio construction. Consequently, they formulated the problem of optimal stock to bond proportion in the management of a pension fund portfolio could be formulated in terms of the solutions to the HJB equation. Federicol et al. \cite{Federico} investigated the utility maximization problem for an investment-consumption portfolio when the current utility depends on the wealth process, regularity of solutions to the HJB equation. They defined a dual problem and treated it by means of dynamic programming to show that the viscosity solutions of the associated HJB equation belong to a class of smooth function. Recently, Ishimura and \v{S}ev\v{c}ovi\v{c} \cite{IshSev} constructed and analyzed solutions to the class of HJB equation (\ref{eq_HJB}) with range bounds on the optimal response variable. They constructed monotone traveling wave solutions and identified parametric regions for which the traveling wave solutions have positive or negative wave speed. Abe and Ishimura \cite{AI} employed the Riccati transformation method for solving the full nonlinear HJB equations, which was later studied and generalized by Kilianov\'a and \v{S}ev\v{c}ovi\v{c} \cite{KilianovaSevcovicANZIAM}. They investigated solutions of a fully nonlinear HJB equation for a constrained dynamic stochastic optimal allocation problem. However, no attempt has been made in solving the fully nonlinear HJB parabolic equation arising from portfolio optimization in high-dimensional space using the monotone operator technique. The monotone operator method is essential because it does not only give constructive proof for existence theorems, but it also leads to various comparison results, which are effective tools for studying qualitative properties of solutions. In this chapter, we consider the case when the utility function $u$ is increasing, as a consequence, $\partial_x V(x,\tau) >0$.

\subsection{Static Markowitz model for portfolio optimization}

This subsection presents the motivation for studying such nonlinear parabolic equation  (\ref{eq_HJB}). It describes the mathematical formulation of the classical Markowitz model for portfolio optimization. In the static portfolio optimization, this model aims to maximize the mean return of the set of stochastic returns $X^i, i=1,\ldots, n$, under the constraint that variance of the portfolio is bounded by a given constant $\sigma_0^2/2$. 
Given a vector $\bm{\theta}=(\theta^1, \ldots, \theta^n)^T $ of weights, we construct a portfolio $X =\sum_{i=1}^n \theta^i X^i$. Let $\bm{\mu}\in\mathbb{R}^n, \mu^i=\mathbb{E}(X^i)$ be the vector of mean returns of stochastic asset returns and $\bm{\Sigma}$ be their covariance matrix, $\bm{\Sigma}_{ij} = cov(X^i X^j)$, then 
$\mathbb{E}(X) = \bm{\mu}^T \bm{\theta}$, and the variance $\mathbb{D}(X) = \bm{\theta}^T\bm{\Sigma} \bm{\theta}$. The Markowitz optimal portfolio optimization problem can be formulated as the following convex optimization problem:
\begin{eqnarray*}
&&\max_{\bm{\theta}\in\R^n} \bm{\mu}^T \bm{\theta} \qquad \hbox{-- maximize the mean return},
\\
&& s.t.\ \ \frac12 \bm{\theta}^T\bm{\Sigma} \bm{\theta} \le  \frac12\sigma_0^2 \qquad \hbox{-- the variance is prescribed},
\\
&& \ \ \ \ \ \ \sum_{i=1}^n \theta^i = 1 \qquad \hbox{-- weights sum up to 100\%},
\\
&& \ \ \ \ \ \ \bm{\theta}\ge 0 \qquad \hbox{-- no short positions allowed}.
\end{eqnarray*}
The corresponding Lagrange function for the minimization of $-\bm{\mu}^T\bm{\theta}$ has the following form:
\begin{equation}
    {\mathcal L}(\bm{\theta}, \varphi, \lambda, \xi) = - \bm{\mu}^T \bm{\theta} 
+ { \varphi \frac12 \bm{\theta}^T\bm{\Sigma} \bm{\theta}}
+ \lambda 1^T\bm{\theta} + \xi^T \bm{\theta},
\label{lagrange}
\end{equation}
where $\varphi\in\R, \lambda\in\mathbb{R}, \xi\in\mathbb{R}^n$, and $\xi\ge0$ are Lagrange multipliers. It is worth noting that the same Lagrange function corresponds to the minimization problem:
\begin{eqnarray*}
&&\alpha(\varphi) := \min_{\bm{\theta}\in \triangle } - \bm{\mu}^T \bm{\theta} + { \varphi \frac12 \bm{\theta}^T\bm{\Sigma} \bm{\theta}}
\\
&& \ \ \text{where} \ \ \triangle \equiv  \{ \bm{\theta}\in\mathbb{R}^n, \  \sum_{i=1}^n \theta^i = 1, \ \  \bm{\theta}\ge 0 \},
\end{eqnarray*}
provided the Lagrange multiplier $\varphi>0$ is known and fixed. Figure \ref{fig:pies} shows the optimal asset selection for German DAX30 stock index for various values of $\varphi>0$. The optimal value is denoted by $\alpha(\varphi)$. The value of the Lagrange multiplier $\varphi$ can be viewed as a measure of investor's risk-aversion (see Fig.~\ref{fig:pies}). Therefore, the higher the value of the risk aversion, the more portfolio is diversified among less risky assets with smaller mean returns.  

\begin{figure}
    \centering
\includegraphics[width=0.22\textwidth]{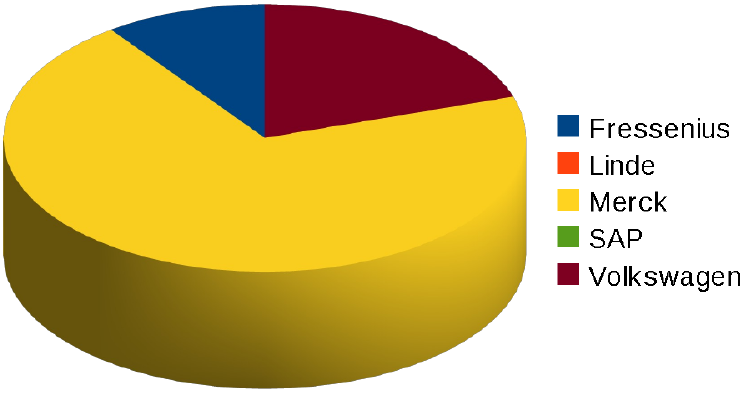}
\ \includegraphics[width=0.22\textwidth]{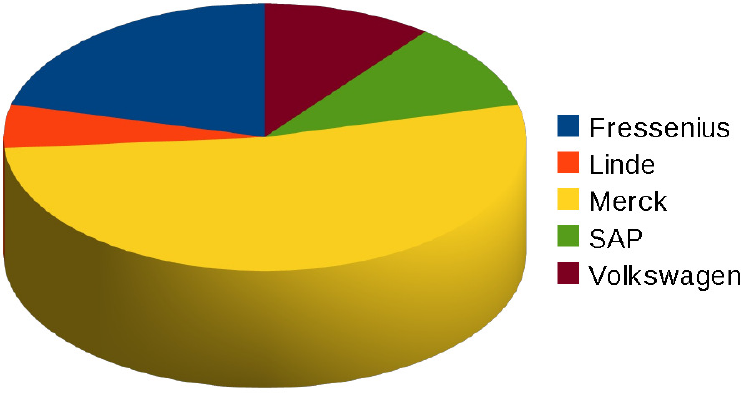}
\ \includegraphics[width=0.22\textwidth]{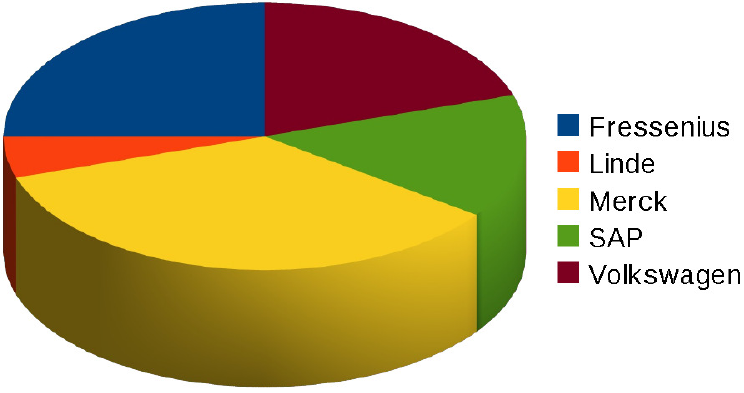}
\ \includegraphics[width=0.22\textwidth]{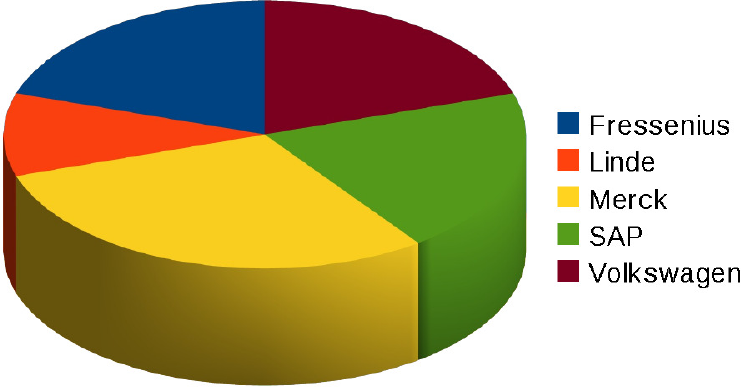}

$\varphi \approx 1$ \hglue2truecm $\varphi \approx 4$ \hglue 2truecm $\varphi \approx 6$ \hglue2truecm $\varphi \approx 8$

\caption{\small Optimal asset selection for German DAX30 stock index for various values of $\varphi>0$.}
\label{fig:pies}
\end{figure}

\subsection{Riccati transformation of the HJB equation and application to optimal portfolio selection problem}
\label{sec:HJB}

\subsubsection{Riccati transformation}

This section presents how the HJB equation (\ref{eq_HJB}) can be transformed into a quasilinear PDE using the so-called Ricatti transformation techniques. Such transformed parabolic equation correspond to the Cauchy problem (\ref{generalPDE}), which is obtained after some pertubation in the main operator.

The Riccati transformation $\varphi$ of the value function $V$ can be defined based on the approached introduced by Abe and Ishimura \cite{AI}, Ishimura and \v{S}ev\v{c}ovi\v{c} \cite{IshSev}, \v{S}ev\v{c}ovi\v{c} and Macov\'a \cite{MS}, and Kilianov\'a and \v{S}ev\v{c}ovi\v{c} \cite{KilianovaSevcovicANZIAM} as follows:

\begin{equation}
\varphi(x,\tau) = - \frac{\partial_x^2 V(x,t)}{\partial_x V(x,t)}, \quad\hbox{where}\ \ \tau=T-t.
\label{eq_varphi}
\end{equation}

Suppose the value function $V(x,t)$ is increasing in the $x$-variable. In other words, assume that the terminal utility function $u(x)$ is increasing. Then, the HJB equation (\ref{eq_HJB}) can be rewritten as follows:
\begin{equation}
\partial_t V - \alpha(\cdot,\varphi) \partial_x V = 0, \qquad V(\cdot ,T)=u(\cdot),
\label{eq_HJBtransf}
\end{equation}
where $\alpha(x,\tau,\varphi)$ is the value function of the following parametric optimization problem:
\begin{equation}
\alpha(x,\tau,\varphi) = \min_{ {\bm{\theta}} \in \triangle} 
\left(
-\mu(x,t,{\bm{\theta}}) +  \frac{\varphi}{2}\sigma(x,t,{\bm{\theta}})^2\right), \quad \tau=T-t\,.
\label{eq_alpha_def}
\end{equation}

\begin{remark}
It is worth noting that the optimization problem (\ref{eq_alpha_def}) is related to the classical Markowitz model on optimal portfolio selection problem formulated as maximization of the mean return $\mu({\bm{\theta}})\equiv {\bm{\mu}}^T {\bm{\theta}}$ under the volatility constraint $\frac12\sigma({\bm{\theta}})^2\equiv\frac12 {\bm{\theta}}^T{\bm{\Sigma}} {\bm{\theta}} \le  \frac12\sigma^2_0$, i.e.,
\begin{eqnarray*}
&&\max_{{\bm{\theta}}\in\triangle} {\bm{\mu}}^T {\bm{\theta}},
\quad s.t.\ \ \frac12 {\bm{\theta}}^T{\bm{\Sigma}} {\bm{\theta}} \le  \frac12\sigma^2_0,
\end{eqnarray*}
where the decision set is the simplex
$\triangle= \{{\bm{\theta}} \in \mathbb{R}^n\  |\  {\bm{\theta}} \ge \mathbf{0}, \mathbf{1}^{T} {\bm{\theta}} = 1\}$. The Lagrange multiplier for the volatility constraint can be viewed as the parameter $\varphi$ entering the parametric optimization problem (\ref{eq_alpha_def}).
\end{remark}

Next, let $\partial_x\alpha$ be the total differential of the function $\alpha(x,\tau,\varphi)$, where $\varphi=\varphi(x,\tau)$, i.e.,
\[
\partial_x\alpha (x,\tau,\varphi) = \alpha^\prime_x(x,\tau,\varphi) + \alpha^\prime_\varphi(x,\tau,\varphi)\, \partial_x \varphi.
\]
Here, $\alpha^\prime_x$ and $\alpha^\prime_\varphi$ are partial derivatives of $\alpha$ with respect to  $x$- and $\varphi$- variables, respectively.

Kilianov\'a and \v{S}ev\v{c}ovi\v{c} \cite[Theorem 4.2]{KilianovaSevcovicKybernetika} recently established the relationship between the transformed function $\varphi$ and the value function $V$. The reported that an increasing value function $V(x,t)$ in the $x$-variable is a solution to the HJB equation (\ref{eq_HJB}) if and only if the transformed function $\varphi(x,\tau) = -\partial_x^2 V(x,t)/\partial_x V(x,t),\,\ t=T-\tau$, is a solution to the Cauchy problem for the quasilinear parabolic PDE:
\begin{eqnarray}
&&\partial_\tau \varphi - \partial^2_x \alpha(\cdot,\varphi) = - \partial_x \left( \alpha(\cdot,\varphi)\varphi\right),
\label{eq_PDEphi-cons}
\\
&&\varphi(x,0) = \varphi_0(x) \equiv -u''(x)/u'(x),\quad (x,\tau)\in\mathbb{R}\times(0,T). 
\label{init_PDEphi-cons}
\end{eqnarray}

Note that the Cauchy problem for the quasilinear parabolic PDE (\ref{eq_PDEphi-cons}) is equivalent to the nonlinear parabolic equation (\ref{generalPDE}) in one-dimensional space. This is obtainable after some shift/perturbation in the main operator of the transformed equation (\ref{eq_PDEphi-cons}). This nonlinear parabolic equation (\ref{generalPDE}) presented in an abstract setting can be viewed as nonlinear PIDE in some sense.

\subsubsection{Properties of the value function as a diffusion function}

In this subsection, we presents the qualitative properties of the value function and sufficient conditions on the decision set $\triangle$ and functions $\mu$ and $\sigma$ that guarantee higher smoothness of the value function $\alpha$. Let $C^{k,1}(\mathcal{D})$ be the space consisting of all $k$-differentiable functions defined on the domain $\mathcal{D}\subset \mathbb{R}^{d+2}$, whose $k$-th derivative is globally Lipschitz continuous. The next proposition shows (under certain assumptions) that the value function $\alpha$ belongs to $C^{0,1}(\mathcal{D})$, where $\mathcal{D}=\mathbb{R}^d\times (0,T)\times (\varphi_{min},\infty)$. The following proposition and its proof are contained in our recent paper (\cite[Theorem 3]{udeani2021application}).

\begin{proposition}\cite[Proposition 1]{udeani2021application}
\label{smootheness0}
Let $\triangle\subset\mathbb{R}^n$ be a given compact decision set. Assume that the functions $\mu(x,t,{\bm{\theta}})$ and $\sigma(x,t,{\bm{\theta}})^2$ are globally Lipschitz continuous in $x\in\mathbb{R}^d, t\in[0,T]$ and ${\bm{\theta}}\in\triangle$ variables and there exist positive constants $\omega, L>0$ such that $\omega\le \frac12\sigma(x,t,{\bm{\theta}})^2\le L$ for any $x\in\mathbb{R}^d, t\in[0,T]$, and ${\bm{\theta}}\in\triangle$. 

Then, $\alpha\in C^{0,1}(\mathcal{D})$. Moreover, the function $\alpha$ is strictly increasing, and
\begin{equation}
0<\omega\le \frac{\alpha(x,\tau, \varphi_2)-\alpha(x,\tau, \varphi_1)}{\varphi_2-\varphi_1} \le L, 
\quad \text{for any}\ (x,\tau,\varphi_i)\in\mathcal{D}, 
\label{lipschitz}
\end{equation}
i.e., $\omega\le\alpha^\prime_\varphi(x,\tau,\varphi)\le L$, and
\begin{equation}
|\nabla_x\alpha(x,\tau,\varphi)| \le  p(x,\tau) + L_{0} |\varphi|, 
\label{lipschitz-x}
\end{equation}
for a.e., $(x,\tau,\varphi)\in\mathcal{D}$, where $p(x,\tau) := \max_{ {\bm{\theta}} \in \triangle} |\nabla_x\mu(x,t,{\bm{\theta}})|$
and $ L_0 := \max_{{\bm{\theta}}\in\triangle, t\in[0,T], x\in\mathbb{R}^d} |\nabla_x \sigma^2(x, t, \theta)|$, where $t=T-\tau$.
\end{proposition}

\begin{proof}
First, let $\alpha^{\bm{\theta}}(x,\tau,\varphi) := 
-\mu(x,t,{\bm{\theta}}) +  \frac{\varphi}{2}\sigma(x,t,{\bm{\theta}})^2$, where $t=T-\tau$. Then,
\[
\alpha(x,\tau,\varphi) = \min_{ {\bm{\theta}} \in \triangle} \alpha^{\bm{\theta}}(x,\tau,\varphi) \,.
\]
For any given ${\bm{\theta}}\in \triangle$, the function $\alpha^{\bm{\theta}}(x,\tau,\varphi)$ is globally Lipschitz continuous in all variables. Therefore, the minimal function $\alpha$ is globally Lipschitz continuous. Moreover, the function $\alpha^{\bm{\theta}}(x,\tau,\varphi)$ satisfies the inequality (\ref{lipschitz}) for any ${\bm{\theta}}\in \triangle$, and so does the minimal function $\alpha$.

Next, we prove inequality (\ref{lipschitz-x}). Let $ x_1 ,x_2\in\mathbb{R}^d$ such that $x_2 = x_1 + h e^i$, where $e^i, i=1,\cdots, d,$ is the standard normal vector, i.e., $ e^i = (0,0,...,0,1,0,...,0)^T$. Then, we have 
\begin{eqnarray*}
&& \alpha^{\bm{\theta}}(x_1,\tau,\varphi)  - \alpha^{\bm{\theta}}(x_2,\tau,\varphi) = -(\mu(x_1,\tau,{\bm{\theta}})- \mu(x_2,\tau,{\bm{\theta}})) 
+ \frac{\varphi}{2} (\sigma(x_1,\tau,{\bm{\theta}})^2- \sigma(x_2,\tau,{\bm{\theta}})^2) 
\\
&& = \int_{0}^{h} (-\partial_{x_i}\mu(x_1+\xi e^i,\tau,{\bm{\theta}})) d\xi + \int_{0}^{h}  \frac{\varphi}{2} \partial_{x_i}\sigma^2(x_1+\xi e^i,\tau,{\bm{\theta}}) d\xi
\\
&&
\leq \int_{0}^{h} |\partial_{x_i}\mu(x_1+\xi e^i,\tau,{\bm{\theta}}))| d\xi + \int_{0}^{h}  \frac{|\varphi|}{2} |\partial_{x_i}\sigma^2(x_1+\xi e^i,\tau,{\bm{\theta}})| d\xi
\\
&& 
\le \max_{\bm{\theta}\in\triangle, 0\le\xi\le h} |\partial_{x_i}\mu(x_1+\xi e^i,\tau,{\bm{\theta}})| \, h  + \max_{{\bm{\theta}}\in\triangle, x\in\mathbb{R}^d} |\partial_{x_i} \sigma^2(x, \tau, \theta)|\frac{|\varphi|}{2}h.
\end{eqnarray*}
Hence,
\[
\alpha^{\bm{\theta}}(x_1,\tau,\varphi) \le  \alpha^{\bm{\theta}}(x_2,\tau,\varphi) +  \max_{{\bm{\theta}}\in\triangle, 0\le\xi\le h} |\partial_{x_i}\mu(x_1+\xi e^i,\tau,{\bm{\theta}})| \, h + \max_{{\bm{\theta}}\in\triangle, x\in\mathbb{R}^d} |\partial_{x_i} \sigma^2(x, \tau, \theta)||\varphi|h.
\]
We note that $x_2 - x_1 = h e^i$ so that $ |x_2 - x_1| = h$. Taking minimum over ${\bm{\theta}}\in\triangle$, we obtain 
\[
\alpha(x_1,\tau,\varphi) \le  \alpha(x_2,\tau,\varphi) +  \max_{\bm{\theta}\in\triangle,  0\le\xi\le h} |\partial_{x_i}\mu(x_1+ \xi e^i,\tau,{\bm{\theta}})| \, h + \max_{{\bm{\theta}}\in\triangle, x\in\mathbb{R}^d} |\partial_{x_i} \sigma^2(x, \tau, \theta)||\varphi| h .
\]
By exchanging the role of $x_1$ and $x_2$ and taking the limit as $x_2\to x_1$, i.e., $h\to 0$, we obtain inequality (\ref{lipschitz-x}), as stated.

\end{proof}

According to Proposition~\ref{smootheness0}, the value function $\alpha$ given in (\ref{eq_alpha_def}) satisfies the assumptions of Theorem~\ref{th:alpha-existence} provided that the functions 
\[
p(x,\tau)= \max_{{\bm{\theta}}\in\triangle} |\nabla_x\mu(x,\tau,{\bm{\theta}})|, \quad\text{and}\ 
h(x,\tau)= \alpha(x,\tau,0) = -\max_{{\bm{\theta}}\in\triangle} \mu(x,\tau,{\bm{\theta}})
\]
belong to the Banach space $L^\infty((0,T); H)$.\\

The next result was proved in \cite{KilianovaSevcovicJIAM}. It gives sufficient conditions imposed on the decision set $\triangle$ and functions $\mu$ and $\sigma$ that guarantee higher smoothness of the value function $\alpha$. Its proof is based on the classical envelope theorem due to Milgrom and Segal \cite{milgrom_segal2002} and the result on the Lipschitz continuity of the minimizer $\hat{\bm{\theta}}=\hat{\bm{\theta}}(x,\tau,\varphi)$ belonging to a convex compact set $\triangle$ according to Klatte \cite{Klatte}.

\begin{theorem}\cite[Theorem 1]{KilianovaSevcovicJIAM}
\label{smootheness1}
Suppose that $\triangle\subset\mathbb{R}^n$ is a convex compact set and the functions $\mu(x,t,{\bm{\theta}})$ and $\sigma(x,t,{\bm{\theta}})^2$ are $C^{1,1}$ smooth such that the objective function 
$f(x,t,\varphi, {\bm{\theta}}):= - \mu(x,t,{\bm{\theta}}) + \frac{\varphi}{2} \sigma(x,t,{\bm{\theta}})^2$ is strictly convex in the variable ${\bm{\theta}}\in\triangle$ for any  $\varphi\in(\varphi_{min}, \infty)$. Then, the function $\alpha$ belongs to the space $C^{1,1}(\mathcal{D})$.
\end{theorem}

\subsubsection{Point-wise a-priori estimates of solution with their existence and uniqueness}

In this section, we present a-priori estimates on a solution $\varphi$ to the Cauchy problem (\ref{eq_PDEphi-cons}). We will assume that the function $\mu$ is independent of time $t\in[0,T]$, and $\sigma$ is independent of $x\in\mathbb{R}$ and $t\in[0,T]$, i.e., 
\[
\mu=\mu(x,\bm{\theta}), \qquad \sigma = \sigma(\bm{\theta}).
\]
Then, the value function $\alpha=\alpha(x,\varphi)$ is also independent of the $\tau=T-t$ variable. Next, we prove a-priori estimates for the transformed function $\psi=\psi(x,\tau)$ defined as $\psi(x,\tau)=\alpha(x,\varphi(x,\tau))$. Since the function $\alpha$ is strictly increasing in the $\varphi$ variable, there exists an inverse function $\beta(x,\psi)$ such that $\alpha(x,\beta(x,\psi)) = \psi$. Therefore, the function $\varphi(x,\tau)$ is a solution to (\ref{eq_PDEphi-cons}) if and only if the function $\psi(x,\tau)$ is a solution to the following linear parabolic PDE: 
\[
-\partial_\tau\psi + a(x,\tau) \partial_x^2 \psi + b(x,\tau) \partial_x \psi + c(x,\tau)\psi  = 0,
\]
where
\begin{eqnarray*}
a(x,\tau) &=& \alpha^\prime_\varphi(x,\varphi(x,\tau)), 
\quad
b(x,\tau) = - \alpha^\prime_\varphi(x,\varphi(x,\tau)) \varphi(x,\tau) - \alpha(x,\varphi(x,\tau)),
\\
c(x,\tau) &=& \alpha^\prime_x (x,\varphi(x,\tau)).
\end{eqnarray*}
Note that $0<\omega \le a(x,\tau)\le L$. Next, suppose that the function $c(x,\tau)$ is bounded from above by a constant $\lambda\ge 0$. Then, the function $\psi_\lambda(x,\tau) = \psi(x,\tau) e^{-\lambda\tau}$
is a solution to the linear PDE:
\[
{\mathcal L}[\psi_\lambda] = 0, \qquad \text{where}\ \  {\mathcal L}[\psi_\lambda] \equiv \partial_\tau\psi_\lambda - a(x,\tau) \partial_x^2 \psi_\lambda - b(x,\tau) \partial_x \psi_\lambda - c_\lambda(x,\tau)\psi_\lambda,
\]
where $c_\lambda(x,\tau)= c(x,\tau) -\lambda$ is nonpositive, i.e., $c_\lambda (x,\tau) \le0$ for all $x,\tau$.
Let $\underline{\psi}\le 0$ be a constant. Then, ${\mathcal L}[\psi_\lambda - \underline{\psi}] =  c_\lambda(x,\tau)\underline{\psi}\ge 0$. We employ the maximum principle for parabolic equations on unbounded domains according to Meyer and Needham \cite[Theorem 3.4]{Meyer} to obtain $\psi_\lambda(x,\tau) - \underline{\psi}\ge 0$ for all $x,\tau$, provided that $\psi_\lambda(x,0) - \underline{\psi}=\psi(x,0) - \underline{\psi}\ge 0$ for all $x$. In other words, $\underline{\psi}$ is a subsolution. Similarly, if $\overline{\psi}\ge 0$ is a given constant, then ${\mathcal L}[\psi_\lambda - \overline{\psi}] =  c_\lambda(x,\tau)\overline{\psi}\le 0$ and $\psi_\lambda(x,\tau) - \overline{\psi}\le 0$ for all $x,\tau$, provided that $\psi(x,0) - \overline{\psi}\le 0$ for all $x$, i.e., $\overline{\psi} $ is a supersolution. Consequently, we have the following implication:
\[
\underline{\psi}\le \psi(x,0) \le \overline{\psi} \quad \Longrightarrow \quad 
\underline{\psi} e ^{\lambda \tau}\le \psi(x,\tau) \le \overline{\psi} e ^{\lambda \tau} \quad\text{for all}\ x\in\mathbb{R}, \tau\in[0,T].
\]
In terms of the solution $\varphi$ to the Cauchy problem (\ref{eq_PDEphi-cons})--(\ref{init_PDEphi-cons}), we obtain the following a-priori estimate:
\begin{equation}
\underline{\psi} e ^{\lambda \tau}\le \alpha(x,\varphi(x,\tau)) \le \overline{\psi} e ^{\lambda \tau} \quad\text{for all}\ x\in\mathbb{R}, \tau\in[0,T], 
\label{pointwise}
\end{equation}
where 
\begin{equation}
\underline{\psi} = \min\{ 0, \inf_{x\in\mathbb{R}} \alpha(x,\varphi(x,0)) \}, \qquad 
\overline{\psi} = \max\{0, \sup_{x\in\mathbb{R}} \alpha(x,\varphi(x,0)) \}. 
\label{psiplusminus}
\end{equation}

Now, we can apply the general Theorem~\ref{th:alpha-existence} on existence and uniqueness of a solution. 

\begin{theorem} \cite[Theorem 5]{udeani2021application}
Let the decision set $\triangle \subset \mathbb{R}^n$ be compact and $u:\mathbb{R}\to\mathbb{R}$ be an increasing utility function such that $\varphi_0(x) = -u''(x)/u'(x)$ belongs to the space $L^2(\mathbb{R})\cap L^\infty(\mathbb{R})$. Suppose that the drift $\mu(x, \bm{\theta})$ and volatility function $\sigma^2(\bm{\theta})>0$ are $C^1$ continuous in the $x$ and $\bm{\theta}$ variables and the value function $\alpha(x,\varphi)$ given in (\ref{eq_alpha_def}) satisfies $p\in L^2(\mathbb{R})\cap L^\infty(\mathbb{R}), h\in L^\infty(\mathbb{R}), \; \text{and} \;\partial^2_x h \in L^2(\mathbb{R})$, where  
\begin{eqnarray*}
&&p(x)= \max_{{\bm{\theta}}\in\triangle} |\partial_x\mu(x,{\bm{\theta}})|, 
\quad 
h(x)=  -\max_{{\bm{\theta}}\in\triangle} \mu(x,{\bm{\theta}}).
\end{eqnarray*}
Then, for any $T>0$, there exists a unique solution $\varphi$ to the Cauchy problem 

\begin{equation}
\label{finalequation}
\partial_\tau \varphi - \partial^2_x \alpha(\cdot,\varphi) = - \partial_x \left( \alpha(\cdot,\varphi)\varphi\right),
\quad \varphi(x,0) = \varphi_0(x),\quad (x,\tau)\in\mathbb{R}\times(0,T), 
\end{equation}
satisfying $\varphi\in C([0,T]; H)\cap L^2((0,T); V)\cap L^\infty((0,T)\times \mathbb{R})$.

\end{theorem}

\begin{proof}

Since $\sigma^2(\bm{\theta})>0$ and $\triangle$ is a compact set, there exist constants $0<\omega\le L$ such that $0<\omega\le \sigma^2(\bm{\theta})\le L$ for all $\bm{\theta}\in \triangle$. It follows from Proposition~\ref{smootheness0} that 
\begin{equation}
\omega|\varphi| \le |\alpha(x,\varphi) - \alpha(x,0)| \le L |\varphi|.
\label{monotonicity}
\end{equation}
Since $\varphi_0, h\in L^\infty(\mathbb{R})$ and $h(x) = \alpha(x,0)$, we obtain 
$M:=\sup_{x\in\mathbb{R}} |\alpha(x,\varphi_0(x))| <\infty$.

Now, let us define the shift diffusion function by $\tilde\alpha(x,\varphi)=\alpha(x,\varphi) - \alpha(x,0)$. Note that $\alpha(x,0)= \min_{{\bm{\theta}}\in\triangle} -\mu(x,{\bm{\theta}}) = h(x)$. Then, equation (\ref{finalequation}) is equivalent to 
\[
\partial_\tau \varphi  + A \tilde \alpha(\cdot,\varphi) = 
\tilde \alpha(\cdot,\varphi)+ \partial^2_x h - \partial_x \left( \alpha(\cdot,\varphi)\varphi\right),
\]
where $A=I-\partial^2_x$.

Next, let $g_0(\varphi) = \tilde \alpha(\cdot,\varphi)+ \partial^2_x h$ and $g_1(\varphi) = -w(\alpha(\cdot,\varphi)) \varphi$. Here, $w:\mathbb{R}\to\mathbb{R}$ is a suitable cutoff function 
\[
w(\alpha) = \left\{ 
\begin{array}{ll}
\underline{\psi} e^{\lambda T}, & \text{if}\ \ \alpha  \le \underline{\psi} e^{\lambda T},  \\
\alpha, &  \text{if} \ \ \underline{\psi} e^{\lambda T} < \alpha < \overline{\psi} e^{\lambda T}, \\
\overline{\psi} e^{\lambda T}, & \text{if}\ \ \alpha \ge \overline{\psi} e^{\lambda T},  \\
\end{array}
\right.
\]
where $\overline{\psi}=M, \underline{\psi}=-M$.
Then, the functions $g_0, g_1:H\to H$ are globally Lipschitz continuous. 

Note that the diffusion function $\tilde\alpha$ satisfies the assumptions of Theorem~\ref{th:alpha-existence} with $\tilde h(x)=\tilde\alpha(x,0)\equiv 0$. Now, applying Theorems~\ref{th:alpha-existence} and \ref{cor:abs-continuous}, we obtain the existence and uniqueness of a solution $\varphi\in C([0,T]; H)\cap L^2((0,T); V)$ to the Cauchy problem (\ref{equ:g_0,g_1}). The solution $\varphi$ satisfies the point-wise estimate (\ref{pointwise}). Thus, $w(\alpha(x,\varphi(x,\tau))) = \alpha(x,\varphi(x,\tau))$, and $\varphi$ is a solution to the Cauchy problem (\ref{finalequation}). 

Finally, from (\ref{monotonicity}), we deduce the $L^\infty((0,T)\times \mathbb{R})$ estimate for the solution $\varphi$ since $\sup_{x\in\mathbb{R}} |\alpha(x,\varphi(x,\tau))| \le M e^{\lambda \tau}$, where $\lambda = \sup_{x\in\mathbb{R}} p(x)$. Moreover, $\varphi\in L^\infty((0,T)\times \mathbb{R})$, and
\[
\sup_{x\in\mathbb{R}, \tau\in [0,T]} |\varphi(x,\tau)| \le \omega^{-1} 
(M e^{\lambda T} + \max_{x\in\mathbb{R}} |h(x)|).
\]

\end{proof}

\subsubsection{Application to stochastic dynamic optimal portfolio selection problem}
\label{application}

This subsection presents examples of the stochastic process (\ref{process_x}). We also present examples of the corresponding utility function. For the the stochastic process (\ref{process_x}), we can consider a portfolio optimization problem with regular cash inflow ($\varepsilon>0$)/outflow ($\varepsilon<0$) to a portfolio representing pension funds savings. Here, we consider Slovak pension fund savings according to Kilianov\'a and {\v S}ev{\v c}ovi{\v c} \cite{KilianovaSevcovicANZIAM}). It is well-known that in a stylized financial market, the stochastic process $\{y_t \}_{t\ge0}$ driven by the stochastic differential equation
\begin{equation}
d y_t = (\varepsilon(y_t) + {\bm{\mu}}^T {\bm{\theta}} y_t)
d t + \sigma(\bm{\theta}) y_t d W_t, \label{processYeps}
\end{equation}
represents a stochastic evolution of the value of a synthetized portfolio $y_t$ consisting of $n$-assets with weights ${\bm{\theta}}=(\theta_1, \cdots, \theta_n)^T$, mean returns ${\bm{\mu}}=(\mu_1, \cdots, \mu_n)^T$, and an $n\times n$ positive definite covariance matrix ${\bm{\Sigma}}$, i.e., $\sigma({\bm{\theta}})^2 ={\bm{\theta}}^T {\bm{\Sigma}}  {\bm{\theta}}$.

In this study, we assume that the value of the inflow/outflow rate $\varepsilon=\varepsilon(y)$ also depends on the $y$ such that it is characterized by the following 
\[
\varepsilon(y)=\left\{ 
\begin{array}{ll}
0,     & \hbox{if}\ 0<y<y_{-},\\
 C,     & \hbox{if} \ y_{-} < y,
\end{array}
\right.
\]
where $C$ is a constant. It represents a realistic pension saving model in which there is no inflow/outflow provided that the value $y$ of the portfolio is very small. Based on the logarithmic transformation $x=\ln y$ and It\^{o}'s lemma, the stochastic process $\{x_t\}$ satisfies (\ref{process_x}), where 
$\mu(x,{\bm{\theta}}) = {\bm{\mu}}^T{\bm{\theta}} - \frac12 \sigma({\bm{\theta}})^2  +\varepsilon(e^x)  e^{-x}$.

Recently, Kilianov\'a and Trnovsk\'a \cite{KilianovaTrnovska} established a further generalization of the drift and volatility functions arising from the so-called worst-case portfolio optimization problem. In such case, the volatility function is given by
\[
\sigma({\bm{\theta}})^2 = \max_{{\bm{\Sigma}}\in{\mathcal K}}{\bm{\theta}}^T {\bm{\Sigma}}  {\bm{\theta}}, 
\]
where $\mathcal K$ is a bounded uncertainty convex set of positive definite covariance matrices. In general, only a part of the covariance matrix can be calculated precisely, whereas entries are not precisely determined. For instance, if only the diagonal $d$ is known, we have ${\mathcal K}=\{{\bm{\Sigma}}\succ0, \ diag({\bm{\Sigma}})=d\}$. Furthermore, the drift function is given by
\[
\mu(x,{\bm{\theta}}) = \min_{{\bm{\mu}}\in{\mathcal E}}{\bm{\mu}}^T{\bm{\theta}} - \frac12 \sigma({\bm{\theta}})^2  +\varepsilon(e^x) e^{-x} ,
\]
where $\mathcal E$ is a given bounded uncertainty convex set of mean returns.

\begin{remark}

Consider a class of utility function characterized by a pair of exponential functions: 
\begin{equation}
u(x) = \begin{cases} - e^{-a_0 x} - c^*, & x \le x^\ast, \\ 
- (a_0/a_1) e^{-a_1 x + (a_1-a_0)x^\ast}, & x > x^\ast, \end{cases} 
\label{eq:utility_DARA}
\end{equation}
where $c^*=e^{-a_0 x^*}(a_0-a_1)/a_1$ and $a_0, a_1\in \mathbb{R}>0$ are given constants. Here, $x^\ast \in \mathbb{R}$ is a point at which the risk aversion changes. We observe that $u$ is an increasing $C^1$ function with a jump in the second derivative at the point $x^\ast$. 

If $a_0>a_1>0$, then the utility function $u$ is called the decreasing absolute risk aversion (DARA) function, representing an investor with a non-constant decreasing risk aversion. It means that the higher the wealth of investors, the lower their risk aversion, and hence the higher exposition of the portfolio to more risky assets. According to Post, Fang and Kopa \cite{Post}, the piece-wise exponential DARA utility function plays an essential role in the analysis of decreasing absolute risk aversion stochastic dominance introduced by Vickson \cite{Vickson} (see also \cite{KilianovaSevcovicKybernetika}). It is worth noting that the coefficients of absolute risk aversion of the above utility functions $-u^{\prime\prime}(x)/u^{\prime}(x)$ is equal to $a_0$ if $x\le x^*$ or to $a_1$ if $x> x^*$. 

The piece-wise constant function $\varphi_0(x)=-u''(x)/u'(x)$ should be truncated outside of some interval $(-\gamma,\gamma)$, where $\gamma$ is large enough. Then, $\varphi_0\in L^2(\mathbb{R})\cap L^\infty(\mathbb{R})$. Therefore, the underlying utility function is modified by linear functions for $x<-\gamma$ and $x>\gamma$.

Another simple example of a convex-concave utility function is the function $u(x)=\arctan(x)$. Then, $\varphi_0(x) = -u''(x)/u'(x) = 2x/(1+x^2)$. Clearly, $\varphi_0\in H=L^2(\mathbb{R})$. It is worth noting that the individual's reduction in marginal utility arising from a loss is absolutely greater than the marginal utility from a financial gain. In the domain of gain ($x>x^*$), the utility function is concave, indicating that investors show risk aversion in this domain. In contrast, investors become risk-seeker when dealing with losses, i.e., the utility function is convex for $x\leq x^*$. 
\end{remark}

\subsubsection{Numerical examples}

This subsection presents examples of the value function $\alpha=\alpha(\varphi)$ to the parametric optimization problem with different decision sets. First, we consider a simple decision set $\triangle=\{{\bm{\theta}}\in\mathbb{R}^2,\,  {\bm{\theta}} \ge0, \bm{1}^T {\bm{\theta}} =1\}, n=2, \mu({\bm{\theta}})={\bm{\mu}}^T{\bm{\theta}}, \sigma^2({\bm{\theta}})={\bm{\theta}}^T{\bm{\Sigma}}{\bm{\theta}}$, where ${\bm{\Sigma}}$ is a positive definite covariance matrix, and ${\bm{\mu}}$ is a positive vector of mean return. The value function $\alpha=\alpha(\varphi)$ can be explicitly expressed as follows:
\[
\alpha(\varphi)=\left\{ 
\begin{array}{ll}
E^- \varphi + D^-,     & \hbox{if}\ 0<\varphi\le \varphi_*^-,\\
 A - \frac{B}{\varphi} + C\varphi,     & \hbox{if} \ \varphi_*^-<\varphi<\varphi_*^+, \\
E^+ \varphi + D^+,     & \hbox{if}\ \varphi_*^+\le \varphi.
\end{array}
\right.
\]
Here, $(\varphi_*^-,\varphi_*^+)$ is the maximal interval where the optimal value $\hat{\bm{\theta}}(\varphi)\in\triangle$ of the function ${\bm{\theta}} \mapsto -{\bm{\mu}}^T{\bm{\theta}} +\frac{\varphi}{2} {\bm{\theta}}^T{\bm{\Sigma}}{\bm{\theta}}$ is strictly positive ($\hat{\bm{\theta}}(\varphi)>0$) for $\varphi\in (\varphi_*^-,\varphi_*^+)$, and $C, E^\pm>0$, $B\ge 0$, $A,D^\pm$ are constants explicitly depending on the covariance matrix ${\bm{\Sigma}}$ and the vector of mean return ${\bm{\mu}}$ such that the function $\alpha$ is $C^1$ continuous at $\varphi_*^\pm$, i.e., $E^\pm=B/(\varphi_*^\pm)^2+C$ and $D^\pm=A-B/\varphi_*^\pm +C\varphi^\pm - E^\pm\varphi_*^\pm$. It is clear that $\alpha$ is only $C^{1,1}$ continuous function with two points $\varphi_*^\pm$ of discontinuity of the second derivative $\alpha''$. 

Furthermore, consider a decision set consisting of finite number of points. Then, the value function $\alpha(\varphi)$ corresponding to such decision set is only piece-wise linear. In other words, if $\hat\triangle = \{{\bm{\theta}}^1, \cdots, {\bm{\theta}}^k\} \subset \{{\bm{\theta}}\in\mathbb{R}^2,\,  {\bm{\theta}} \ge0, \bm{1}^T {\bm{\theta}} =1\}$, then $\alpha(\varphi)=\min_{i=1,\cdots,k} \alpha^{i}(\varphi)$, where $\alpha^i(\varphi) = E^i \varphi + D^i$ is a linear function with the slope $E^i = (1/2) ({\bm{\theta}}^i)^T{\bm{\Sigma}} {\bm{\theta}}^i >0$ and intercept $D^i = - {\bm{\mu}}^T {\bm{\theta}}^i$.

Figure~\ref{fig:alpha_alphader_alphaderder} a) shows the graph of the value function $\alpha$ corresponding to the Slovak pension fund system with the two types of decision sets. According to the data-set obtained from \cite{KMS}, the portfolio consists of the stock index (with a high mean return $\mu_s=0.10$ and high volatility $\sigma_s=0.3$) and bonds (with mean return $\mu_b=0.03$ and very low volatility $\sigma_s=0.01$). The returns on stocks index and bonds are negatively correlated $\varrho=-0.15$. ${\bm{\mu}}=(\mu_s,\mu_b)^T=(0.1, 0.05)^T$. Then ${\bm{\Sigma}}_{11}=\sigma_2^2, {\bm{\Sigma}}_{22}=\sigma_b^2, {\bm{\Sigma}}_{12}={\bm{\Sigma}}_{21}=\varrho\sigma_s\sigma_b$. As shown in Figure~\ref{fig:alpha_alphader_alphaderder} a), the solid blue line corresponds to the convex compact decision set $\triangle=\{{\bm{\theta}}\in\mathbb{R}^2,\,  {\bm{\theta}} \ge0, \bm{1}^T {\bm{\theta}} =1\}$. The piece-wise linear value function $\alpha$ (dotted red line) corresponds to the discrete decision set $\hat\triangle=\{ {\bm{\theta}}^1, {\bm{\theta}}^2, {\bm{\theta}}^3\}\subset \triangle$. It represents the Slovak pension fund system consisting of three funds: growth funds with ${\bm{\theta}}^1=(0.8,0.2)^T$ (80\% of stocks and 20\% of bonds), balanced funds with  ${\bm{\theta}}^2=(0.5, 0.5)^T$ (equal proportion of stocks and bonds), and conservative funds with ${\bm{\theta}}^3=(0, 1)^T$ (only bonds) (c.f. \cite{KMS}). Figure~\ref{fig:alpha_alphader_alphaderder} b) shows the graph of the second derivative $\alpha''_\varphi(\varphi)$ of the value function $\alpha(\varphi)$ corresponding to the convex compact decision set $\triangle$. It has the first point of discontinuity $\varphi_*^-$ close to the value 2.
For $n>2$, the number of discontinuities of $\alpha_\varphi^{\prime\prime}$ increases (c.f. \cite{KilianovaSevcovicANZIAM}). Figure ~\ref{fig:alpha_alphader_alphaderder-dax} shows another example of the value function and its second derivative for the portfolio consisting of five stocks (BASF, Bayer, Degussa-Huls, FMC, and Schering) entering DAX30 German stocks index. The covariance matrix ${\bm{\Sigma}}$ is taken from \cite{DDV}. We set the vector of yields ${\bm{\mu}}=(0.03,0.02,0.04,0.01,0.01)^T$.

\begin{figure}
\centering
\includegraphics[height=6truecm]{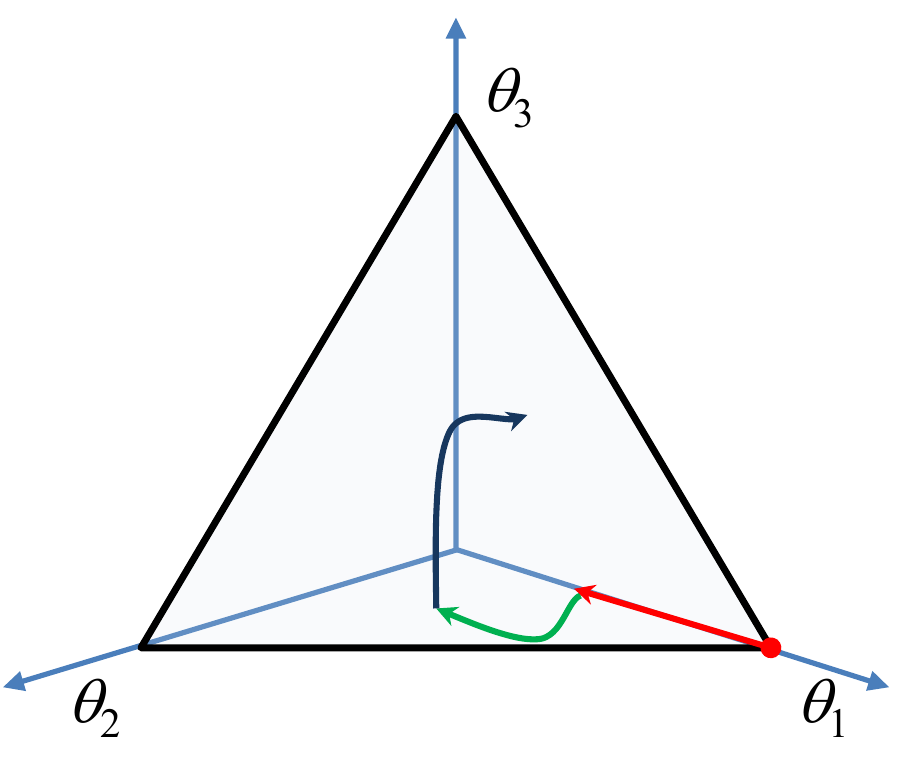}
\caption{The path $\hat{\bm{\theta}}(\varphi)$ as a function of $\varphi$.}
\label{fig:simplex}
\end{figure}

The minimizer $\hat{\bm{\theta}}(\varphi)$ of the convex optimization problem 
\[
\alpha(\varphi) = \min_{ \bm{\theta}
\in \triangle } \{ -\bm{\mu}^T \bm{\theta} +
\frac{\varphi}{2}\bm{\theta}^T \Sigma \bm{\theta}\}
\]
when considered as a function of the risk aversion parameter $\varphi$ is only Lipschitz continuous in $\varphi$. According to Millgrom--Segal envelope theorem, the derivative $\alpha'(\varphi)$ is given by $\alpha'(\varphi) = \frac{1}{2} \hat{\bm{\theta}}(\varphi)^T\Sigma\hat{\bm{\theta}}(\varphi)$. Figure~\ref{fig:simplex} shows the path $\hat{\bm{\theta}}(\varphi)$ as a function of $\varphi$ when increasing $\varphi$ from $\varphi=0$ to $\varphi \to \infty$. For small values of $\varphi$, only one asset with maximal mean return is active, i.e., $\theta_1>0, \theta_2=\theta_3=0$. For intermediate values of $\varphi$, two assets are active ${ \theta_1>0, \theta_2>0}, \theta_3=0$. Moreover, for larger values of $\varphi$, all three assets are active, i.e.,  ${ \theta_1>0, \theta_2>0, \theta_3>0}$. The path $\varphi\mapsto \hat{\bm{\theta}}(\varphi)$ has a discontinuity in the first derivative when it leaves lower dimensional object (vertex, edge) enter a higher-dimensional object volume.


\begin{figure}
    \centering
    \includegraphics[width=0.48\textwidth]{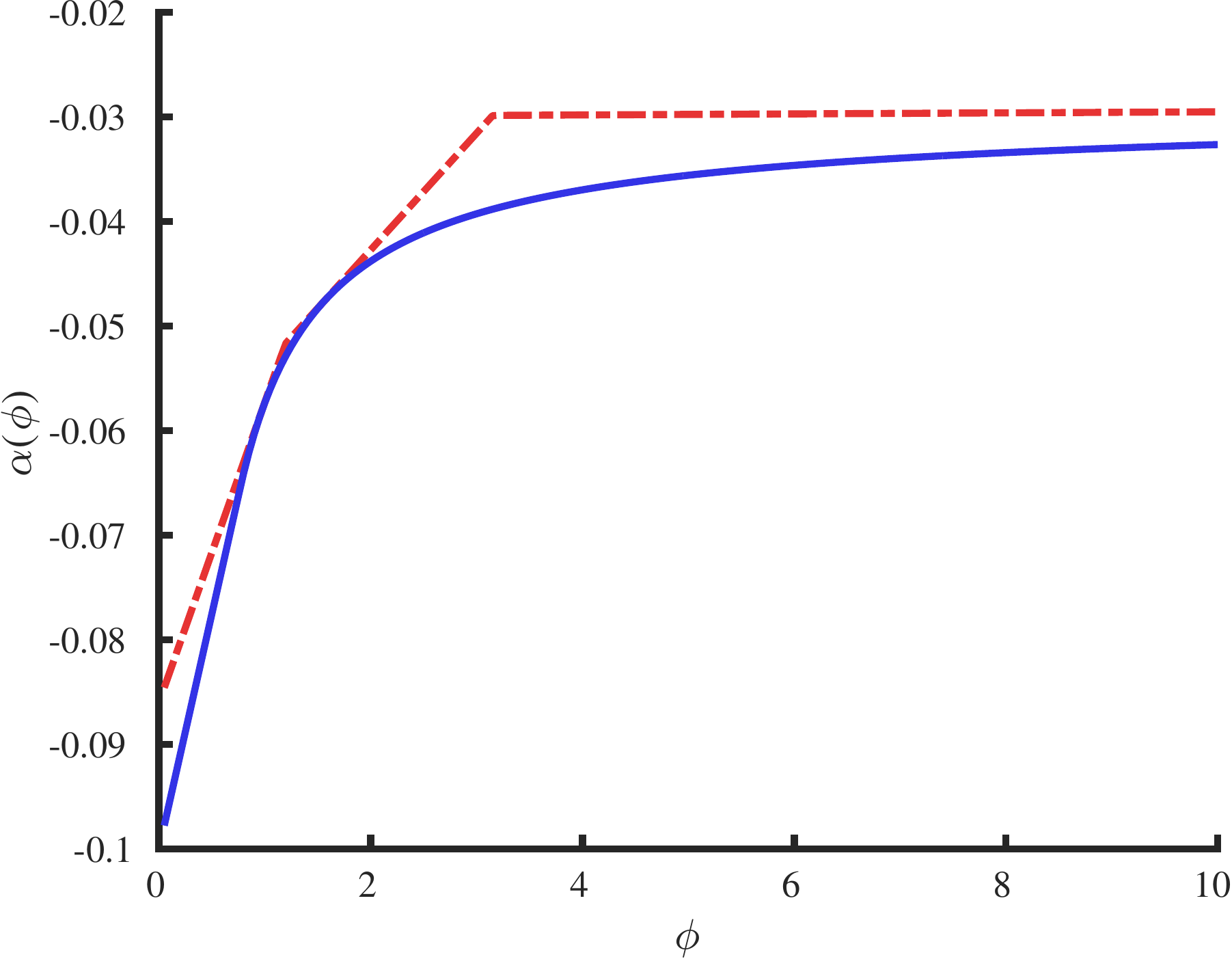}
    \includegraphics[width=0.48\textwidth]{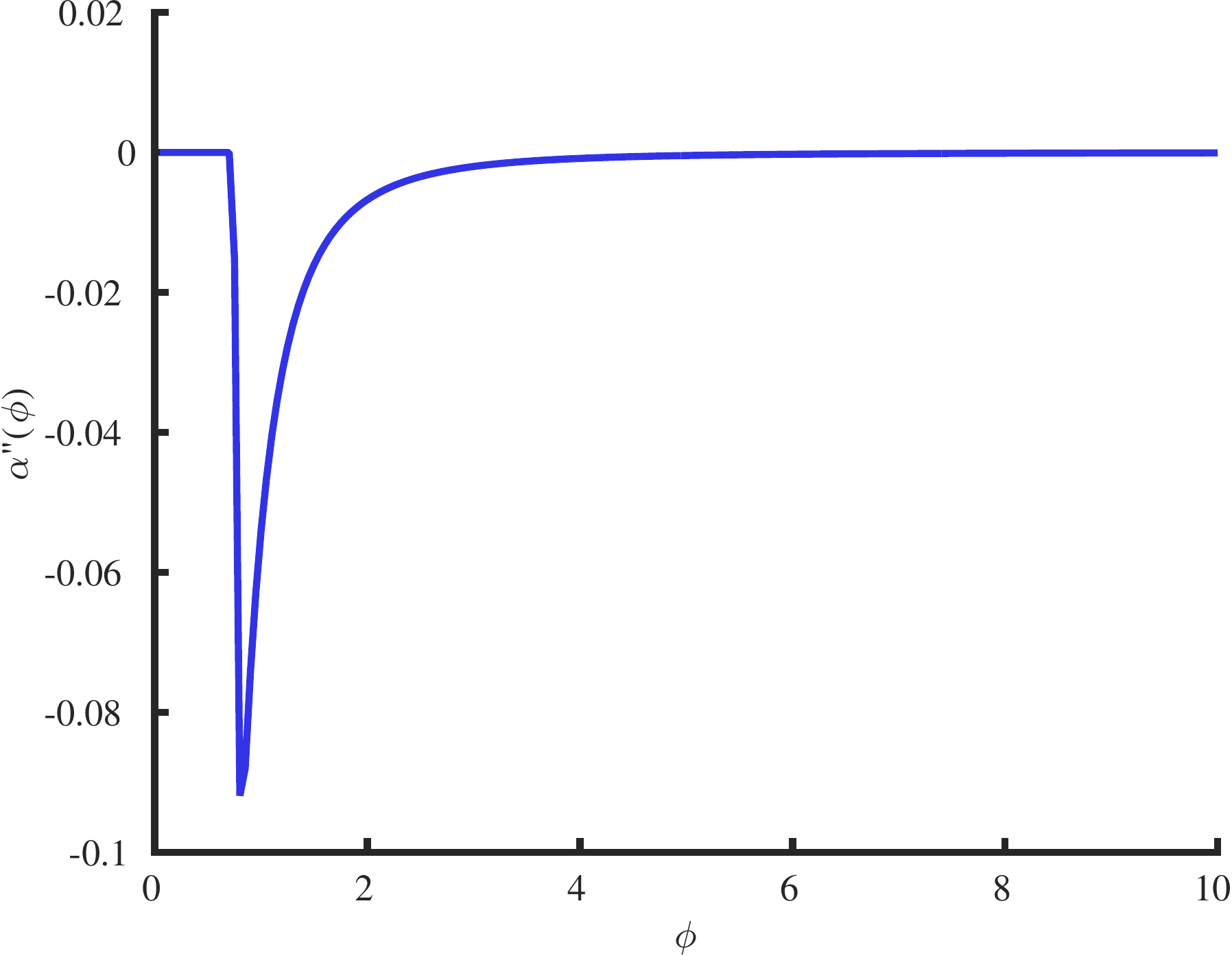}
    
    a) \hglue6truecm \qquad b)
    
    \caption{a) A graph of the value function $\alpha$, b) its second derivative $\alpha''(\varphi)$ for the portfolio consisting of the stocks index and bonds (c.f. \cite{KMS}) for the convex compact decision set $\triangle$. The dotted line in a) corresponds to the discrete decision set $\hat\triangle=\{ {\bm{\theta}}^1, {\bm{\theta}}^2, {\bm{\theta}}^3\}\subset \triangle$. Source: our computation is based on the method from \cite{udeani2021application}. }
    \label{fig:alpha_alphader_alphaderder}
\end{figure}

\begin{figure}
    \centering
    \includegraphics[width=0.48\textwidth]{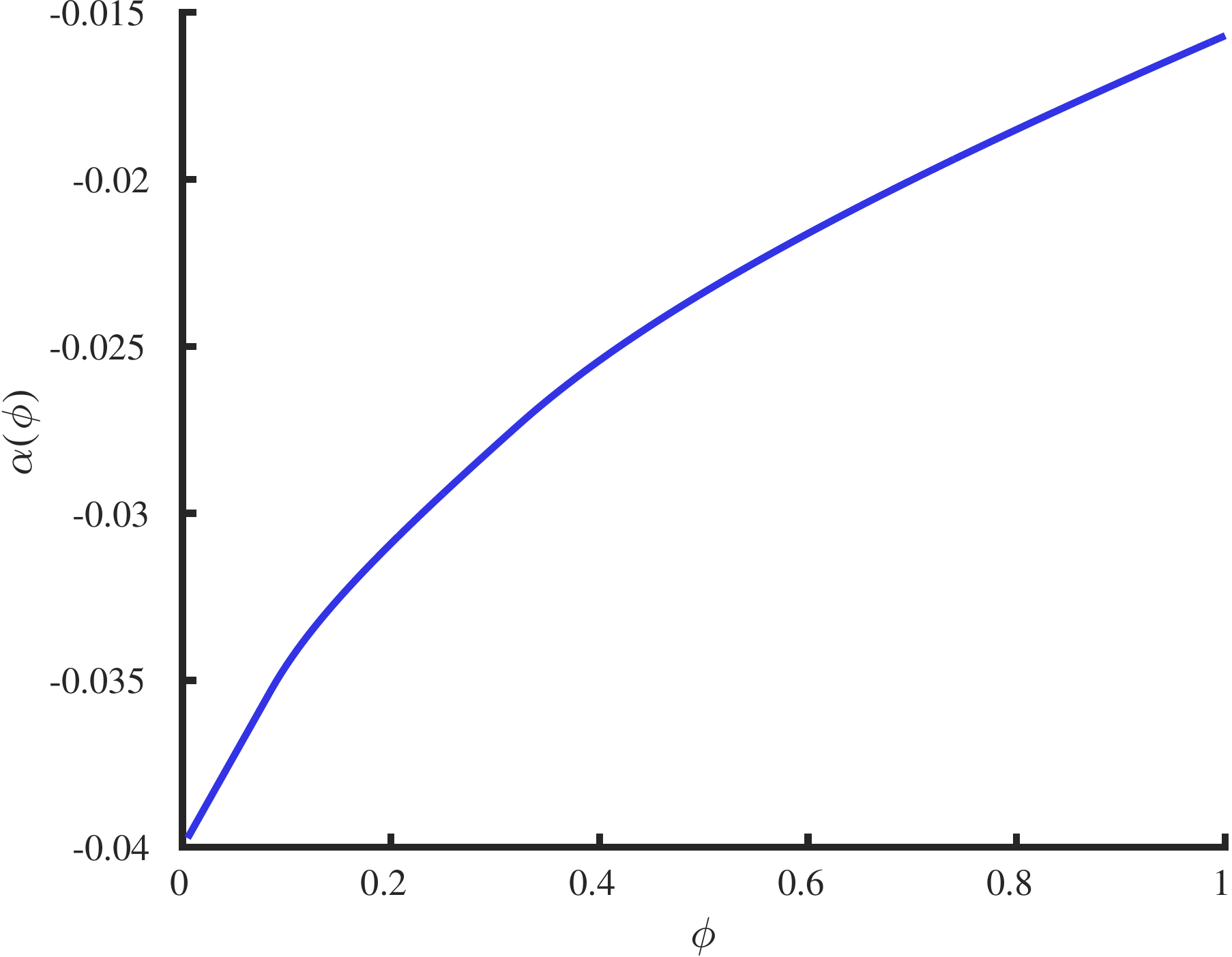}
    \includegraphics[width=0.48\textwidth]{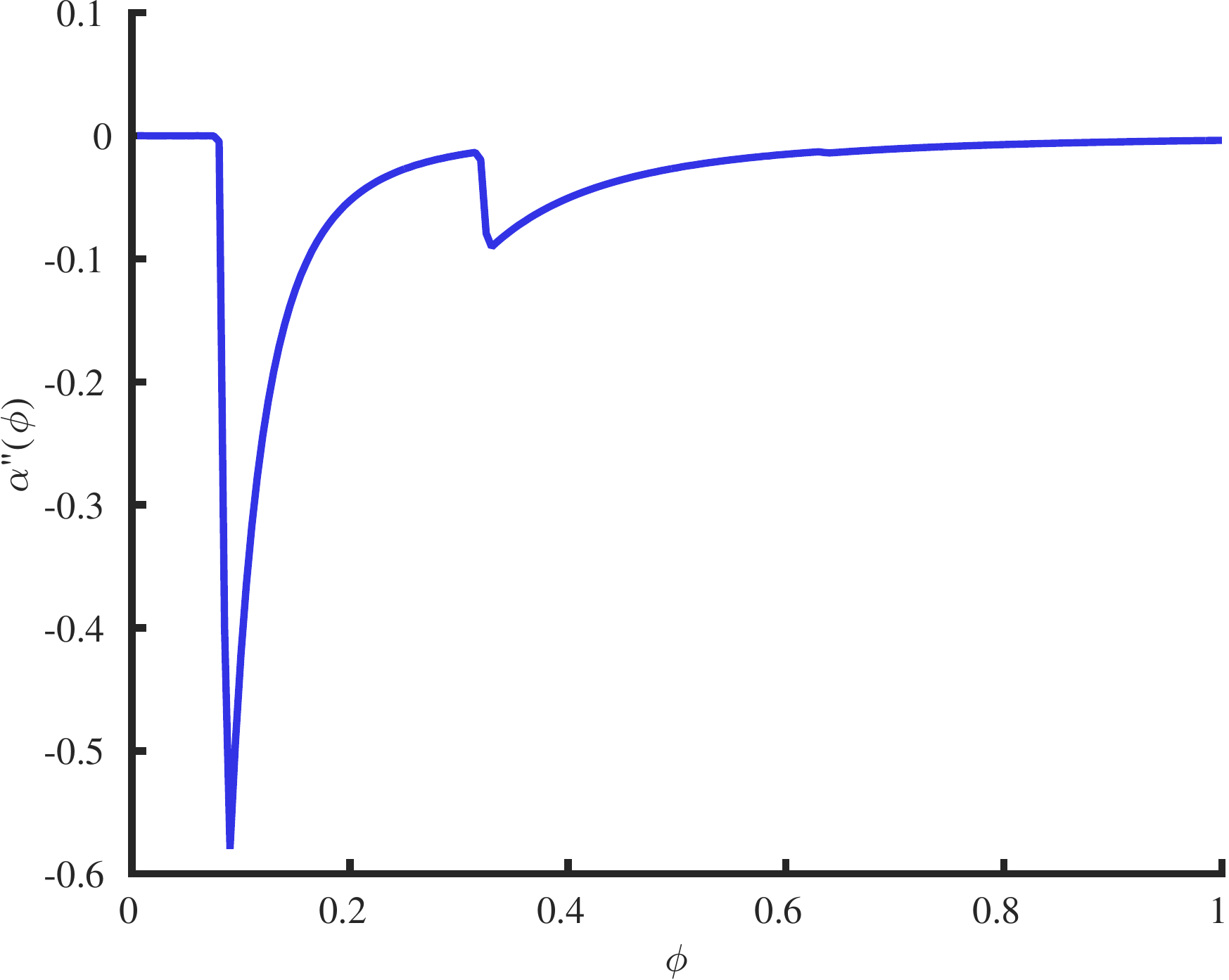}
    
    a) \hglue6truecm \qquad b)
    
    \caption{a) A graph of the value function $\alpha(\varphi)$, and b) the second derivative $\alpha''_\varphi(\varphi)$ corresponding to five stocks (BASF, Bayer, Degussa--Huls, FMC Scheringfrom) entering DAX30 index. Source: our computation is based on the method from \cite{udeani2021application}.}
    \label{fig:alpha_alphader_alphaderder-dax}
\end{figure}

The advantage of the Riccati transformation of the original HJB is twofold. First, the diffusion function $\alpha$ can be computed in advance as a result of quadratic optimization problem when the vector $\bm{\mu}$ and the covariance matrix $\bm{\Sigma}$ are given or semidefinite programming problem when they belong to a uncertainity set of returns and covariance matrices (c.f. \cite{KilianovaTrnovska}). Figure \ref{fig:vysledky-DAX} shows the vector of optimal weights $\bm{\theta}$, as a function of the parameter $\varphi$, obtained as the optimal solution to the quadratic optimization problem with the covariance matrix from \cite{DDV} corresponding to the five assets  (BASF, Bayer, Degussa--Huls, FMC, Scheringfrom) entering DAX30 index from 2008. There are more nontrivial weights $\theta_i$ when the parameter $\varphi$ increases.

\begin{figure}
    \centering
    \includegraphics[width=0.5\textwidth]{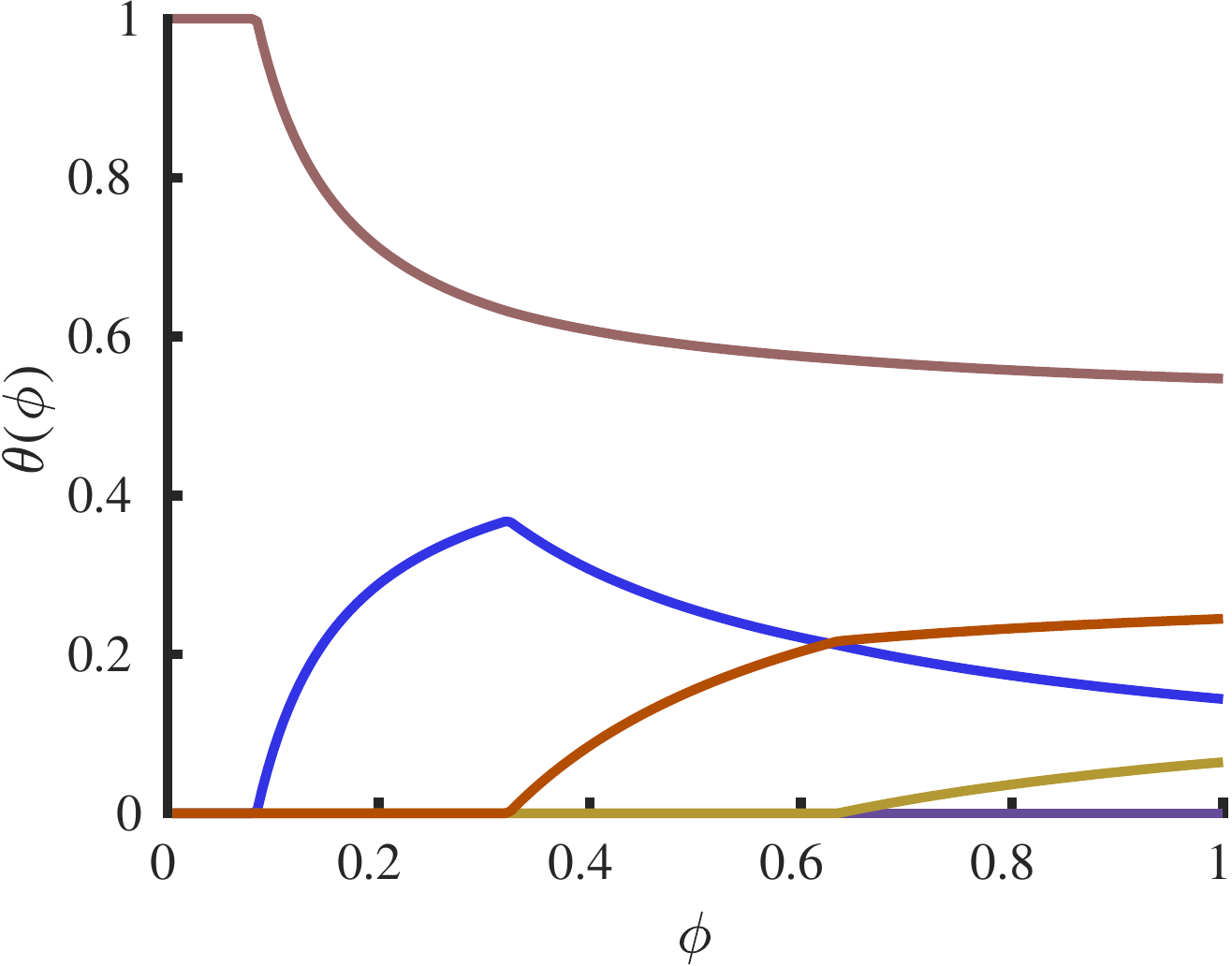} \\
    
    \caption{The optimal vector $\bm{\theta}=(\theta_1, \cdots,\theta_n)^T$ as a function of $\varphi$ for the German DAX30 index. Source: our computation is based on the method from \cite{udeani2021application}, \cite{KilianovaSevcovicANZIAM}.}
    \label{fig:vysledky-DAX}
\end{figure}

In contrast to the fully nonlinear character of the original HJB equation (\ref{eq_HJB}), the transformed equation (\ref{eq_PDEphi-cons}) represents a quasilinear parabolic equation in the divergence form. Thus, efficient numerical schemes can be constructed for this class of equation. In our computational experiments, we employ the finite volume discretization scheme proposed and investigated by Kilianov\'a and \v{S}ev\v{c}ovi\v{c} \cite{KilianovaSevcovicKybernetika, KilianovaSevcovicANZIAM, KilianovaSevcovicJIAM}). Figure~\ref{fig:vysledky1} a) shows the results of time dependent sequence of profiles $\varphi(x,\tau)$ for a constant initial condition $\varphi_0\equiv 9$. Figure~\ref{fig:vysledky1} b) shows the solution profiles for the initial condition $\varphi_0$ attaining four decreasing values $\{9,8\}$. It represents DARA utility function. The function $\varphi(x,\tau)$ is increasing in the $x$ variable and decreasing in the $\tau=T-t$ variable. Therefore, the optimal vector $\bm{\theta}(x,\tau)$ contains more diversified portfolio of assets when $x$ increases and the time $t\to T$ (see Figure~\ref{fig:vysledky-DAX}).  Furthermore, it is reasonable to invest in an asset with the highest expected return when the account value $x$ is low, whereas an investor has to diversify the portfolio when $x$ is large and time $t$ is approaching terminal maturity $T$.

\begin{figure}
\centering
\includegraphics[width=0.5\textwidth]{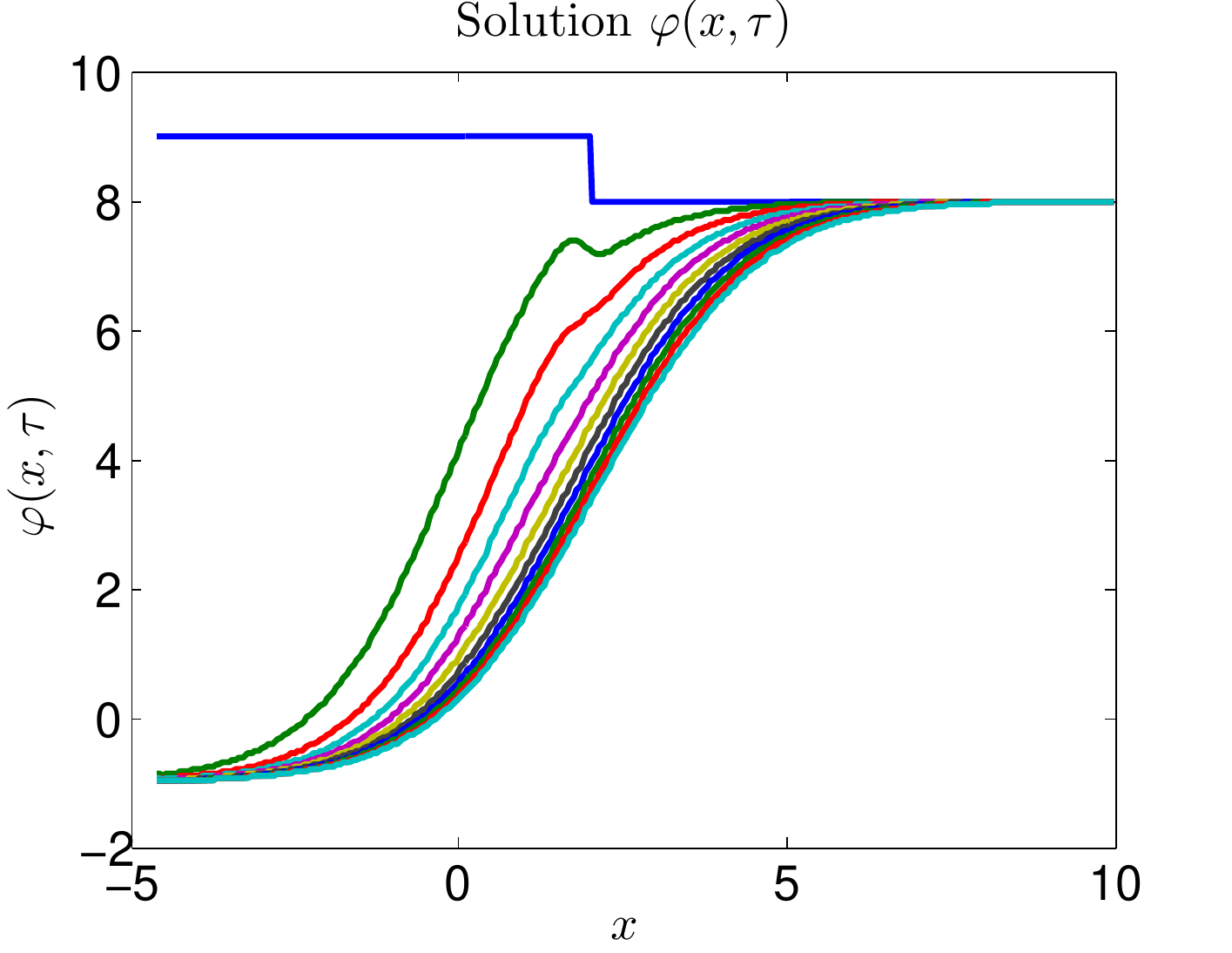}     

    \caption{A solution $\varphi(x,\tau)$ for the DARA utility function with $a_0=9$, $a_1=8$, $x^\ast=2$. Source: our computations based on the numerical method from \cite{udeani2021application}.}
    \label{fig:vysledky1}
\end{figure}

\section{Conclusions}
This review paper presents the analysis of solution to nonlinear and nonlocal partial integro-differential equation arising from financial market. Specifically, we studied and analyzed linear and nonlinear PIDEs arising from option pricing and portfolio selection problem. For the option pricing, we investigated the systematic relationships between the corresponding PIDEs and Black--Scholes models that study market's illiquidity when the underlying asset price follows a L\'evy stochastic process with jumps. We employ the theory of abstract semilinear parabolic equation in the Bessel potential spaces to establish the qualitative properties of solutions to nonlocal linear and nonlinear PIDE for a general class of the so-called admissible L\'evy measures satisfying suitable growth conditions at infinity and origin are also established in the multidimensional space. We considered a general shift function arising from nonlinear option pricing models, which takes into account a large trader stock-trading strategy with the underlying asset price following the L\'evy process. Then, for the portfolio allocation problem, we presented the qualitative properties results to the fully nonlinear HBJ equation arising from stochastic dynamic optimization problem in Sobolev spaces using the theory of monotone operator technique. Such HJB equation, presented in abstract setting, arises from portfolio management problem, where the goal of an investor is to maximize the condition expected terminal utility of a portfolio. We also presented a stable, convergent, and consistent numerical scheme that approximates solution of such PIDE. Several numerical simulations are conducted to demonstrate the influence of a large trader and intensity of jumps on the option price.

\label{lastpage-01}

\end{document}